\documentclass[a4paper,10pt,leqno]{amsart}

\usepackage{a4wide,color,array}
\usepackage{cite}
\usepackage{hyperref}
\usepackage{amsmath,amssymb,amsthm}
\hypersetup{linktocpage,colorlinks,citecolor={magenta}}
\usepackage[english]{babel}
\usepackage[utf8]{inputenc}
\usepackage{nicefrac}
\usepackage{esint}
\usepackage{graphicx}
\usepackage{lmodern}
\usepackage{enumitem}

\numberwithin{equation}{section}
\setlist{noitemsep}

\newtheorem{theo}{Theorem}
\newtheorem{prop}{Proposition}[section]
\newtheorem{lem}[prop]{Lemma}
\newtheorem{coro}[prop]{Corollary}

\newcommand{\R}{\mathbb{R}}
\newcommand{\dd}{\mathrm{d}}
\def\Esp{\mathbb{E}}
\def\hal{\frac{1}{2}}
\def\div{\mathrm{div} \,}
\def\mc{\mathcal}
\def\dist{\mathrm{dist}}
\def\nab{\nabla}
\def\indic{\mathbf{1}}

\def\Xint#1{\mathchoice
  {\XXint\displaystyle\textstyle{#1}}%
  {\XXint\textstyle\scriptstyle{#1}}%
  {\XXint\scriptstyle\scriptscriptstyle{#1}}%
  {\XXint\scriptscriptstyle\scriptscriptstyle{#1}}%
  \!\int}
\def\XXint#1#2#3{{\setbox0=\hbox{$#1{#2#3}{\int}$}
  \vcenter{\hbox{$#2#3$}}\kern-.5\wd0}}
\def\dashint{\Xint-}

\def\XN{X_N}
\def\YN{Y_N}
\def\ZNbeta{Z_{N,\beta}}
\def\Z{\mathsf{Z}_{N,\beta}^{\lambda}}
\def\P{\mathbb{P}_{N,\beta}^{\lambda}}
\def\Fluct{\mathrm{Fluct}}
\def\g{\mathsf{g}}
\def\F{\mathsf{F}}
\def\K{\mathsf{K}}
\def\rr{\mathsf{r}}
\newcommand{\dip}{\mathrm{dip}}
\def\pair{\mathrm{pair}}
\def\nn{\mathrm{nn}}

\begin{document}
\title[Dipoles for the 2CP]{Dipole Formation in the Two-component Plasma}
\author{Jeanne Boursier}
\address{Department of Mathematics, Massachusetts Institute of Technology, 77 Massachusetts Avenue, Cambridge, MA 02139.}
\email{boursier@mit.edu}
\author{Sylvia Serfaty}
\address{Courant Institute of Mathematical Sciences, 251 Mercer Street, New York, NY 10012.\\
\& Sorbonne Universit\'e, CNRS, Laboratoire Jacques-Louis Lions.}
\thanks{S.S. is supported by NSF grant DMS-2247846 and by the Simons Foundation through the Simons Investigator program.}
\email{serfaty@cims.nyu.edu}
\begin{abstract}
  We consider the two-dimensional two-component plasma, or Coulomb gas, consisting of $N$ positive and $N$ negative charges with logarithmic interaction. We introduce a suitable regularization of the interaction by smearing the charges over a small length scale $\lambda$, which allows us to give meaning to the system in the continuum at any temperature. We provide an expansion of the free energy for each inverse temperature $\beta\ge2$, as $N \to \infty$ and $\lambda \to 0$. Doing so allows us to show that, for $\beta \geq 2$, the charges (for the most part) pair into neutral dipoles of very small size as $\lambda \to 0$. This complements the prior work of Leblé, Zeitouni, and the second author, which proved that this does not happen for $\beta < 2$, thereby implying a transition at $\beta = 2$ from free charges to dipole pairs. Moreover, we obtain an estimate on the size of linear statistics.

  The description in terms of dipoles is made via a decomposition into nearest-neighbor graphs of the point configurations, \`a la Gunson--Panta. This is combined with new energy estimates obtained via an electric reformulation of the interaction energy and a ball-growth method, which are expressed in terms of the nearest-neighbor graph distances only. In this way, the model is compared to a reduced nearest-neighbor interaction model, by showing the relative smallness of the dipole-dipole interactions.

\end{abstract}
\maketitle

\section{Introduction}

\subsection{Setting of the problem}
In the 1970s, Kosterlitz and Thouless \cite{Kosterlitz1974,Kosterlitz1973} and independently Berezinsky \cite{Berezinsky1970fr} predicted a completely new type of phase transition without long-range order in two-dimensional systems, now called the Berezinsky-Kosterlitz-Thouless (BKT) transition. This celebrated transition (see \cite{bietenholz2016berezinskii} for a review) was predicted to happen in a whole range of models which exhibit quantized vortices in a neutral ensemble, more specifically the XY or ``rotator" spin model, models of dislocations and superfluids, and it has important consequences for condensed matter physics.

The transition in the XY model is probably the one that has attracted the most attention in the mathematical physics community.
In this model unit spins are sampled on a lattice, constituting a $\mathbb{U}(1)$ analogue of the Ising model. The correlation function between distant spins decays exponentially above the BKT transition temperature and as a power law below \cite{mcbryanspencer,frohlichspencer,mkp1,mkp2,bricmont}.

This transition is explained by the formation of topological vortices, which are points around which the spin field has a nonzero degree or winding number. Below the transition temperature, vortices are bound into dipole pairs (i.e.,~pairs of vortices with opposite winding numbers), while above the transition temperature, vortices are like ``free particles".

In the original papers, as well as in subsequent research, it is expected that in the XY model (or its simplified variant, the Villain model) the energy of the system can be split into a vortex-gas energy and a spin-wave contribution, corresponding to the fluctuations around the vortex configurations \cite{Kosterlitz1974,Kosterlitz1973,kennedy}.
This statement turns out to be delicate to prove rigorously, and this has attracted the attention of researchers, even recently \cite{garban2020statistical,garban2020quantitative}.

Once the spin-wave contribution can be separated, the model reduces to a (2D) gas of positive and negative charges with logarithmic interaction, which can also be called a two-component plasma, or (two-component) Coulomb gas. The Coulomb gas is thus a fundamental model on which to understand the BKT transition, as seen in the original paper of Kosterlitz \cite{Kosterlitz1973}.

The lattice (two-component) Coulomb gas was studied in the seminal work of Fr\"ohlich-Spencer \cite{frohlichspencer} via the sine-Gordon representation and expansions into multipole ensembles, allowing them to analyze the decay rate of correlation functions, thus giving the first proof of the BKT transition (see also \cite{mcbryanspencer,mkp1,mkp2}).

The Coulomb gas may as well be studied in the continuum rather than on a lattice, and is expected to exhibit the same transition between a situation with free vortices and a situation with vortices of opposite sign bound in dipole pairs. There is a subtlety, however, due to the fact that this ``dipole transition" should happen at inverse temperature $\beta=2$ in the units we use, while the KT transition between exponentially and algebraically decaying correlations is expected to happen at $\beta=4$ in this setting. Also, it is a little delicate to directly compare the situation of the Coulomb gas in the continuum where one takes a fixed number $N$ of charges of each sign, corresponding to a canonical ensemble, and the situation of the XY model, corresponding to a grand-canonical ensemble where the number of vortices is not prescribed.

Here we will focus on the continuum Coulomb gas or ``two-component plasma" in the canonical case and we will provide a proof of the ``dipole transition" based simply on the analysis of dipole pairs via large deviations techniques that allow us to weigh their energy and entropy costs, in some sense very close to the arguments and computations found in the original papers \cite{Kosterlitz1973,Kosterlitz1974} and also in the seminal paper \cite{GunPan}.

It was predicted in \cite{frohlichspencer,BenfattoGallavottiNicolo,GallavottiNicolo} that such a system also has an infinite sequence of transitions at
\begin{equation*}\beta_p := 4 - \frac{2}{p},\quad \text{where $p\in \mathbb{N}^*$}, \end{equation*}
which corresponds to the divergence (at small scale) of the free energy of multipoles of size $2p$ (i.e.~sets of $p$ dipoles clustered together).
This was mostly analyzed by sine-Gordon transformation, using field theory techniques and Mayer expansions \cite{frohlichspencer,BenfattoGallavottiNicolo,GallavottiNicolo,dimockhurd}. Analyzing these multipole transitions, and obtaining refined asymptotic expansions via cluster expansions, is the object of our subsequent paper \cite{BS2}.

\subsection{Model}
We first consider the continuum Coulomb gas, defined as an ensemble with configurations $(\XN, \YN)$ (with $X_N=(x_1, \dots, x_N) \in \Lambda^N $ and $Y_N=(y_1, \dots, y_N)\in \Lambda^N$) of $N$ positive and $N$ negative particles (or vortices with degrees $+1$ or $-1$) in the blown-up cube $\Lambda = [0,\sqrt{N}]^2$ of $\R^2$, having energy
\begin{equation}\label{1.1} \F(\XN, \YN)= \hal \left( \sum_{i\neq j} \bigl(- \log |x_i-x_j|-\log |y_i-y_j|\bigr)+ 2 \sum_{i,j} \log |x_i-y_j|\right),\end{equation} and
consider the (canonical) ensemble
\begin{equation}\label{ZGP} \frac{1}{Z_{N,\beta}} \exp\left(- \beta \F(X_N, Y_N) \right) \dd X_N \dd Y_N ,\end{equation}
with $\dd X_N$ and $\dd Y_N$ the Lebesgue measures on $\Lambda^N$.
This model was studied in particular in \cite{GunPan,DeutschLavaud}, and more recently in \cite{LSZ}.
The integral defining $\ZNbeta$ diverges as soon as $\beta \ge 2$, since the Boltzmann weight of very short dipoles is not integrable, which corresponds to the dipole transition.
The ensemble \eqref{ZGP} was thus studied only in the regime $\beta <2$ in the aforementioned works \cite{GunPan,DeutschLavaud,LSZ}.
The latest results of \cite{LSZ}, building on important insights from \cite{GunPan} and techniques developed for the study of the one-component Coulomb gas in \cite{SS2d,RougSer,LS1,LS2}, show an expansion of $\log \ZNbeta$ as $N \to \infty$, as well as a large deviations principle on point processes, which characterize a situation with free interacting particles, with competition between the attraction of opposite charges and the entropic repulsion. This corresponds to temperatures above the critical temperature.

In order to study such a system for $\beta \ge 2$, a truncation of the interaction is needed, as already recognized in \cite{Kosterlitz1974,Kosterlitz1973} and analyzed in \cite{frohlichspencer} (see also the discussion in \cite{LRVboundary}).
For brevity, let us always denote
\begin{equation*}
  \g(x)=-\log |x|,
\end{equation*}
and we will abuse notation by considering $\g$ as a function on $\R^2$ or on $\R$ depending on the context.

Truncating the interaction involves introducing a small length scale $\lambda$ and \emph{renormalizing} the divergent part of the free energy as $\lambda\to 0$. A natural proposed way is to truncate the energy at a distance $\lambda$ and consider
\begin{equation}\label{eq:F2}
  \hal \sum_{i,j}\bigl(\g(x_i-x_j)\wedge\g(\lambda)+\g(y_i-y_j)\wedge\g(\lambda)-2\bigl(\g(x_i-y_j)\wedge\g(\lambda)\bigr)\bigr)-N\g(\lambda),
\end{equation}
where $\wedge $ denotes the minimum of two numbers.

We use a variant convenient for our techniques: instead of truncating $\g$ we consider charges smeared on discs of radius $\lambda$, with $\lambda$ small, interacting otherwise in the normal Coulomb fashion: letting $\delta_{z}^{(\lambda)}$ denote the uniform measure of mass $1$ supported on $ B(z, \lambda)$, we let
\begin{equation*}
  \kappa:= \iint \g(x-y) \delta_0^{(1)}(x)\delta_0^{(1)}(y),\end{equation*}
and observe, by scaling, that
\begin{equation}\label{geta0}\iint \g(x-y) \delta_0^{(\lambda)} (x) \delta_0^{(\lambda)}(y)= \g(\lambda)+ \kappa.\end{equation}

We then consider the energy
\begin{equation}\label{eq:F3}
  \F_\lambda(\XN, \YN)= \hal \iint \g(x-y) \dd\left( \sum_{i=1}^N \bigl(\delta_{x_i}^{(\lambda)}-\delta_{y_i}^{(\lambda)}\bigr)\right)(x)
  \dd\left( \sum_{i=1}^N \bigl(\delta_{x_i}^{(\lambda)}-\delta_{y_i}^{(\lambda)}\bigr)\right)(y) - N ( \g( \lambda)+\kappa).
\end{equation}
Here, compared to \eqref{1.1} we have reinserted the self-interaction terms which are no longer infinite but equal to $\g(\lambda) + \kappa$,
and then subtracted them off.

We will denote by
\begin{equation} \label{eq:defgeta} \g_\lambda(z)= \iint \g(x-y) \delta_0^{(\lambda)}(x)\delta_z^{(\lambda)} (y),
\end{equation} the effective interaction between two points at distance $|z|$. Since the convolution $\g* \delta_0^{(\lambda)}$ is harmonic outside of $B(0, \lambda)$, it follows from the mean-value theorem (or Newton's theorem) that
\begin{equation*} \g_\lambda(z)= \int \g* \delta^{(\lambda)}_0 \delta_z^{(\lambda)} = \g(z) \quad \text{for } |z|\ge 2\lambda,\end{equation*} so $\g_\lambda$ coincides with $\g$ at large enough distance.

Thus $\F_\lambda$ is obtained from \eqref{eq:F2} by replacing the truncated kernel by $\g_\lambda$, including in the self-interaction term, and subtracting $N\g_\lambda(0)=N(\g(\lambda)+\kappa)$ instead of $N\g(\lambda)$. If all pairwise distances are at least $2\lambda$, then $\F_\lambda=\F$. In particular, $\F_\lambda\to\F$ as $\lambda\to0$ for configurations of distinct points. The kernel $\g_1$ and the constant $\kappa$ depend on the smearing profile, as does $\mathcal Z_\beta$ below when $\beta>2$.
Finally, we could in principle use charges smeared on circles instead of discs as in previous works \cite{petrache2017next,LS1,LS2}. It does make the initial computations simpler but the potential generated by a circle is too singular for our needs here.
\smallskip

We thus work with \eqref{eq:F3} and study
\begin{equation} \label{defP}
  \dd\P= \frac{1}{\Z} \exp\left(- \beta \F_\lambda(\XN, \YN)\right) \dd X_N \dd Y_N
\end{equation}
in the limit where $\lambda$ tends to zero, where
\begin{equation*} \Z= \int_{\Lambda^{N}\times \Lambda^{N}} \exp\left(-\beta \F_\lambda(\XN, \YN)\right) \dd \XN \dd\YN.
\end{equation*}
When $\beta <2$ one can set $\lambda=0$ and recover the model studied in \cite{LSZ}, but when $\beta\ge2$ one expects $\log \Z$ to diverge as $\lambda \to 0$.
The picture that emerges from the literature, mostly based on the sine-Gordon representation,
is well-described in \cite{LRVboundary}:
for $\beta>2$, the divergence of the system as $\lambda\to 0$ corresponds to the pairing of short dipoles (of length scale $\lambda$), and the transition at $\beta=2$ is followed as $\beta$ increases by a sequence of transitions corresponding to the formation of a subdominant proportion of multipoles (quadrupoles, sextupoles, etc.) as the temperature is decreased and the entropic repulsion becomes less strong \cite{frohlichspencer}. This is due to the weakly attractive nature of the dipole-dipole interaction. When $\beta $ reaches $4$, in the grand canonical setting (when the number of particles is not fixed) the system is expected to collapse under the attraction of the dipoles, as first shown in \cite{Frohlich} via a Euclidean Field Theory approach. However, we will see that it is not the case in the canonical setting here.

Our main goal here is to analyze \eqref{eq:F3}--\eqref{defP} via a simple and direct approach based solely on energy and entropy, precisely via electrostatic estimates for the energy. We obtain below a precise free energy expansion as $N \to \infty$ and $\lambda \to 0$, and use it to prove that configurations mostly form isolated neutral dipoles for all $\beta \in [2,+\infty)$ (this shows that the multipoles, though present, concern only a vanishing fraction of the particles). Combined with the description of \cite{LSZ}, it constitutes a proof of the dipole transition, and we hope this point of view will also inform the understanding of the BKT transition. We also address the important question of the fluctuations of the (two-component) Coulomb gas.

Note that the two-component Coulomb gas or plasma is quite different from the one-component Coulomb gas or plasma, which consists only of positively charged particles repelling each other and confined by an external potential, or, in the jellium case, by a uniform negative background charge. The one-component Coulomb gas never diverges nor forms dipole pairs, but rather the particle density converges macroscopically to an equilibrium measure limit (dictated by the external potential) while at the microscopic scale the particles arrange themselves in more and more ordered point patterns as temperature decreases. In fact the system is expected to crystallize at zero temperature, at least in low dimensions. There has been much progress on the one-component plasma in recent years, including free energy expansion \cite{Lieb1997,LS1,Armstrong2019LocalLA,serfaty2020gaussian}, local laws for the distribution of points down to the microscale \cite{leble2017local,Armstrong2019LocalLA}, variational characterization of the limiting point processes \cite{LS1}, and CLTs for the fluctuations of linear statistics \cite{ameur2011fluctuations,rider2007noise,bauerschmidt2016two,LS2,serfaty2020gaussian}. We refer to \cite{serfaty2024lectures} for a comprehensive description.

The main points in common with \cite{SS2d,LS1,leble2017local,LS2,Armstrong2019LocalLA,serfaty2020gaussian} but also with \cite{LSZ} will be the general philosophy of electrostatic energy and large deviations techniques, as well as the \emph{electric formulation} of the energy that we present just below.

\subsection{Main results}

The electric formulation mentioned above consists in reexpressing the energy in terms of the \emph{electric potential} $h_\lambda$ generated by the configuration $(X_N, Y_N)$ and defined as a function \emph{over all} $\R^2$ by
\begin{equation}\label{eq:heta}h_\lambda[X_N, Y_N]:= \g* \left( \sum_{i=1}^N \bigl(\delta_{x_i}^{(\lambda)}-\delta_{y_i}^{(\lambda)}\bigr)\right),\end{equation} where $*$ denotes the convolution.
In the sequel, we will most often drop the $[X_N, Y_N]$ dependence in the notation.

Note that by definition of $\g$, $h_\lambda$ satisfies the Poisson equation
\begin{equation}\label{eq:Deltaheta}
  -\Delta h_\lambda[X_N, Y_N]= 2\pi \left( \sum_{i=1}^N \bigl(\delta_{x_i}^{(\lambda)}-\delta_{y_i}^{(\lambda)}\bigr)\right).\end{equation}
A direct insertion into \eqref{eq:F3} and integration by parts using \eqref{eq:Deltaheta} yields the following electric rewriting of the energy
\begin{equation}\label{eq:rewrF}\F_\lambda(\XN, Y_N)= \frac{1}{4\pi} \int_{\R^2} \left|\nab h_\lambda[X_N, Y_N]\right|^2 - N (\g(\lambda)+\kappa).\end{equation}
Neutrality gives $h_\lambda(x)=O(|x|^{-1})$ and $\nabla h_\lambda(x)=O(|x|^{-2})$ at infinity. A smooth cutoff tending to~$1$ justifies the integration by parts.

Before stating our main result, let us introduce some more notation.
For $\beta>2$, let us define the probability measure
\begin{equation*}
  \dd \mu_\beta(r):=\frac{2\pi}{\mathcal Z_\beta}\exp\Bigr(\beta\g_1(r)\Bigr)\indic_{\R^+}(r)r \dd r,
\end{equation*}
where $\g_1$ was defined in \eqref{eq:defgeta} and $\mathcal Z_\beta$ is given by
\begin{equation}\label{def:Cbeta}
  \mathcal Z_{\beta}:=\begin{cases}
    2\pi\int_0^{\infty}\exp\Bigr(\beta\g_1(r)\Bigr)r \dd r & \text{if $\beta>2$}\\
    2\pi & \text{if $\beta=2$}.
  \end{cases}
\end{equation}
For $\beta>2$, we have $0<\mathcal Z_\beta<\infty$, since $\g_1$ is bounded on $[0,2]$ and $e^{\beta\g_1(r)}r=r^{1-\beta}$ for $r\ge2$. This constant is related to our precise way of smearing the Dirac charges and corresponds to the interaction of overlapping disc charges.

We will denote by $(z_1,\dots,z_{2N})=(x_1,\dots,x_N,y_1,\dots,y_N)$ the collection of all points (positive or negative) and set $d_i=1$ if $i\in [1,N]$ and $d_i=-1$ if $i\in [N+1, 2N]$.
We exclude configurations with coinciding points or equal distances between distinct unordered pairs. This is a null set by Fubini. We denote by $\phi_1(i)$ the index of the nearest neighbor to $z_i$.
A generic configuration can be quite complicated. In particular, it is not obvious how to extract its dipoles.
In this paper, dipoles are defined via the nearest neighbor graph of a configuration, following important ideas of \cite{GunPan}. This definition ignores, in the first pass, the sign of the particle, by considering only for each point $z_i$, its first nearest neighbor $z_{\phi_1(i)}$ of arbitrary sign, and its scaled, truncated nearest neighbor distance $\rr_1(z_i)$ defined in \eqref{nn2}.
Other definitions are possible, and we will pursue another approach in forthcoming work.

The nearest-neighbor graph is a (directed) graph on $\{1, \dots, 2N\}$ with an edge from $i$ to $j=\phi_1(i)$, i.e.~if $z_j$ is the nearest neighbor to $z_i$. It is a general fact (see \cite{eppstein1997}) that the connected components of such nearest-neighbor graphs are all formed of a 2-cycle with trees possibly attached to them. Note that points in a 2-cycle are characterized by the fact that $\phi_1\circ \phi_1(i)=i$.

\begin{figure}[htbp]
  \centering
  \includegraphics[width=0.6\textwidth]{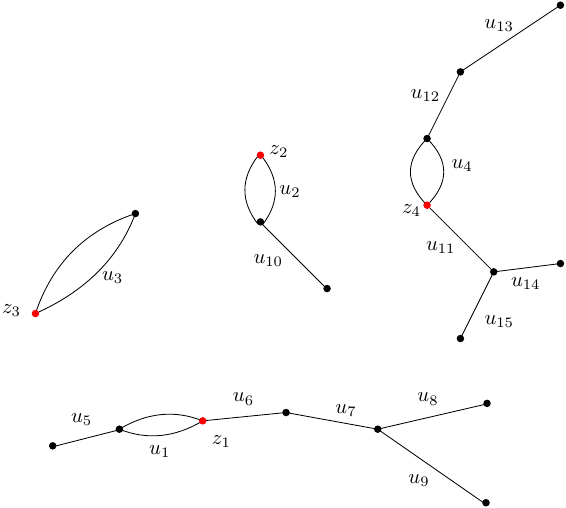}
  \caption{A nearest-neighbor graph with Gunson--Panta variables}
\end{figure}

We will say that a 2-cycle is isolated if it has no tree attached to it, i.e.~forms a whole connected component of the graph. We will say that such a 2-cycle is an isolated neutral 2-cycle if in addition it connects two points of opposite sign, i.e.,~it connects $i$ and $j$ with $d_i d_j=-1$. We also call such 2-cycles ``isolated dipoles." The length of the dipole $\{i,j\}$ is $|z_i-z_j|$.

Our first theorem provides a free energy expansion and a concentration on dipole configurations.
Here, and throughout the paper, $C_\beta, c_\beta$ denote positive constants depending only on $\beta$. We take $M_0=36$.

\begin{theo}\label{theorem:th1}
  For $\beta\ge2$ and $0<\lambda<e^{-1}$, let
  \begin{equation}\label{eq:defgammal}
    \omega_\lambda:=\begin{cases}
      |\log \lambda|^{-\frac{1}{5}} & \text{if $\beta=2$}\\
      \lambda^{\frac{2(\beta-2)}{12-\beta+2M_0}} &\text{if $\beta\in (2,4)$}\\
      (\lambda|\log \lambda|)^{\frac{2}{6+M_0}} & \text{if $\beta=4$}\\
      \lambda^{\frac{2}{4+M_0}} & \text{if $\beta>4$},
    \end{cases}
  \end{equation} and observe that $\omega_\lambda = o_\lambda(1)$ as $\lambda \to 0$.
  For each $\beta\ge2$, there exists $\lambda_\beta\in(0,e^{-1})$ such that the following bounds hold uniformly for $N\ge2$ and $0<\lambda<\lambda_\beta$.
  \begin{enumerate}
    \item We have $0<\Z<\infty$, and the free energy satisfies the following expansion. If $\beta>2$, then
    \begin{subequations}\label{eq:devlogz}
    \begin{equation}
      \log \Z=2N\log N+N\bigl((2-\beta)\log\lambda+\log\mathcal Z_\beta-1\bigr)+O_\beta(N\omega_\lambda+\log(N+1)),
    \end{equation}
    and if $\beta=2$, then
    \begin{equation}
      \log \Z=2N\log N+N\bigl(\log(2\pi|\log\lambda|)-1\bigr)+O(N\omega_\lambda+\log(N+1)).
    \end{equation}
    \end{subequations}
    \item The energy essentially reduces to the nearest neighbor interactions of dipoles (or neutral 2-cycles) in the sense that
    \begin{equation}\label{eq:momexp}
      \log \Esp_{\P}\Bigr[ \exp\Bigr(\beta\Bigr(\F_\lambda-\tilde{\F}_\lambda\Bigr)\Bigr) \Bigr]\leq C_\beta N\omega_\lambda+C_\beta\log(N+1),
    \end{equation}
    where
    \begin{equation*}
      \tilde{\F}_\lambda=-\frac{1}{2}\sum_{i=1}^{2N}\g_\lambda(z_i-z_{\phi_1(i)})\indic_{\phi_1\circ\phi_1(i)=i,d_id_{\phi_1(i)}=-1}.
    \end{equation*}
    \item The Gibbs measure is concentrated on configurations made mostly of isolated dipoles, with the following bounds on their number and lengths:\\
    $\bullet$ For $p_0\in\{0,\ldots,N\}$, letting $\mc{A}(p_0)$ be the event on which there are exactly $p_0$ isolated dipoles, we have
    \begin{equation}\label{eq:b1}
      \log \P(\mc{A}(p_0))\leq -c_\beta N \Bigr(1-\frac{p_0}{N}\Bigr)^2+O_\beta(N\omega_\lambda)+C_\beta\log(N+1).
    \end{equation}
    $\bullet$ For $n\in\{0,\ldots,N\}$ and $R>0$, let $\mc{B}(R,n)$ be the event on which there are exactly $n$ isolated dipoles of length larger than $R$. For $\beta>2$ and $R\ge2\lambda$, let $x_0:=\frac{n}{N}$ and $\alpha:=(\frac{\lambda}{R})^{\beta-2}$. There exist $C_\beta,c_\beta>0$ such that if $x_0\geq C_\beta\alpha$, then
    \begin{equation}\label{eq:b2}
      \log \P(\mc{B}(R,n))\leq -c_\beta N x_0\log \Bigr(\frac{x_0}{\alpha}\Bigr)+O_\beta(N\omega_\lambda)+C_\beta\log(N+1).
    \end{equation}
    $\bullet$ For $\beta=2$, there exists $C>0$ such that, for all integers $N\ge2$ and $0\le n\le N$ and all $\lambda,R$ with $0<\lambda<\lambda_2$ and $2\lambda\le R\le1$,
    \begin{equation}\label{eq:critical lengths}
      \P\bigl(\mc B(R,n)\bigr)
      \le (N+1)^C\exp\left[-n+CN\left(\omega_\lambda+
      \frac{\log(1/R)}{|\log\lambda|}\right)\right].
    \end{equation}
  \end{enumerate}
\end{theo}
The error rate $\omega_\lambda$ in the present proof is not sharp and could be optimized. However, we do not pursue this goal here. A sharper expansion is obtained in \cite{BS2}.

\medskip

The formula \eqref{eq:devlogz} can be compared with the formula for $\beta<2$ obtained in \cite{LSZ}\footnote{Note that the scaling used in \cite{LSZ} is different in the sense that the particles are in a box of size $1$.} and this already exhibits a transition at $\beta=2$, since the divergence in $\lambda$ is present only for $\beta \ge 2$.
When the thermodynamic limit exists, we write
\begin{equation}\label{thermolim}
  f(\beta,\lambda)=\lim_{N\to\infty}\frac1N\bigl(\log\Z-2N\log N\bigr).
\end{equation}
For $\beta<2$, \eqref{thermolim} is proven in \cite{LSZ} (with $\lambda=0$) and a variational characterization of $f$ is also provided there:
\begin{equation*} f(\beta, 0)= -\inf_{P\in\mathcal P} \mathcal F_\beta(P),\end{equation*}
where $\mathcal P $ corresponds to the space of stationary signed point processes of intensity $1$ (each species has intensity 1), and
$\mathcal F_\beta$ is the free-energy functional appearing there. It is the sum of $\beta$ times a suitable ``renormalized energy" of point processes (an infinite volume Coulomb interaction energy) and a specific relative entropy with respect to two independent Poisson point processes of intensity~$1$, one for each sign. We see that in that regime we have free particles of positive and negative charges, whose positions are governed by the minimization of $\mathcal F_\beta$.

Whenever \eqref{thermolim} exists, \eqref{eq:devlogz} gives, for $\beta\ge2$,
\begin{equation*}
  f(\beta,\lambda)=
  \begin{cases}
    (2-\beta)\log\lambda+\log\mathcal Z_\beta-1+o_\lambda(1) & \text{if $\beta>2$},\\
    \log(2\pi|\log\lambda|)-1+o_\lambda(1) & \text{if $\beta=2$}.
  \end{cases}
\end{equation*}
Here the energy is dominated by pure dipole energy, demonstrating the transition from free particles to bound pairs.
The question of the sharp rate of convergence $o_\lambda(1)$ is very important as it encodes the multipole transitions.

We finally address the important question of the fluctuations of linear statistics. We provide an energetic control, similar in spirit to \cite[Prop. 2.5]{LS2}, giving a fluctuation bound at scale $\sqrt N\,\omega_\lambda^{1/4}$ for Lipschitz test-functions. The tail below tends to zero when $\sqrt N\,\omega_\lambda\to\infty$. One could obtain a better rate of fluctuations in $o(N^{1/2})$ at fixed $\lambda$, but this would require more involved arguments.

\begin{theo}\label{prop:linear}
  For $\beta\ge2$, there exist $\lambda_\beta\in(0,e^{-1})$, possibly smaller than in Theorem~\ref{theorem:th1}, and $M_\beta,C_\beta>0$ with the following property. Consider a compactly supported Lipschitz test-function $\xi_0 :\R^2 \to \R$ and set $\xi=\xi_0(N^{-\nicefrac{1}{2}}\cdot)$. Define \[\Fluct_N(\xi):= \int_{\Lambda} \xi \, \dd \Big(\sum_{i=1}^N \bigl(\delta_{x_i}-\delta_{y_i}\bigr)\Big).\]
  For every $M>M_\beta$, there is $c_{\beta,M}>0$, independent of $\xi_0$, such that for $N\ge2$ and $0<\lambda<\lambda_\beta$,
  \begin{equation*}
    \P\left( \Fluct_N(\xi)^2 > M N \omega_\lambda^{\nicefrac{1}{2}}\Vert\nabla\xi_0\Vert_{L^{\infty}}^2 \right)\leq C_\beta e^{-c_{\beta,M} N^{\nicefrac{1}{2}}\omega_\lambda},
  \end{equation*}
  where $\omega_\lambda$ is as in \eqref{eq:defgammal}.
\end{theo}
Here $\|\nabla\xi_0\|_{L^\infty}$ is the essential supremum of the a.e. gradient, equal to the Lipschitz constant of $\xi_0$.

  {\hypersetup{linkcolor=blue}\renewcommand{\thefootnote}{{\arabic{footnote}}}Theorems~\ref{theorem:th1} and~\ref{prop:linear} have been formalized in Lean~4.\footnote{{The formalization is available at \hypersetup{urlcolor=blue}\url{https://github.com/serfatysylvia/two-component-coulomb-dipoles-dev}.}}}

\subsection{Method}
In order to prove Theorem \ref{theorem:th1}
we need to obtain sharp upper and lower bounds on the energy of a configuration in terms of its dipoles, and good corresponding ``volume estimates".
This amounts to comparing the original model to a reduced nearest neighbor model
and the main work is to prove that the two models are close. Their closeness then allows us to show that the Gibbs measure concentrates on dipole configurations, control their sizes and control linear statistics.

Combining the description in terms of the nearest neighbor graph \`a la Gunson--Panta with the electric formulation \eqref{eq:rewrF} turns out to be an efficient way of obtaining energy lower bounds. This is done by a \emph{ball-growth method}, which consists in expanding the smeared charges $\delta_{z_i}^{(\lambda)} $ into a charge $\delta_{z_i}^{(\rr_1(z_i))}$ of the same mass but supported in the disc $ B(z_i, \rr_1(z_i))$. Every disc whose radius grows stays disjoint from the other discs, so Newton's theorem preserves the pairwise interactions. This, however, misses a term of order 1 in the interaction energy of each dipole.
In order to avoid this loss, we push the method further by examining second nearest neighbor distances $\rr_2(z_i)$ and grow each pair to a common radius determined by these distances.
In the works on the one-component plasma \cite{LS2,Armstrong2019LocalLA,serfaty2020gaussian}, see also \cite[Chap. 4]{serfaty2024lectures}, we were able to take advantage of the fact that when all the charges are positive, the interaction of disc charges decreases when the radii are increased. This is no longer the case in a situation with different charge signs, and so instead of using monotonicity, we proceed to a direct estimation of the change in the interaction when the discs are increased. We then obtain an estimate which bounds from below the energy \emph{in terms of nearest neighbor interactions only}.
This part contains the most delicate estimates as we need to control the contributions of the dipole-dipole interactions and the errors they create, something not handled in \cite{GunPan,LSZ}. We provide more precise estimates than in those papers via new ideas.

An important feature of this dipole decomposition lower bound is that it is amenable to integration in phase-space with the method of Gunson--Panta, as revisited in \cite{LSZ}, of separating the integral over types of nearest neighbor graphs.
For the matching lower bound on the free energy, we bound the energy from above by a sum of local Neumann energies, one for each dipole.
Together with the constrained upper bounds, this yields concentration on dipole configurations, as in \eqref{eq:b1}.

The bound on fluctuations of Theorem \ref{prop:linear} is obtained by using the electric formulation of the energy which allows us to show that the energy directly controls the fluctuations, and combining it with the ball-growth method which is this time used by inflating the balls to an intermediate length scale $\theta$ much larger than $\lambda$ and much smaller than $1$.

\medskip

\textbf{Plan of the paper:} Section \ref{sec2} is devoted to the proof of the energy lower bound in terms of nearest neighbor interactions with dipolar interaction errors, via the ball-growth method. In Section \ref{sec3}, this lower bound is inserted into the Gibbs measure to produce, via a suitable decomposition of the phase-space and large deviation estimates, the free energy upper bound.
Section \ref{sec4} provides a matching lower bound by explicit construction of configurations and estimates of their free energy. Finally, Section \ref{sec5} provides the proof of Theorem \ref{prop:linear}.

\subsection{Notation}\label{subsec:notation}
\begin{itemize}
  \item Unless stated otherwise, $C_\beta,c_\beta>0$ depend only on $\beta$ and are uniform in all other parameters. They may vary from line to line.
  \item The notation $O_\beta(a)$ denotes a quantity bounded in absolute value by $C_\beta|a|$, uniformly in the stated parameter ranges.
  \item We denote by $a\vee b$ the number $\max(a,b)$ and $a \wedge b$ the number $\min (a,b)$.
  \item We write $\ell:=|\log\lambda|$. For $0<\lambda\le e^{-1}$, we have $\ell\ge1$.
  \item For $\beta>2$ and $0<\lambda\le e^{-1}$, we write
  \[
    \varepsilon_\lambda:=\begin{cases}
      \lambda^{\beta-2}&\text{if $2<\beta<4$},\\
      \lambda^2\ell^2&\text{if $\beta=4$},\\
      \lambda^2\ell&\text{if $\beta>4$}.
    \end{cases}
  \]
  It measures the error per dipole in the free energy bounds for $\beta>2$.
  \item Empty products equal $1$. We use $0\log0=0$ and $\max\varnothing=-\infty$.
  \item The empty configuration has measure $1$, energy zero and graph $\varnothing$. We set $D_{0,0}=\{\varnothing\}$.
\end{itemize}

\subsection*{Acknowledgments}
We thank Scott Armstrong for providing the skill files used for the Lean formalization. The first version of this paper was written entirely without AI assistance. During the last round of revisions, we used GPT-6 Astra and Claude Opus 5.5 to formalize the paper in Lean, to make corrections and to propose simplifications of the proofs.

\section{Nearest neighbor lower bound}\label{sec2}
Throughout this section, we assume $N\ge2$ and work outside the null set excluded above.

\subsection{Definitions}
\subsubsection*{Signed point configurations} With the shortcut $Z_{2N}$ for $(X_N, Y_N)$ and $d_i=\pm 1$ for their signs,
we are able to rewrite \eqref{eq:heta} as
\begin{equation*}
  h_\lambda=\g* \left( \sum_{i=1}^{2N} d_i \delta_{z_i}^{(\lambda)}\right),
\end{equation*}
with
\begin{equation*} -\Delta h_\lambda= 2\pi \sum_{i=1}^{2N} d_i \delta_{z_i}^{(\lambda)}.\end{equation*}

When increasing the discs we will also denote similarly
for any vector $\vec{\alpha}=(\alpha_1, \dots , \alpha_{2N})$ in $(0,\infty)^{2N}$
\begin{equation}\label{defhalpha} h_{\vec{\alpha}} = \g* \left( \sum_{i=1}^{2N} d_i \delta_{z_i}^{(\alpha_i)} \right).\end{equation}

\subsubsection*{Successive nearest neighbor distances}
First we set
\begin{equation}\label{nn2}\rr_1(z_i):= \left(\frac14 \min_{j\neq i} |z_i-z_j| \right)\vee \lambda,
\end{equation} then for each $2\le p\le2N-1$,
\begin{equation}\label{nn3} \rr_p(z_i) := \left(\frac14 \min_{j\notin\{i, \phi_1(i) ,\dots, \phi_{p-1}(i)\}} |z_i-z_j|\right)\vee \lambda,\end{equation}
where $\phi_0(i)=i$ and for $1\le k\le2N-1$, $\phi_k(i)$ minimizes $|z_i-z_j|$ over $j\notin\{\phi_0(i),\ldots,\phi_{k-1}(i)\}$. We call $z_{\phi_k(i)}$ the $k$-th nearest neighbor to $z_i$.
We note that we always have
\begin{equation*} \left(\frac 14 |z_i -z_{\phi_k(i)}|\right) \vee \lambda= \rr_k(z_i).\end{equation*}

\subsubsection*{Nearest neighbor graphs}
As discussed in the introduction, the dipole decomposition estimate will be used in conjunction with the method of Gunson--Panta \cite{GunPan} which breaks the configuration into nearest neighbor graphs. It is worth noting, however, that \cite{GunPan} builds the nearest neighbor graphs of all particles, irrespective of their sign, whereas for us the sign will play an important role later.

Recall from the introduction that the nearest-neighbor graph $\gamma_{2N}(Z_{2N})$ of $Z_{2N}$ is the directed graph on $\{1,\ldots,2N\}$ with an edge from $i$ to $\phi_1(i)$ for every $1\le i\le2N$. Connected components mean weakly connected components, obtained by forgetting edge orientations. The graph $\gamma_{2N}(Z_{2N})$ has between $1$ and $N$ connected components and each of its connected components contains exactly one directed $2$-cycle, with possibly empty trees directed toward its two vertices. We denote by $D_{2N,K}$ the set of nearest-neighbor graphs on $\{1,\ldots,2N\}$ with $K$ connected components. We retain the vertex labels.

Let $\gamma\in D_{2N,K}$. Let us denote by $\mathcal{ I}_1,\ldots,\mathcal{ I}_K$ the connected components of $\gamma$ and for each $k\in \{1,\ldots,K\}$, let us label the two vertices of the 2-cycle in $\mathcal{ I}_k$ by $m_k$ and $m_k'$ and call $\mathcal C_k= \{m_k, m_k'\}$ the corresponding 2-cycle.
We also let
\begin{equation}
  \label{defIdip}
  I^{\pair}:=\bigcup_{\mathcal I_k=\mathcal C_k}\mathcal C_k,
  \qquad I^{\dip}:=\bigcup_{\substack{\mathcal I_k=\mathcal C_k\\d_{m_k}d_{m_k'}=-1}}\mathcal C_k.\end{equation}
$I^\pair$ corresponds to points in a 2-cycle which forms a connected component of the graph, i.e., has no tree attached. We call such pairs \emph{isolated}.
Among the pairs we denote by $I^\dip$ the indices corresponding to
isolated (neutral) dipoles.
We note that $2(N-K)$ bounds the number of connected components that are not reduced to a 2-cycle, i.e.,~have trees attached to them.

\begin{lem}[Bounded in-degree of nearest-neighbor graphs]\label{lem:in-degree}
  Let $z_1,\ldots,z_{2N}$ be distinct points of the plane such that distinct unordered pairs of points are at distinct distances. Then each point $z_j$ is the first nearest neighbor of at most six points and the second nearest neighbor of at most twelve points. Moreover, if we choose one vertex from each isolated pair of the nearest-neighbor graph, then at most six chosen vertices have $z_j$ as second nearest neighbor.
\end{lem}

\begin{proof}
  Divide the plane into six half-open cones of angle $\pi/3$ with apex $z_j$. Let $z$ and $w$ be two other points of the configuration in one cone, with $|z-z_j|\le|w-z_j|$. The angle at $z_j$ between them is less than $\pi/3$, so the cosine rule gives
  \[
    |z-w|^2<|z-z_j|^2+|w-z_j|^2-|z-z_j||w-z_j|\le|w-z_j|^2\,.
  \]
  Hence $z$ is closer to $w$ than $z_j$ is. As a consequence, each cone contains at most one point whose first nearest neighbor is $z_j$. It also contains at most two points whose second nearest neighbor is $z_j$, since otherwise the farthest of three such points would have two points closer than $z_j$. This proves the first two bounds.

  Now let $z$ and $w$ be two chosen vertices in one cone, with $|z-z_j|\le|w-z_j|$, and assume that both have $z_j$ as second nearest neighbor. The partner of $w$ in its pair is the first nearest neighbor of $w$, so it is closer to $w$ than $z_j$. The point $z$ is also closer to $w$ than $z_j$, and it is not the partner of $w$, since we chose one vertex from each pair. This contradicts the fact that $z_j$ is the second nearest neighbor of $w$. Hence each cone contains at most one chosen vertex whose second nearest neighbor is $z_j$.
\end{proof}

\subsubsection*{Additional results on $\g_\lambda$}
Let $\psi:=\g*\delta_0^{(1)}$ be the potential of a unit charge smeared on $B(0,1)$. Explicitly,
\[
  \psi(x)=
  \begin{cases}
    (1-|x|^2)/2,& |x|\le1,\\
    -\log|x|,& |x|\ge1.
  \end{cases}
\]
Its gradient is $1$-Lipschitz and $\dashint_{B(0,1)}\nab\psi=0$, where $\dashint$ denotes the average. Hence
\begin{equation}\label{idg1}
  \begin{aligned}
    \g_1(w)-\kappa
    &=\dashint_{B(0,1)}\bigl(\psi(x+w)-\psi(x)-\nab\psi(x)\cdot w\bigr)\,\dd x,\\
    |\g_1(w)-\kappa|&\le\hal|w|^2 \qquad\text{for every }w\in\R^2.
  \end{aligned}
\end{equation}

We have
\begin{equation}\label{approxgeta} \g_\lambda(z)= \g(\lambda)+ \g_1(z/\lambda), \qquad
  |\g_\lambda(z)-\g(|z|\vee \lambda) |\le C,\end{equation}
where $C$ is some universal constant.
The first identity follows by scaling. For the bound, write
\[
  \g_\lambda(z)-\g(|z|\vee\lambda)=H(z/\lambda),\qquad
  H(w):=\g_1(w)+\log(|w|\vee1).
\]
The function $H$ is continuous and vanishes for $|w|\ge2$ by Newton's theorem, hence is bounded.

\subsection{Nearest neighbor expansion of the energy}
For $i\in \mathcal C_k$ for some $k=1, \dots, K$, we let
\begin{equation*} D_i= d_i + d_{\phi_1(i)} .\end{equation*}
If the $2$-cycle is neutral, as in the most typical case, then $D_i=0$.

The following proposition will be proven below.
\begin{prop}[Nearest neighbor decomposition of the energy]\label{prop:mino2}
  Let $(X_N, Y_N)$ be any configuration satisfying the standing assumptions of this section and consider its nearest neighbor graph decomposition as above. We have
  \begin{multline}
    \label{minoF02}
    \F_\lambda(X_N,Y_N) \\
    \ge \hal \sum_{k=1}^K \Bigg( \sum_{i\in \mathcal C_k} \Bigl[d_i d_{\phi_1(i)} \g_\lambda(z_i-z_{\phi_1(i)} ) -d_i D_i\left( \g\left( \rr_2(z_i) \wedge \rr_2(z_{\phi_1(i)} ) \right)+\kappa\right) -C \left( \frac{\rr_1(z_i)}{\rr_2(z_i)}\right)^2\Bigr]
    \\-\sum_{i\in \mathcal{I}_k \backslash \mathcal C_k} \left(\g(\rr_2(z_i)) +\kappa \right)\Bigg),
  \end{multline}
  where $C$ is universal.
\end{prop}
The remainder $C(\frac{\rr_1}{\rr_2})^2$ measures the separation of a $2$-cycle from the remaining charges. For neutral cycles, it controls the interaction error associated with the dipoles.
For nonneutral cycles, combining the signed pair term with the $d_iD_i$ correction gives, relative to the neutral nearest-neighbor energy, at least
\[
  \log\frac{\rr_2(z_i)\wedge\rr_2(z_{\phi_1(i)})}{\rr_1(z_i)}-C
\]
per vertex, as shown below.

We now rephrase this inequality into one that is less sharp but will be more convenient for our purposes and which will be the basis for our free energy upper bound.
\begin{coro}\label{coromino}
  For any configuration $Z_{2N}$ satisfying the standing assumptions of this section, using the above notation, we have
  \begin{multline}
    \label{minoF02 bis}
    \F_\lambda(Z_{2N})
    \ge - \hal \sum_{ i=1}^{2N} \g_\lambda(z_i-z_{\phi_1(i)} )+ \sum_{\substack{i\in I^{\pair}\backslash I^{\dip}\\ \phi_2(i) \in I^\pair, \phi_2(\phi_1(i)) \in I^\pair}}\left( \log
    \frac{\rr_2(z_i) \wedge \rr_2(z_{\phi_1(i)})}{\rr_1(z_i)}-C\right)
    \\
    - C \sum_{i \in I^{\dip}, \phi_2(i)\in I^\pair , \phi_2(\phi_1(i)) \in I^\pair}
    \left( \frac{\rr_1(z_i)}{\rr_2(z_i)}\right)^2
    - C (N-K),
  \end{multline}where $C$ is universal, $I^\pair$ and $I^\dip$ are as in \eqref{defIdip} and $K$ is the number of components of the nearest neighbor graph.
\end{coro}
The right-hand side is the nearest-neighbor energy $\F_\lambda^\nn(Z_{2N}):=-\hal\sum_{i=1}^{2N}\g_\lambda(z_i-z_{\phi_1(i)})$ with three corrections: a logarithmic term for retained nonneutral pairs, a quadratic loss for retained neutral pairs, and the bulk loss $C(N-K)$. An isolated pair is retained when both endpoints have their second nearest neighbors in $I^\pair$. We use the same notation $\F_\lambda^\nn$ for configurations with a different number of points.
\begin{proof}[Proof of the corollary]
  Each connected component contains exactly one $2$-cycle, so there are exactly $2N-2K$ tree vertices.

  $\bullet$ The case $i\in \mathcal C_k$ and $d_id_{\phi_1(i)}=-1$. In that case we have $D_i=0$ and the corresponding term in the sum \eqref{minoF02} reduces to
  \[- \g_\lambda(z_i-z_{\phi_1(i)}) - C \left( \frac{\rr_1(z_i)}{\rr_2(z_i)}\right)^2 .\] There are $2(N-K)$ tree vertices and at most $2(N-K)$ nonisolated components, so $|(I^\pair)^c|\le6(N-K)$. If $i\notin I^{\pair}$ then we can absorb the errors $(\frac{\rr_1}{\rr_2})^2\le 1$ into $C(N-K)$. Moreover, by Lemma~\ref{lem:in-degree} a point can be the second nearest neighbor to at most 12 points. Hence, the number of $i$'s such that $\phi_2(i) \notin I^\pair$ is bounded by $C(N-K)$, so for such $i$'s we can again absorb the error $(\frac{\rr_1}{\rr_2})^2\le 1$ into $C(N-K)$. We also discard their partners, which at most doubles this bound.

  $\bullet$ The case $i\in \mathcal C_k$ and $d_i d_{\phi_1(i)}=1$. Then $d_i D_i= 2$. By definition of $\rr_2$ in \eqref{nn3} and the fact that $i\in \mathcal C_k$ and hence $z_i$ and $z_{\phi_1(i)}$ are each other's nearest neighbor,
  we have \[\rr_2(z_i) \wedge \rr_2(z_{\phi_1(i)})\ge \rr_1(z_i)=\rr_1(z_{\phi_1(i)}) \] and by definition \eqref{nn2} and \eqref{approxgeta}, we deduce that
  \begin{equation*}
    -\g(\rr_2(z_i)\wedge \rr_2(z_{\phi_1(i)} ) )\ge -\g(\rr_1(z_i) )\ge -\g_\lambda(z_i-z_{\phi_1(i)}) -C,
  \end{equation*}
  for $C>0$ a universal constant. Thus, since $\rr_1\le \rr_2$, we have
  \begin{multline*}\g_\lambda(z_i-z_{\phi_1(i)} ) -2 \left( \g\left( \rr_2(z_i) \wedge \rr_2(z_{\phi_1(i)} ) \right)+\kappa\right)- C \left(\frac{\rr_1(z_i)}{\rr_2(z_i)}\right)^2
    \\
    \ge - \g_\lambda(z_i-z_{\phi_1(i)} ) - 2\g\left(\frac{\rr_2(z_i) \wedge \rr_2(z_{\phi_1(i)})}{\rr_1(z_i)}\right) - C,\end{multline*} for some universal constant $C$.
  Moreover, $- 2\g\left(\frac{\rr_2(z_i) \wedge \rr_2(z_{\phi_1(i)})}{\rr_1(z_i)}\right) - C\ge - C$ since $\rr_2\ge \rr_1$, and thus,
  arguing as in the previous case, we may absorb such error terms into $C(N-K)$ unless $i\in I^\pair$ and $\phi_2(i),\phi_2(\phi_1(i))\in I^\pair$.

  $\bullet$ The case $i\in\mathcal I_k\setminus\mathcal C_k$. In that case the term in the sum is just
  \[-\left(\g(\rr_2(z_i)) +\kappa \right)\ge - \g_\lambda(z_i-z_{\phi_1(i)}) -C
  \] after using that $\rr_2(z_i) \ge \rr_1(z_i)$ and again \eqref{nn2} and \eqref{approxgeta}. There are $2(N-K)$ tree vertices in $\bigcup_{k=1}^K(\mathcal I_k\setminus\mathcal C_k)$, so the constant error terms can again be absorbed into $C(N-K)$ as above.
  Combining all the cases, we obtain the result.
\end{proof}

We now turn to the proof of Proposition \ref{prop:mino2}.
As explained in the introduction, the proof relies on an enlargement of the disc charges. To evaluate the change of energy, we use the following lemma.

\begin{lem}\label{lemchr} For $\vec{\tau},\vec{\alpha}\in(0,\infty)^{2N}$, we have
  \begin{equation}
    \label{llemchr}
    \frac{1}{2\pi}\int_{\R^2}\left( |\nab h_{\vec{\tau}}|^2-|\nab h_{\vec{\alpha}}|^2\right) =\sum_{i,j} \int_{\R^2} d_i d_j
    \left(\g*\delta_{z_i}^{(\tau_i)} - \g* \delta_{z_i}^{(\alpha_i)} \right)
    \left( \delta_{z_j}^{(\tau_j)}+\delta_{z_j}^{(\alpha_j)} \right).\end{equation}

\end{lem}
\begin{proof}
  By Newton's theorem, $h_{\vec\tau}-h_{\vec\alpha}$ has compact support. Positive radii give local $C^{1,1}$ regularity, and neutrality shows that $\nab h_{\vec\alpha}$ and $\nab h_{\vec\tau}$ are $O(|x|^{-2})$ at infinity. Thus both gradients lie in $L^2(\R^2)$, and integration by parts gives
  \[
    \int_{\R^2}\bigl(|\nab h_{\vec\tau}|^2-|\nab h_{\vec\alpha}|^2\bigr)
    =\int_{\R^2}\nab(h_{\vec\tau}-h_{\vec\alpha})\cdot\nab(h_{\vec\tau}+h_{\vec\alpha})
    =-\int_{\R^2}(h_{\vec\tau}-h_{\vec\alpha})\Delta(h_{\vec\tau}+h_{\vec\alpha}).
  \]
  The Poisson equations give \eqref{llemchr}.
\end{proof}

\begin{proof}[Proof of Proposition \ref{prop:mino2}]
  We are going to define for each point $z_i$ in the configuration an appropriate radius $\tau_i$.
  Each index $i$ belongs to one connected component $\mathcal{I}_k$ of the nearest-neighbor graph $\gamma$ of the configuration. We let
  \begin{equation}\label{deftau} \tau_i= \begin{cases}
      \rr_2(z_{i})\wedge \rr_2(z_{\phi_1(i) }) & \text{if $i\in \mathcal C_k$}\\
      \rr_2(z_i) & \text{otherwise}.\end{cases}\end{equation}

  We then apply Lemma \ref{lemchr} and increase the balls from $\alpha_i= \lambda $ to $\tau_i$.
  The mixed interaction sums cancel by symmetry under $i\leftrightarrow j$, giving
  \begin{multline}\label{peq}
    \int_{\R^2} |\nab h_{\lambda}|^2-
    \int_{\R^2} |\nab h_{\vec{\tau}}|^2\\=
    2\pi \sum_{i,j} d_i d_j \left(\iint \g(x-y) \delta_{z_i}^{(\lambda)} (x)
    \delta_{z_j}^{(\lambda)} (y)- \iint \g(x-y) \delta_{z_i}^{(\tau_i)} (x)
    \delta_{z_j}^{(\tau_j)} (y)
    \right)
    . \end{multline}
  First, if $\tau_i=\tau_j=\lambda$, the terms in parentheses cancel. We may thus restrict the sum to the situation where $\max (\tau_i, \tau_j) > \lambda$, which also means that $\rr_2(z_i)$ or $\rr_2(z_j)$ is the true (quarter of the) second neighbor distance.
  Next, if $B(z_i, \lambda) $ and $B(z_j, \lambda)$ intersect, so do $B(z_i, \tau_i) $ and $B(z_j, \tau_j)$ since by definition and \eqref{nn3}, $\tau_i\ge \lambda$, $\tau_j \ge \lambda$.
  If, on the other hand, $B(z_i,\lambda)$ and $B(z_j, \lambda)$ are disjoint, and $B(z_i, \tau_i)$ and $B(z_j, \tau_j)$ as well, then
  by Newton's theorem and the mean value theorem, the two terms in parentheses are equal to $\g(z_i-z_j)$ and hence cancel.
  We may thus restrict the sum to the situation where $\max(\tau_i, \tau_j) > \lambda$ and $B(z_i, \tau_i)$ and $B(z_j, \tau_j)$ intersect, hence
  \begin{equation}\label{disz}
    |z_i-z_j|\le \tau_i+\tau_j \le \rr_2(z_i) + \rr_2(z_j) \le 2\bigl(\rr_2(z_i)\vee\rr_2(z_j)\bigr) \end{equation} in view of the definitions \eqref{deftau}.

  If $i\ne j$, the larger of $\rr_2(z_i),\rr_2(z_j)$ exceeds $\lambda$. By \eqref{disz}, the corresponding point has the other as its first nearest neighbor. Suppose that $j\ne\phi_1(i)$, so $i=\phi_1(j)$, and write $d=|z_i-z_j|$. Since $z_{\phi_1(i)}$ is distinct from $z_i,z_j$, the triangle inequality gives
  \[
    \rr_2(z_i)\le(d/4)\vee\lambda,
    \qquad \rr_2(z_j)\le(d/2)\vee\lambda.
  \]
  If $\lambda\ge d/2$, then $\tau_i=\tau_j=\lambda$, contrary to the growth condition. Otherwise $\rr_2(z_i)+\rr_2(z_j)<d$, contradicting \eqref{disz}. By symmetry, the sum reduces to $i=j$ or mutual nearest neighbors, which we denote by $i\sim j$. For a mutual-neighbor pair whose enlarged discs intersect, \eqref{disz} becomes
  \[
    |z_i-z_j|\le2\tau_i=2\bigl(\rr_2(z_i)\wedge\rr_2(z_{\phi_1(i)})\bigr).
  \]

  Terms with both radii unchanged or with disjoint enlarged discs vanish. We may therefore restore all mutual-neighbor pairs to the sum. Dropping $\int_{\R^2}|\nab h_{\vec\tau}|^2\ge0$ in \eqref{peq} and using \eqref{eq:defgeta} and \eqref{geta0}, we obtain
  \begin{multline}\label{trcomp}
    \int_{\R^2} |\nab h_{\lambda}|^2 \\
    \ge 2\pi \sum_{i=1}^{2N} \left( \g_\lambda(0) - \g(\tau_i)-\kappa \right)+
    2\pi \sum_{i \neq j: j \sim i }
    d_i d_j \Bigl( \g_\lambda ( z_i-z_j) -\iint \g(x-y) \delta_{z_i}^{(\tau_i)}(x) \delta_{z_j}^{(\tau_j)} (y)\Bigr).
  \end{multline}

  We examine the contribution of the $2$-cycles. If $i\in \mathcal C_k$ and $j=\phi_1(i)$, then $\tau_i=\tau_j $ by definition and thus by \eqref{approxgeta} the contribution of the parenthesis in \eqref{trcomp} is
  \[2\pi d_id_j\Bigl( \g_\lambda(z_i-z_j)- \g(\tau_i) - \g_1(\frac{z_i-z_j}{\tau_i})\Bigr) .\]
  This term appears twice due to the two edges between the vertices of the cycle and to the equality $\tau_i=\tau_j$.

  If $i\sim j$, then, considering $z_k$ a ``second nearest neighbor'' to $z_i$, i.e.~achieving the min in \eqref{nn3}, by the triangle inequality we have
  \begin{equation*}\rr_2 (z_j) \le \frac14|z_i-z_j|+\rr_2(z_i).\end{equation*}
  Reversing the roles of $i$ and $j$, we obtain, for $i \sim j$,
  \begin{equation}\label{rrj}|\rr_2(z_i)-\rr_2(z_j)|\le \frac14|z_i-z_j| \le \rr_{1} (z_i) . \end{equation}
  Writing $d=|z_i-z_j|$, \eqref{rrj} gives
  \[
    \rr_1(z_i)\le\tau_i,\qquad
    \rr_2(z_i)\le\tau_i+d/4\le2\tau_i,\qquad
    d\le4\rr_1(z_i).
  \]
  By \eqref{idg1},
  \[
    E_i:=\kappa-\g_1\left(\frac{z_i-z_j}{\tau_i}\right),\qquad
    |E_i|\le\frac{d^2}{2\tau_i^2}
    \le32\left(\frac{\rr_1(z_i)}{\rr_2(z_i)}\right)^2.
  \]
  Inserting these facts into \eqref{trcomp}, we obtain that
  \[
    \int_{\R^2} |\nab h_{\lambda}|^2
    \ge 4\pi N \g_\lambda(0) - 2\pi
    \sum_{i=1}^{2N} ( \g(\tau_i)+\kappa)+
    2\pi \sum_{i\neq j:j \sim i }
    d_i d_j \Bigg( \g_\lambda ( z_i-z_j) - \g(\tau_i) -\kappa + E_i \Bigg)\,.
  \]
  Since $d_id_j=\pm1$, we have $d_id_jE_i\ge-C(\frac{\rr_1(z_i)}{\rr_2(z_i)})^2$. Regrouping over the components gives
  \begin{multline*}
    \int_{\R^2} |\nab h_{\lambda}|^2
    \ge 4\pi N \g_\lambda(0) \\ +
    2\pi\sum_{k=1}^K \Bigg( \sum_{i\in \mathcal C_k} \Bigl[d_i d_{\phi_1(i)} \g_\lambda(z_i-z_{\phi_1(i)} ) -d_i \left( d_i+ d_{\phi_1(i)} \right)( \g(\tau_i) +\kappa) -C\Bigl(\frac{\rr_1(z_i)}{\rr_2(z_i)}\Bigr)^2\Bigr]
    \\
    -\sum_{i\in \mathcal{I}_k \backslash \mathcal C_k} \bigl(\g(\rr_2(z_i))+\kappa\bigr) \Bigg).
  \end{multline*}
  In view of the rewriting of $\F_\lambda$ \eqref{eq:rewrF} we obtain the result.

\end{proof}

\section{Free energy upper bound}\label{sec3}

This section is devoted to the proof of the free energy upper bound.
This will be based on the energy lower bound of Corollary \ref{coromino} which reduces the interaction to nearest neighbor terms, together with a quadratic error depending on second nearest neighbor distances. The main difficulty is to partition the phase-space efficiently to integrate the exponential of this reduced energy. This is based on refinements of the Gunson--Panta change of variables and approach \cite{GunPan}.
We will start by proving upper bounds for the simpler nearest neighbor model $\F_\lambda^\nn$, then for simpler models with errors, and build up to the upper bound for the full model including the second nearest neighbor errors.
Throughout this section, $N\ge2$ and the estimates as $\lambda\to0$ are uniform in $N$.

\subsection{Preliminaries: Dirichlet integrals, Gaussian decoupling and an optimization}
We start by presenting some tools needed to implement the integrations.
\subsubsection*{Dirichlet integrals}We first recall a result on computing ``multiple Dirichlet integrals of type 1" (see \cite[p. 258]{whittaker_watson_1996}).
\begin{lem}[Dirichlet integrals]\label{lemdirichlet}
  Given an integer $k\geq 1$, $\alpha_1,\ldots,\alpha_k>0$ and $s>0$, let
  \begin{equation*}
    I_{k,s}(\alpha_1, \dots, \alpha_k) :=\int_{(\R^+)^k} \indic_{0<t_1+\ldots+t_k<s}\, t_1^{\alpha_1-1}\ldots t_k^{\alpha_k-1}\dd t_1\ldots \dd t_k.
  \end{equation*}
  We have
  \begin{equation}
    \label{eq:ind I}
    I_{k,s}(\alpha_1, \dots, \alpha_k) =s^{\alpha_1+\dots+ \alpha_k} \frac{\Gamma(\alpha_1)\ldots \Gamma(\alpha_k)}{\Gamma(\alpha_1+\ldots+\alpha_k)}\frac{1}{\alpha_1+\ldots+\alpha_k}.
  \end{equation}

  For $k_0\in\{0,\ldots,k\}$, if the $k$-tuple $(\alpha_i)_{i=1}^k$ is defined by $\alpha_i=\alpha>0$ for $1\le i\le k_0$ and $\alpha_i=1$ for $k_0+ 1\le i \le k$, then
  \begin{equation}\label{eq:estimate Dirichlet}
    \log I_{k,s}=(\alpha k_0+(k-k_0))\left( \log s- \log(\alpha k_0+k-k_0)\right)+k+O_\alpha(k_0+\log(k+1)).
  \end{equation}
  Moreover,
  \begin{equation*} I_{k,s}\le \left( \frac{s}{\alpha k_0 +k-k_0} \right)^{\alpha k_0+k-k_0}e^k C_\alpha^{k_0}.\end{equation*}
\end{lem}
\begin{proof}
  The identity
  \begin{equation*}\int_{(\R^+)^k} \indic_{0<t_1+\ldots+t_k<1 }\ t_1^{\alpha_1-1}\ldots t_k^{\alpha_k-1}\dd t_1\ldots \dd t_k
    = \frac{\Gamma(\alpha_1)\ldots \Gamma(\alpha_k)}{\Gamma(\alpha_1+\ldots+\alpha_k)}\frac{1}{\alpha_1+\ldots+\alpha_k}\end{equation*}
  follows by successive applications of Euler's beta integral.
  It follows by the change of variables $t_i=s u_i$ that \eqref{eq:ind I} holds.

  For the particular choice of $\alpha_i$'s stated in the lemma, we thus find that
  \begin{equation*}
    I_{k,s}=s^{ \alpha k_0+k-k_0}\frac{\Gamma(\alpha)^{k_0}}{\Gamma(\alpha k_0+k-k_0)}\frac{1}{\alpha k_0+k-k_0}.
  \end{equation*}

  Applying Stirling's formula, we obtain
  \begin{multline*}\log I_{k,s}= (\alpha k_0+k-k_0)\log s+ k_0\log \Gamma(\alpha) \\ - (\alpha k_0+k-k_0) \log (\alpha k_0+k-k_0)+ (\alpha k_0+k-k_0)+O_\alpha(\log(k+1)).\end{multline*}

  The relation \eqref{eq:estimate Dirichlet} follows. We deduce the upper bound from $\Gamma(x+1)\ge(x/e)^x$ for $x>0$.\end{proof}

\subsubsection*{Gunson--Panta change of variables}
The Gunson--Panta method relies on considering the nearest neighbor graph of a configuration and partitioning the phase-space accordingly. We use the notation introduced already in Section \ref{sec2}, but we can apply it more generally to $p$ points, and not necessarily an even number of points.
For $p\ge2$ distinct points $z_1,\ldots,z_p\in\Lambda$ with unique nearest neighbors, let $\mathcal{I}_1, \dots , \mathcal{I}_K$ denote the connected components of their nearest-neighbor graph, and denote by $m_k$ and $m_k'$ the indices of the two vertices of the 2-cycle $\mathcal C_k$. On each graph, choose $m_k'=\max\mathcal C_k$ and set
\[
  \Omega_\gamma:=\{Z_p\in\Lambda^p:\gamma_p(Z_p)=\gamma\}.
\]
We then define the Gunson--Panta change of variables via the map
\begin{equation}\label{defphi}
  \Phi^{GP}_p (z_1, \dots, z_p) = (u_1, \dots, u_p),\end{equation} where
\begin{equation*}
  u_i=\left\{
  \begin{array}{ll}
    z_i-z_{\phi_1(i)} & \text{if }i\in \mathcal{I}_k,i\neq m_k' \\
    z_{i} & \text{if } i=m_k'.
  \end{array}
  \right.
\end{equation*}
Deleting the outgoing edge of each $m_k'$ leaves a rooted tree in each component. Ordering parents before children makes the map triangular with identity diagonal blocks. Hence $\Phi_p^{GP}$ is invertible on $\Omega_\gamma$ and $|\det D\Phi_p^{GP}|=1$.

\subsubsection*{Gaussian decoupling of packing constraints}
After the Gunson--Panta change of variables, the integrands are products of radial weights in the variables $u_j$, but the domain of integration is not a product. The discs $B(z_i,|z_i-z_{\phi_1(i)}|/2)$ are disjoint and contained in $[-\sqrt N,2\sqrt N]^2$, which has area $9N$. Hence the variables $u_i$ other than the reference points $u_{m_k'}=z_{m_k'}$ satisfy the packing constraint
\[
  \sum_{k=1}^K\sum_{i\in\mathcal I_k\setminus\{m_k'\}}|u_i|^2\le\frac{36N}{\pi}\,.
\]
Lemma~\ref{lemdirichlet} computes such constrained integrals exactly for power weights. For the other weights, we decouple the constraint as follows.

\begin{lem}[Gaussian decoupling]\label{lem:gaussian decoupling}
  Let $k\ge1$ be an integer and let $s,a>0$. For all $u_1,\ldots,u_k\in\R^2$, we have
  \begin{equation}\label{eq:gaussian-decoupling}
    \indic_{\sum_{j=1}^k|u_j|^2\le s}\le e^{a(s-\sum_{j=1}^k|u_j|^2)}\,.
  \end{equation}
  Consequently, for all measurable functions $w_1,\ldots,w_k:\R^2\to[0,\infty)$, we have
  \[
    \int_{(\R^2)^k}\indic_{\sum_{j=1}^k|u_j|^2\le s}\prod_{j=1}^kw_j(u_j)\dd u_1\ldots\dd u_k
    \le e^{as}\prod_{j=1}^k\int_{\R^2}w_j(u)e^{-a|u|^2}\dd u\,.
  \]
\end{lem}

\begin{proof}
  The right side of \eqref{eq:gaussian-decoupling} is positive, and it is at least $1$ when $\sum_j|u_j|^2\le s$. This proves \eqref{eq:gaussian-decoupling}. Integrating it against $\prod_jw_j(u_j)$ and using Fubini's theorem gives the second bound.
\end{proof}

In the proof of Proposition~\ref{prop:auxk}, the parameter $a$ is of order $1$, and its value affects only the constants. The weight $e^{2\g_\lambda(u)}$ behaves differently. It equals $|u|^{-2}$ for $|u|\ge2\lambda$, so its integral over $\R^2$ diverges logarithmically. After the Gaussian cutoff, its mass is $\pi\log(1/(a\lambda^2))+O(1)$, which depends on $a$ at leading order. This weight appears for $2$-cycles at $\beta=2$ and for tree edges at $\beta=4$. The next lemma collects its integrals.

\begin{lem}\label{lem:critical radial}
  Let $0<\lambda\le e^{-1}$ and $\ell=|\log\lambda|$. For $a>0$ and $R\ge0$, let
  \[
    J_\lambda(a;R):=\int_{|u|\ge R}e^{2\g_\lambda(u)-a|u|^2}\dd u
    \qquad\text{and}\qquad
    J_\lambda(a):=J_\lambda(a;0)\,.
  \]
  There is a universal constant $C>0$ such that the following hold, where every $O(1)$ term is bounded by $C$ in absolute value.
  \begin{enumerate}
    \item If $a\lambda^2\le1/4$, then
    \begin{equation}\label{eq:critical-gaussian}
      J_\lambda(a)=\pi\log\frac1{a\lambda^2}+O(1)\,.
    \end{equation}
    \item We have
    \begin{equation}\label{eq:critical-B}
      J_\lambda(1/\ell)=2\pi\Bigl(\ell+\frac12\log\ell\Bigr)+O(1)
      \quad\text{and}\quad
      \log J_\lambda(1/\ell)\le\log(2\pi\ell)+C\frac{1+\log\ell}{\ell}\,.
    \end{equation}
    If $2\lambda\le R\le1$, then
    \begin{equation}\label{eq:critical-BR}
      J_\lambda(1/\ell;R)=2\pi\Bigl(\log\frac1R+\frac12\log\ell\Bigr)+O(1)\,.
    \end{equation}
    \item For every $C_1>0$, there is $C_2>0$ such that, for every $\rho\ge\lambda$,
    \begin{equation}\label{eq:critical-correction}
      \begin{aligned}
        &\int_{|u|\le4\rho}e^{2\g_\lambda(u)}\Bigl(e^{C_1((|u|/4)\vee\lambda)^2/\rho^2}-1\Bigr)\dd u\le C_2\qquad\text{and}\\
        &\int_{|u|\le4\rho}e^{2\g_\lambda(u)}\Bigl(\frac{(|u|/4)\vee\lambda}{\rho}\Bigr)^4\dd u\le C_2\,.
      \end{aligned}
    \end{equation}
  \end{enumerate}
\end{lem}

\begin{proof}
  By \eqref{approxgeta}, $\g_\lambda(u)\le\g(|u|\vee\lambda)+C\le\ell+C$ on $B(0,2\lambda)$, so the core satisfies
  \begin{equation}\label{eq:critical-core}
    0\le\int_{|u|<2\lambda}e^{2\g_\lambda(u)}\dd u\le4\pi e^{2C}.
  \end{equation}
  By Newton's theorem, $\g_\lambda(u)=-\log|u|$ for $|u|\ge2\lambda$. Let $E_1(x):=\int_x^\infty e^{-t}t^{-1}\dd t$. For $0<x\le1$, the bounds $|e^{-t}-1|\le t$ and $\int_1^\infty e^{-t}t^{-1}\dd t\le e^{-1}$ give
  \begin{equation}\label{eq:E1}
    \Bigl|E_1(x)-\log\frac1x\Bigr|\le1.
  \end{equation}

  (1) In polar coordinates and with $t=ar^2$, the part $|u|\ge2\lambda$ of $J_\lambda(a)$ equals
  \[
    2\pi\int_{2\lambda}^\infty\frac{e^{-ar^2}}{r}\dd r=\pi E_1(4a\lambda^2).
  \]
  Since $4a\lambda^2\le1$, \eqref{eq:E1} and \eqref{eq:critical-core} give \eqref{eq:critical-gaussian}.

  (2) Since $\lambda^2/\ell\le e^{-2}<1/4$, item (1) applies with $a=1/\ell$ and gives $J_\lambda(1/\ell)=\pi\log(\ell/\lambda^2)+O(1)=2\pi\ell+\pi\log\ell+O(1)$. Hence $J_\lambda(1/\ell)\le2\pi\ell\,(1+(\log\ell+C)/(2\ell))$, and $\log(1+x)\le x$ gives the bound on $\log J_\lambda(1/\ell)$. Since $R\ge2\lambda$, the integrand of $J_\lambda(1/\ell;R)$ is $|u|^{-2}e^{-|u|^2/\ell}$, so $J_\lambda(1/\ell;R)=\pi E_1(R^2/\ell)$. Since $R^2/\ell\le1$, \eqref{eq:E1} gives \eqref{eq:critical-BR}.

  (3) For $|u|\le4\rho$, we have $(|u|/4)\vee\lambda\le\rho$, and $e^{C_1x^2}-1\le(e^{C_1}-1)x^2$ for $0\le x\le1$, by convexity. It therefore suffices to bound $\int_{|u|\le4\rho}e^{2\g_\lambda(u)}((|u|/4)\vee\lambda)^2\rho^{-2}\dd u$, which also dominates the second integral. On the core, the factor $((|u|/4)\vee\lambda)^2\rho^{-2}$ is at most $1$, and \eqref{eq:critical-core} applies. For $|u|\ge2\lambda$, we have $(|u|/4)\vee\lambda\le|u|/2$, so the annulus contributes at most
  \[
    \int_{2\lambda\le|u|\le4\rho}\frac1{|u|^2}\frac{|u|^2}{4\rho^2}\dd u\le\frac{\pi(4\rho)^2}{4\rho^2}=4\pi.
  \]
\end{proof}

\subsubsection*{An elementary optimization}
In the proofs below, the entropy of the graphs and the energy combine into expressions of the form $\delta(A+B\log(1/\delta))$, where $0<\delta\le1$ is a proportion of vertices or of pairs. For every real number $A$ and every $B>0$, we have
\begin{equation}\label{eq:sup-inequality}
  \sup_{0<\delta\le1}\delta\bigl(A+B\log(1/\delta)\bigr)\le Be^{A/B-1}\,.
\end{equation}
Indeed, the function $\delta\mapsto\delta(A+B\log(1/\delta))$ is concave on $(0,\infty)$, and its derivative $A-B+B\log(1/\delta)$ vanishes at $\delta=e^{A/B-1}$. Hence its maximum on $(0,\infty)$ is its value at this point, which is $Be^{A/B-1}$.

\subsection{Upper bound for a nearest neighbor model}

Corollary \ref{coromino} tells us that, up to an error term involving ratios of nearest and second-nearest neighbor distances, one can bound $\F_\lambda$ from below by the nearest-neighbor energy $\F_\lambda^\nn(Z_{2N})$ defined after it.
We examine in the next two lemmas the nearest neighbor model studied in \cite{GunPan} and the reduced dipole model. The proofs will be a model for what follows.

\begin{lem}[Upper bounds for the nearest-neighbor model on a fixed graph]\label{lemma:nearest neigh}
  Let $\beta>2$ and $\mathcal Z_\beta$ be the constant defined in \eqref{def:Cbeta}. Let $p\in \{1,\ldots,N\}$, $K\in \{1,\ldots,p\}$ and $\gamma\in D_{2p,K}$. The bounds below hold uniformly in $N\ge2$ for $0<\lambda<\lambda_\beta$, where $\lambda_\beta\in(0,e^{-2}]$ depends only on $\beta$.

  $\bullet$ For $\beta\in (2,4)$, we have
  \begin{multline*}
    \log \int_{\Omega_\gamma} \exp\left(-\beta \F_\lambda^{\nn}(Z_{2p})\right) \dd Z_{2p} \\ \leq K \log N
    + (2-\beta)K\log\lambda+K \log \mathcal Z_{\beta}+\Bigl(1-\frac{\beta}{4}\Bigr)(2p-2K)\log \frac{N}{p-K}+C_\beta(p-K).
  \end{multline*}

  $\bullet$ For $\beta=4$, we have
  \begin{multline*}
    \log \int_{\Omega_\gamma} \exp\left(-\beta \F_\lambda^{\nn}(Z_{2p})\right) \dd Z_{2p} \leq K \log N
    \\ +(2-\beta)K\log \lambda+K\log \mc{Z}_\beta+ (2p-2K)\log\Bigl(|\log \lambda|+\log \frac{N}{p-K}\Bigr)+C_\beta(p-K).
  \end{multline*}

  $\bullet$ For $\beta>4$, we have
  \begin{multline*}
    \log \int_{\Omega_\gamma} \exp\left(-\beta \F_\lambda^{\nn}(Z_{2p})\right) \dd Z_{2p}
    \leq K \log N \\ +(2-\beta)K\log \lambda+K\log \mc{Z}_\beta+(2p-2K)\Bigl(2-\frac{\beta}{2}\Bigr)\log \lambda+C_\beta(p-K).
  \end{multline*}
  Terms proportional to $p-K$ are zero when $p=K$.
\end{lem}

\begin{proof}\textbf{Step 1: change of variables.}
  Let $\gamma\in D_{2p,K}$. As in Section 2, let us denote by $\mathcal{I}_k$, $k=1,\ldots,K$, the connected components of $\gamma$ and $\mc{C}_k:=\{m_k,m_k'\}\subset \mathcal{I}_k$ the $2$-cycle in each connected component.

  Performing the Gunson--Panta change of variables \eqref{defphi} in the integral of the statement, we may write
  \begin{equation}\label{eq:red neutral}
    \int_{\Omega_\gamma}\exp\left(-\beta \F_\lambda^{\nn}(Z_{2p})\right)\dd Z_{2p}
    =\int_{\Phi^{GP}_{2p}(\Omega_\gamma)}\prod_{k=1}^K \Bigg[\exp(\beta\g_\lambda(u_{m_k}))\prod_{\substack{i\in \mathcal{I}_k\backslash \mathcal C_k }} \exp\Bigl(\frac{\beta}{2}\g_{\lambda}(u_i)\Bigr)\Bigg]\dd U_{2p}.
  \end{equation}
  Indeed, if $\{m_k,m_k'\}$ is a $2$-cycle, the interaction $\g_\lambda(u_{m_k})$ is counted twice, which explains why this term appears in the above integral with a factor $\beta$ instead of $\beta/2$. This turns out to be crucial:
  since $\beta>2$, the integral of the $2$-cycle factor $\exp(\beta \g_\lambda(u_{m_k}))$ diverges as $\lambda\to0$. On the other hand, for a tree edge, that is, for $i\in\mathcal I_k\setminus\mathcal C_k$, the weight $\exp(\frac\beta2\g_\lambda(u_i))$ is locally integrable uniformly in $\lambda$ if $\beta<4$, while its integral near the origin diverges as $\lambda\to0$ if $\beta\ge4$.
  \smallskip

  \noindent
  \textbf{Step 2: separation of variables.}
  The domain of integration in \eqref{eq:red neutral} is a complicated subset of $\R^{4p}$ but one can approximate it by only keeping the volume constraint:
  \begin{equation*}
    \Phi^{GP}_{2p}(\Omega_\gamma)\subset B:=\Bigl\{U_{2p}\in \R^{4p}: \sum_{k=1}^{K}\sum_{i\in \mathcal{I}_k,i\neq m_k'} |u_i|^2 \leq \frac{ 36 N}{\pi}, \forall k\in \{1,\ldots,K\}, u_{m_k'}\in \Lambda \Bigr\}.
  \end{equation*}
  This takes advantage of the fact that the balls $B(z_i, \frac{|z_i-z_{\phi_1(i)}|}{2})$ are disjoint balls in $ [-\sqrt N, 2 \sqrt N]^2$.

  Splitting the variables into two categories (those corresponding to $i =m_k'$, i.e.~the reference points, and the others) and integrating a reference point in each connected component in the box $\Lambda$ of volume $N$, we find the bound below. In its last integral, we relabel the tree-edge variables $u_i$ with $i\in\bigcup_k(\mathcal I_k\setminus\mathcal C_k)$ as $u_1,\ldots,u_{2p-2K}$.
  \begin{multline}\label{eq:cutt}
    \int_{B}\prod_{k=1}^K \Bigg[\exp(\beta\g_\lambda(u_{m_k}))\prod_{i\in \mathcal{I}_k \backslash \mathcal C_k } \exp\Bigl(\frac{\beta}{2}\g_\lambda(u_i) \Bigr)\Bigg]\dd U_{2p}\\ \leq \Bigl(\int_{B_1}
    \prod_{k=1}^K \exp(\beta \g_\lambda(u_{m_k}))
    \dd u_{m_k}\Bigr) N^K
    \times\int_{B_2}\exp\Bigl(\frac{\beta}{2}\g_\lambda(u_1)\Bigr)\ldots \exp\Bigl(\frac{\beta}{2}\g_\lambda(u_{2p-2K})\Bigr)\dd U_{2p-2K},
  \end{multline}
  where
  \begin{equation*}
    B_1:=\Bigl\{U_{K}\in (\R^2)^K:|U_{K}|^2\leq \frac{36N}{\pi}\Bigr\},
  \end{equation*}
  \begin{equation*}
    B_2:=\Bigl\{U_{2p-2K}\in (\R^2)^{2p-2K}:|U_{2p-2K}|^2\leq \frac{ 36 N}{\pi}\Bigr\}.
  \end{equation*}

  \noindent
  \textbf{Step 3: integration of the $2$-cycle distances.}

  Performing a polar change of coordinates and using \eqref{approxgeta}, we get
  \begin{equation}\label{3.20b}
    \int_{\R^2}\exp(\beta \g_\lambda(u))\dd u=2\pi\int_0^{\infty}\exp(\beta \g_\lambda(r))r\dd r=2\pi\lambda^{2-\beta}\int_0^{\infty}\exp(\beta \g_1(r))r\dd r=\lambda^{2-\beta}\mathcal Z_{\beta}>0
  \end{equation}
  by definition of $\mc{Z}_\beta$, see \eqref{def:Cbeta}. Hence, removing the restriction to the set $B_1$, we obtain
  \begin{equation*}
    \log\int_{B_1} \prod_{k=1}^K \exp(\beta \g_\lambda(u_{m_k}))\dd u_{m_k}\leq K\log (\lambda^{2-\beta}\mc{Z}_\beta).
  \end{equation*}

  \noindent
  \textbf{Step 4: Integration of the other distances and conclusion.}
  For $p=K$, the $B_2$ integral is~$1$. We assume $p>K$ in this step.

  $\bullet$ The case $\beta\in (2,4)$. For $i\in \{1,\ldots,2p-2K\}$, the unregularized weight $\exp(\frac{\beta}{2}\g(u_i))$ is locally integrable, and therefore the cutoff $\lambda$ can be removed. Using that $\g_\lambda\leq \g+C$ from \eqref{approxgeta} we find
  \begin{multline*}
    \int_{B_2}\exp\Bigl(\frac{\beta}{2}\g_\lambda(u_1)\Bigr)\cdots \exp\Bigl(\frac{\beta}{2}\g_\lambda(u_{2p-2K})\Bigr)\dd U_{2p-2K}\\ \leq C_\beta^{2p-2K}\int_{B_2}\frac{1}{|u_1|^{\frac{\beta}{2}}}\cdots \frac{1}{|u_{2p-2K}|^{\frac{\beta}{2}}}\dd U_{2p-2K}.
  \end{multline*}
  By performing again a polar change of coordinates and the change of variables $r_i'= r_i^2$, we can rewrite this as an integral over a simplex of $\R^{2p-2K}$:
  \begin{equation}\label{eq:reD}
    \int_{B_2}\frac{1}{|u_1|^{\frac{\beta}{2}} }\cdots \frac{1}{|u_{2p-2K}|^{\frac{\beta}{2} } }\dd U_{2p-2K}=\pi^{2p-2K} \int_{B_2'}{r_1^{-\frac{\beta}{4}}}\ldots{r_{2p-2K}^{-\frac{\beta}{4}} }\dd R_{2p-2K},
  \end{equation}
  where
  \begin{equation*}
    B_2':=\left\{R_{2p-2K}\in (\R^{+})^{2p-2K}: r_1+\ldots+r_{2p-2K}\leq \frac{36 N}{\pi}\right\}.
  \end{equation*}

  Since $\beta\in (2,4)$, we can set $\alpha=1-\frac{\beta}{4}>0$ and $k=k_0=2p-2K$ and insert \eqref{eq:estimate Dirichlet} with $ s= \frac{36 N}{\pi}$ into \eqref{eq:reD} to obtain
  \begin{multline}\label{eq:bound dir}
    \log \int_{B_2}
    \frac{1}{|u_1|^{\frac{\beta}{2}} }\cdots \frac{1}{|u_{2p-2K}|^{\frac{\beta}{2} } }\dd U_{2p-2K}
    \\\leq \Bigl(1-\frac{\beta}{4}\Bigr) (2p-2K) \log N -\Bigl(1-\frac{\beta}{4}\Bigr) (2p- 2K)\log\Bigl(\Bigl(1-\frac\beta4\Bigr) (2p-2K)\Bigr)+C_\beta(p-K)\\
    = \Bigl( 1-\frac\beta4\Bigr)(2p-2K)\log \frac{N}{p-K}+C_\beta(p-K).
  \end{multline}

  $\bullet$ The case $\beta=4$. Here the tree weight is $e^{2\g_\lambda}$. We apply Lemma~\ref{lem:gaussian decoupling} to the $2p-2K$ tree variables, with $s=36N/\pi$ and $a=(2p-2K)/(sL)$, where $L=\ell+\log(N/(p-K))$. Since $2p-2K\le2N$ and $L\ge\ell\ge1$, we have $a\le\pi/18$ and $as\le2p-2K$. Hence $a\lambda^2\le1/4$, and the decoupling costs at most $e^{2p-2K}$. Since $\log(1/(a\lambda^2))=2\ell+\log(sL/(2p-2K))\le2L+\log L+C$, the estimate \eqref{eq:critical-gaussian} gives $\log J_\lambda(a)\le\log L+C$. Hence
  \begin{multline*}
    \log \int_{B_2}\exp\Bigl(\frac{\beta}{2}\g_\lambda(u_1)\Bigr)\cdots \exp\Bigl(\frac{\beta}{2}\g_\lambda(u_{2p-2K})\Bigr)\dd U_{2p-2K}\leq (2p-2K)\log\Bigl(\ell+\log \frac{N}{p-K}\Bigr)
    \\+C_\beta(p-K).
  \end{multline*}

  $\bullet$ The case $\beta>4$. Arguing as in \eqref{3.20b}, one can check that
  \begin{equation*}
    \log \int_{B_2}\exp\Bigl(\frac{\beta}{2}\g_\lambda(u_1)\Bigr)\cdots \exp\Bigl(\frac{\beta}{2}\g_\lambda(u_{2p-2K})\Bigr)\dd U_{2p-2K}\leq (2p-2K)\Bigl(2-\frac{\beta}{2}\Bigr)\log \lambda+C_\beta(p-K).
  \end{equation*}

  Inserting the estimates of Steps 3 and 4 into \eqref{eq:cutt} proves the result.

\end{proof}

We now expand the partition function of the reduced dipole model by summing over the neutral $2$-cycles.

\begin{lem}[Free energy of the reduced dipole model]\label{lemma:part nn}
  Let $\beta\ge2$. There exists $\lambda_\beta\in(0,e^{-2}]$, depending only on $\beta$, with the following property. Let $N\ge2$ and $0<\lambda<\lambda_\beta$, let
  \begin{equation*}
    \tilde{\F}_\lambda(Z_{2N}):= -\frac{1}{2}\sum_{i=1}^{2N} \g_\lambda(z_i-z_{\phi_1(i)}) \indic_{\phi_1\circ\phi_1(i)=i, d_i d_{\phi_1(i) }=-1}
  \end{equation*}
  be the energy of the neutral $2$-cycles, and let
  \begin{equation*}
    \K_{N,\beta}^\lambda:=\int_{\Lambda^{2N} } \exp(-\beta \tilde{\F}_\lambda(Z_{2N}))\dd Z_{2N}.
  \end{equation*}
  Let
  \begin{equation}\label{def:o nn}
    \bar \omega_\lambda:=\begin{cases}
      |\log\lambda|^{-\nicefrac{1}{2}} & \text{if $\beta=2$},\\
      \lambda^{\min\{(\beta-2)/2,1\}} & \text{if $\beta>2$}.
    \end{cases}
  \end{equation}
  Then $0<\K_{N,\beta}^\lambda<\infty$ and
  \[
    \log\K_{N,\beta}^\lambda=2N\log N+
    \begin{cases}
      N\bigl((2-\beta)\log\lambda+\log\mc Z_\beta-1\bigr)+O_\beta(N\bar\omega_\lambda+\log(N+1)) & \text{if $\beta>2$},\\
      N\bigl(\log(2\pi|\log\lambda|)-1\bigr)+O(N\bar\omega_\lambda+\log(N+1)) & \text{if $\beta=2$}.
    \end{cases}
  \]
\end{lem}
Note that the expansion of $\log \K_{N,\beta}^\lambda$ is the same as \eqref{eq:devlogz} up to the error terms.
\smallskip

\begin{proof}
  We obtain the upper bound by summing over the neutral $2$-cycles, and the lower bound by restricting to configurations of $N$ short neutral pairs. First assume $\beta=2$.
  For a configuration, let $\mc M$ be the set of pairs $\{x_i,y_j\}$ which are mutual nearest neighbors. Since nearest neighbors are unique, these pairs are disjoint, so $\mc M$ is a partial matching between positive and negative particles. There are
\[
  \binom{N}{p}^2p!=\frac{(N!)^2}{p!\,(N-p)!^2}
\]
such matchings with $|\mc M|=p$.

  We partition the configurations according to this matching, and we write $u_j$ for the relative position of the two charges of the $j$-th pair of $\mc M$. These two charges are mutual nearest neighbors, so the discs of radius $|u_j|/2$ centered at the $2p$ matched charges are disjoint. Since $|u_j|\le\sqrt{2N}$, these discs lie in $[-\sqrt N,2\sqrt N]^2$, which gives the packing constraint $\sum_{j=1}^p|u_j|^2\le CN$. By Lemma~\ref{lem:gaussian decoupling} with $a=1/\ell$, we may therefore insert $e^{-\sum_j|u_j|^2/\ell}$ at cost $e^{CN/\ell}$. We then integrate one charge of each pair and the $2N-2p$ remaining particles over $\Lambda$, which gives $N^{2N-p}$, and we extend each integral over $u_j$ to $\R^2$. Summing over the matchings yields
  \begin{equation*}
    \K_{N,2}^\lambda\le e^{CN/\ell}\sum_{p=0}^N
    \frac{(N!)^2}{p!(N-p)!^2}N^{2N-p}{J_\lambda(1/\ell)}^p
    \le N!N^N{J_\lambda(1/\ell)}^N e^{CN/\ell+2N/\sqrt{J_\lambda(1/\ell)}}.
  \end{equation*}
  Indeed, setting $d=N-p$ and using $N!/(N-d)!\le N^d$, the sum divided by $N!N^N{J_\lambda(1/\ell)}^N$ is at most
  \[
    \sum_{d=0}^{\infty}\frac{(N/\sqrt{J_\lambda(1/\ell)})^{2d}}{(d!)^2}
    \le\left(\sum_{d=0}^{\infty}\frac{(N/\sqrt{J_\lambda(1/\ell)})^d}{d!}\right)^2
    =e^{2N/\sqrt{J_\lambda(1/\ell)}}.
  \]
  By \eqref{eq:critical-B}, we have $J_\lambda(1/\ell)\ge\pi\ell$ for $\lambda$ small, so the last factor is at most $e^{CN\ell^{-1/2}}$. The bound on $\log J_\lambda(1/\ell)$ in \eqref{eq:critical-B} completes the upper bound.

  For the lower bound, let $t=\ell^{-1/2}$. Restrict positive particles to boundary distance at least $t$ and mutual distances at least $4t$, and dipole lengths to $2\lambda\le|u|\le t^2$. At each placement of a positive particle, the boundary strip excludes area at most $C\sqrt N\,t$, and the previous particles exclude at most $CNt^2$. For $\lambda$ small enough, successive placement therefore gives volume at least
  \[
    N^N e^{-C(\sqrt N\,t+Nt^2)}\ge N^Ne^{-CNt}.
  \]
  Particles in different pairs are farther apart than the partners within each pair. Thus all partners are mutual nearest neighbors, and the $N!$ assignments of negative labels give disjoint configuration sets. Each pair contributes $2\pi(\ell-\log\ell-\log2)$. Since $\log\ell/\ell=O(\ell^{-1/2})$ and $\mc Z_2=2\pi$, the upper and lower bounds and Stirling's formula prove the expansion at $\beta=2$, with error $O(N\ell^{-1/2}+\log(N+1))$.

  We now assume $\beta>2$. Recall from \eqref{3.20b} that $\mc Z_\beta\lambda^{2-\beta}=\int_{\R^2}e^{\beta\g_\lambda(u)}\dd u$. The same matching bound, without the Gaussian factor, gives
  \[
    \K_{N,\beta}^\lambda
    \le\sum_{p=0}^N\frac{(N!)^2}{p!(N-p)!^2}N^{2N-p}(\mc Z_\beta\lambda^{2-\beta})^p
    \le N!N^N(\mc Z_\beta\lambda^{2-\beta})^N e^{2N/\sqrt{\mc Z_\beta\lambda^{2-\beta}}}.
  \]
  For the lower bound, take $t=\lambda^{(\beta-2)/\beta}$. Impose the same restrictions on the positive particles, and allow dipole lengths $|u|\le t$. The preceding placement estimate gives volume at least
  \[
    N^N e^{-C(Nt^2+\sqrt N\,t)}\ge N^N e^{-CNt^2-C}.
  \]
  The last inequality uses $\sqrt N\,t\le Nt^2+1/4$. Points in different pairs are at distance at least $2t$, so the partners are again mutual nearest neighbors and the label assignments are disjoint. For $\lambda$ small enough, $t\ge2\lambda$, and the omitted radial tail is $2\pi t^{2-\beta}/(\beta-2)$. Hence each pair contributes
  \[
    \int_{|u|\le t}e^{\beta\g_\lambda(u)}\dd u
    \ge \mc Z_\beta\lambda^{2-\beta}\left(1-C_\beta(\lambda/t)^{\beta-2}\right).
  \]
  Multiplying these bounds over the pairs and summing over the label assignments gives
  \[
    \K_{N,\beta}^\lambda\ge N!N^N(\mc Z_\beta\lambda^{2-\beta})^N
    \exp\left[-C_\beta N\bigl(t^2+(\lambda/t)^{\beta-2}\bigr)-C\right].
  \]
  Here $t^2=(\lambda/t)^{\beta-2}=\lambda^{2(\beta-2)/\beta}$. Since
  \[
    \frac{2(\beta-2)}{\beta}\ge\min\{(\beta-2)/2,1\},
    \qquad (\mc Z_\beta\lambda^{2-\beta})^{-1/2}=O_\beta(\lambda^{(\beta-2)/2}),
  \]
  the errors in both bounds are $O_\beta(N\bar\omega_\lambda+1)$. Stirling's formula concludes the proof.
\end{proof}

\subsection{Upper bound on the free energy of a nearest-neighbor model with marked points}
We now study a new nearest-neighbor model which will be useful for evaluating the contributions of the dipole-dipole interaction or quadrupole errors that appear in \eqref{minoF02 bis}. We consider an integral over $p$ variables $z_1,\ldots,z_p$ living on $\Lambda^{p}$, where the nearest neighbor interaction energy is counted only for a small subset of $k$ points $z_1,\ldots,z_k$.
The estimates track the combinatorial cost of placing both endpoints of a $2$-cycle in the marked subset. This cost is useful when $k/p$ is small.

\begin{prop}\label{prop:auxk}

  Let $\beta>0$, $0<\lambda\le e^{-1}$ and $0\le k\le p\le N$ be such that $k,p$ are integers and $p\ge\max(2,N/2)$. Let us consider the energy over $\Lambda^{p}$ defined by
  \begin{equation} \label{defFaux}
    F_k(z_1, \dots, z_p) := \hal \sum_{i=1}^{k} \log \rr_1(z_i),
  \end{equation}
  where $\rr_1$ is defined as in \eqref{nn2}.
  Let $\delta=0$ when $\beta\neq 4$ and $\delta\in \{0,1\}$ when $\beta=4$. Then, we have
  \[
    \log \int_{\Lambda^p} \exp(-\beta F_k(Z_p))\prod_{i=1}^k \Bigl(1+\log \frac{\rr_1(z_i)}{\lambda}\Bigr)^{\delta}\dd z_1\ldots \dd z_p \leq p \log N+C_\beta k+\frac{k}{2}\nu_{\beta,k}\,,
  \]
  where
  \begin{equation}\label{def:gbetak}
    \nu_{\beta,k}= \begin{cases}
      0 & \text{if } \ \beta \in (0,2)
      \\ \log\Bigl(|\log \lambda| \frac{k}{N} \vee 1\Bigr)& \text{if } \ \beta =2
      \\ \log\Bigl(\lambda^{2-\beta}\frac{k}{N}\vee 1\Bigr) & \text{if } \ \beta\in (2,4)
      \\ \log\Bigl(\lambda^{-2}\frac{k}{N}\vee |\log\lambda|^{2(1+\delta)}\Bigr)& \text{if $\beta=4$ }
      \\ \log\Bigl(\lambda^{2-\beta}\frac{k}{N}\vee \lambda^{4-\beta}\Bigr) & \text{if } \ \beta>4.
    \end{cases}
  \end{equation}
  Every occurrence of $\nu_{4,k}$ uses the current value of $\delta$.
\end{prop}

\medskip

\begin{proof}
  For $k=0$, the integral equals $N^p$. We assume $k\ge1$ below.
  For brevity, we let $r(z_i)= \frac 14 \min_{j\neq i} |z_i-z_j|$ and note that $\rr_1(z_i) = r(z_i) \vee \lambda$ by \eqref{nn2}.

  \noindent
  \textbf{Step 1: counting graphs.}
  We define the directed graph $\gamma_p':=(V,E)$ where the set of vertices $V$ is given by $\{1,\ldots,p\}$ and there is an edge $\vec{ij}\in E$ when $i\in \{1,\ldots,k\}$ and $\phi_1(i)=j$. The graph $\gamma_p'$ is a subgraph of the nearest-neighbor graph $\gamma_p$. We claim that the nontrivial connected components (i.e.,~not reduced to a point) of $\gamma_p'$ are as follows: either they contain a $2$-cycle within $\{1,\ldots, k\}$ and no point in $ V\setminus \{1, \ldots, k\}$,
  or they do not and they contain a unique point in $V\setminus \{1, \ldots, k\}$. In the latter case the connected component is a tree and its root is in $V\setminus \{1,\ldots,k\}$.

  Let us prove the above claim. We can observe that $V\setminus \{1,\ldots,k\}$ is exactly the set of vertices of $\gamma_p'$ with no outgoing edge. Let $U$ be a nontrivial connected component of $\gamma_p'$. Since $U$ is contained in a connected component of $\gamma_p$, in view of the structure of $\gamma_p$, either it contains a $2$-cycle with trees attached to each vertex of the $2$-cycle or it is a tree. In the first case, there is no vertex with $0$ out-degree and therefore no vertex in $V\setminus \{1,\ldots,k\}$. In the second case, the root of $U$ is the only vertex of $U$ with $0$ out-degree and it is therefore the only vertex in $V\setminus \{1,\ldots,k\}$.

  \begin{figure}[htbp]
    \centering
    \includegraphics[width=0.6\textwidth]{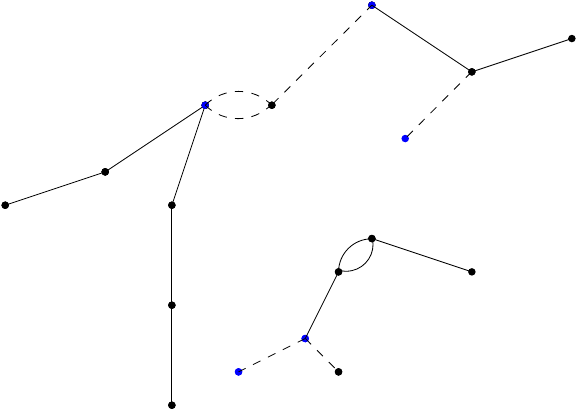}
    \caption{Nearest-neighbor graph and restricted nearest-neighbor graph}
    \label{fig:twotypes}
  \end{figure}
  Figure \ref{fig:twotypes} shows the points from $\{1, \ldots, k\}$ in black and the points from $V \setminus \{1, \ldots, k\}$ in blue. The solid lines represent the connected components of $\gamma_{p}'$, and the dashed lines represent the edges in $\gamma_p$ that are not part of $\gamma_p'$.

  We denote by $D(k_0,k_1)$ the set of such graphs with $k_0$ $2$-cycles and $k_1$ points in $V\setminus \{1,\ldots,k\}$ that do not form a trivial connected component. The number of nontrivial connected components is $k_0+k_1$. Each cycle uses two marked vertices and each rooted tree uses at least one, so $2k_0+k_1\le k$. The measurable cells $\Lambda^p\cap\{\gamma_p'=\gamma\}$ partition $\Lambda^p$ outside the excluded null set.

  Let us enumerate the graphs $\gamma$ in $D(k_0,k_1)$. Let $U_1,\ldots,U_{k_0+k_1}$ be the nontrivial connected components of such a $\gamma$.

  $\bullet $ Since connected components are unordered, we need to divide by $(k_0+k_1)!$.

  $\bullet$ Among these connected components we choose $k_0$ of them to be those that have a $2$-cycle. There are $\binom{k_0+k_1}{k_0}$ choices.

  $\bullet$ For these $k_0$ connected components, we choose the
  $2$-cycles. There are
  \begin{equation*}
    \frac{k!}{(k-2k_0)!2^{k_0}}
  \end{equation*}
  choices, corresponding to ordered collections of $k_0$ disjoint unordered pairs chosen from $k$ points.

  $\bullet$ On the $k_1$ trees, we need to fix the label of the root of the tree, which belongs to $\{k+1,\ldots,p\}$. There are at most $p^{k_1}$ choices.

  $\bullet$ To each of the remaining $k-2k_0$ points not in a $2$-cycle we need to assign a nearest-neighbor. Since there are $k_1$ points of $V\setminus \{1,\ldots,k\}$ in the graph, there are at most $k+k_1$ choices. Therefore, the total number of choices for these points is bounded by $(k+k_1)^{k-2k_0}\leq k^{k-2k_0}C^k$.

  Combining these and using Stirling's formula, we obtain the upper bound
  \begin{multline*}
    \log |D(k_0,k_1)|\le \log \frac{ k! p^{k_1} k^{k-2k_0} C^k (k_0+k_1)!}{(k-2k_0)! 2^{k_0} (k_0+k_1)! k_0! k_1!}
    \\ \leq Ck -k_0\log k_0-k_1\log k_1 +k\log k-(k-2k_0)\log (k-2k_0)+ (k-2k_0)\log k+k_1\log p.
  \end{multline*}

  \noindent
  \textbf{Step 2: integration.}
  Fix $\gamma\in D(k_0,k_1)$ and let $m=k+k_1$ be the number of vertices in its nontrivial connected components. We write
  \[
    I(\gamma):=\int_{\Lambda^p\cap\{\gamma_p'=\gamma\}}e^{-\beta F_k(Z_p)}
    \prod_{i=1}^k\left(1+\log\frac{\rr_1(z_i)}{\lambda}\right)^\delta\dd Z_p.
  \]
  We first integrate the $p-m$ vertices in trivial components. Conditional on the other coordinates, the distances $r(z_i)=|z_i-z_{\phi_1(i)}|/4$, $i\le k$, are fixed, and each remaining vertex avoids all the balls $B(z_i,4r(z_i))$. The balls $B(z_i,2r(z_i))$, $i\le k$, are disjoint. Moreover, since $r(z_i)\le\sqrt{2N}/4$, each intersection with $\Lambda$ contains a quarter of a disc of radius $\sqrt2r(z_i)$. The available volume is therefore at most $N-c\sum_{i=1}^k r(z_i)^2$.

  If $k\le p/4$, then $p-m\ge p-2k\ge N/4$, so this integration contributes at most
  \[
    N^{p-m}\exp\left(-c\sum_{i=1}^k r(z_i)^2\right).
  \]
  If $k>p/4$, packing gives $\sum_{i=1}^k r(z_i)^2\le CN$. By Lemma~\ref{lem:gaussian decoupling}, we insert the same exponential at cost $e^{CN}\le e^{Ck}$, and integrate the trivial components in volume $N^{p-m}$.

  In each nontrivial component, we retain the root or the larger-indexed vertex of its $2$-cycle, and replace every other coordinate by its position relative to its parent. As for \eqref{defphi}, this is an invertible map with determinant $1$. This leaves $k_0+k_1$ reference points, $k_0$ cycle variables and $k-2k_0$ tree-edge variables. We integrate the reference points in $\Lambda$ and extend each integral over a variable $u_j$ to $\R^2$.

  Polar coordinates, followed by $r=|u|/4$, give the annular integrals
  \[
    \int_\lambda^1 r^{1-\beta}\left(1+\log\frac r\lambda\right)^{2\delta}\dd r,
    \qquad
    \int_\lambda^1 r^{1-\beta/2}\left(1+\log\frac r\lambda\right)^\delta\dd r.
  \]
  The cores contribute $O_\beta(\lambda^{2-\beta})$ and $O_\beta(\lambda^{2-\beta/2})$. The Gaussian tails contribute $O_\beta(\ell^{2\delta})$ and $O_\beta(\ell^\delta)$. Evaluating the annular integrals gives
  \[
    I(\gamma)\le C_\beta^k N^{p-k+k_0}
    \exp\left(k\psi_\beta(k_0/k)\right),
  \]
  where
  \[
    \psi_\beta(x):=\begin{cases}
      0,&0<\beta<2,\\
      x\log\ell,&\beta=2,\\
      (2-\beta)x\log\lambda,&2<\beta<4,\\
      -2x\log\lambda+(1+\delta)(1-2x)\log\ell,&\beta=4,\\
      (2-\beta)x\log\lambda+(2-\beta/2)(1-2x)\log\lambda,&\beta>4.
    \end{cases}
  \]
  In particular, at $\beta=4$ each tree edge contributes $C\ell^{1+\delta}$, whereas each cycle contributes $C\lambda^{-2}$.

  \noindent
  \textbf{Step 3: optimization.}
  The count in Step 1 gives
  \[
    |D(k_0,k_1)|\le\frac{C^k k^k p^{k_1}}{k_0!k_1!}.
  \]
  Using $n!\ge(n/e)^n$, $p\le N$ and $0\log0=0$, we find
  \begin{align*}
    \log|D(k_0,k_1)|+\log I(\gamma)
    &\le p\log N+C_\beta k
    -(k-k_0-k_1)\log(N/k)+k\psi_\beta(k_0/k)\\
    &\le p\log N+C_\beta k
    +k_0\log(k/N)+k\psi_\beta(k_0/k),
  \end{align*}
  where the last inequality uses $2k_0+k_1\le k$. Summing over $k_0,k_1$ costs at most $2\log(k+1)\le Ck$. Hence the logarithm of the integral in the statement is bounded by
  \[
    p\log N+C_\beta k+
    k\max_{0\le x\le1/2}\left\{x\log(k/N)+\psi_\beta(x)\right\}.
  \]
  The function being maximized is affine. Its maximum is attained at $0$ or $1/2$ and equals $\nu_{\beta,k}/2$, which proves the result.
\end{proof}

We now prove a variant of Proposition \ref{prop:auxk} allowing two species of points appearing with distinct exponents in the integral.

\begin{prop}\label{prop:two species}Let $\beta >4$. Let $0<\lambda\le e^{-1}$ and let $k,k',p$ be nonnegative integers with $k+k'\le p\le N$ and $p\ge\max(2,N/2)$. Consider
  \begin{equation*}
    F_{k,k'}(z_1,\ldots,z_p):=\frac{1}{2}\sum_{i=k+1}^{k+k'} \log \rr_1(z_i)
  \end{equation*}
  and recall the definition of $F_k$ from \eqref{defFaux}. Then, for $\delta\in \{0,1\}$, we have
  \begin{multline}\label{eq:two species}
    \log \int_{\Lambda^p}\exp(-4F_k(Z_p)-\beta F_{k,k'}(Z_p))\prod_{i=1}^k \left(1+\log\frac{\rr_1(z_i)}{\lambda}\right)^\delta \dd z_1\ldots \dd z_p \\ \leq p\log N+C_\beta(k+k')+\frac{k}{2}\nu_{4,k}+\frac{k'}{2}\nu_{\beta,k'},
  \end{multline}
  where $\nu_{\beta,k}$ is as in \eqref{def:gbetak}, with the current $\delta$ in $\nu_{4,k}$.
\end{prop}

\begin{proof}
  If one marked group is empty, the result follows from Proposition~\ref{prop:auxk}. We assume $k,k'>0$ and set $t=k+k'$. Let $I$ denote the integral in the statement.

  \noindent
  \textbf{Step 1: counting graphs.}
  We use $\gamma_p'$ with marked vertices $\{1,\ldots,t\}$. Let $k_{0,1},k_{0,2},k_{0,3}$ count cycles in the first group, in the second group, and between the groups, respectively. Set $k_0=k_{0,1}+k_{0,2}+k_{0,3}$, and let $k_1$ count the nontrivial tree components. Then
  \[
    2k_{0,1}+k_{0,3}\le k,\qquad
    2k_{0,2}+k_{0,3}\le k',\qquad
    2k_0+k_1\le t.
  \]
  The first two inequalities count the vertices in each group. The third counts two vertices per cycle and at least one marked vertex per rooted tree.
  Disregarding the two species, Step 1 of Proposition~\ref{prop:auxk} bounds the number of graphs with these counts by
  \[
    \frac{C^t t^t p^{k_1}}{k_0!k_1!}.
  \]

  \noindent
  \textbf{Step 2: integration.}
  We repeat the Gaussian integration from Step 2 of Proposition~\ref{prop:auxk}, with $t$ marked vertices. A cycle within the first group contributes $C\lambda^{-2}$, and one within the second contributes $C_\beta\lambda^{2-\beta}$. For a mixed cycle, dropping the Gaussian and scaling $r=\lambda s$ bounds the radial integral by
  \[
    C\lambda^{-\beta/2}\left(1+\int_1^\infty s^{-1-\beta/2}(1+\log s)^\delta\dd s\right)
    \le C_\beta\lambda^{-\beta/2}.
  \]
  Each tree edge from the first group contributes $C\ell^{1+\delta}$, and one from the second contributes $C_\beta\lambda^{2-\beta/2}$. Thus, for a fixed graph $\gamma$, its integral is at most
  \begin{multline*}
    C_\beta^tN^{p-t+k_0}
    \lambda^{-2k_{0,1}+(2-\beta)k_{0,2}-(\beta/2)k_{0,3}}
    \ell^{(1+\delta)(k-2k_{0,1}-k_{0,3})}
    \lambda^{(2-\beta/2)(k'-2k_{0,2}-k_{0,3})}\\
    =C_\beta^tN^{p-t+k_0}\lambda^{-2k_0+(2-\beta/2)k'}
    \ell^{(1+\delta)(k-2k_{0,1}-k_{0,3})}.
  \end{multline*}

  \noindent
  \textbf{Step 3: optimization.}
  Combining the preceding estimates as in Step 3 of Proposition~\ref{prop:auxk}, and summing over the counts, gives
  \begin{multline*}
    \log I\le p\log N+C_\beta t
    +k(1+\delta)\log\ell+(2-\beta/2)k'\log\lambda\\
    +\max\left\{k_{0,1}A+k_{0,2}B+\frac{k_{0,3}}2(A+B)\right\},
  \end{multline*}
  where the maximum is over the counts above and
  \[
    A:=\log\frac{\lambda^{-2}t}{N}-2(1+\delta)\log\ell,
    \qquad B:=\log\frac{\lambda^{-2}t}{N}.
  \]
  Since $k_{0,1}+k_{0,3}/2\le k/2$ and $k_{0,2}+k_{0,3}/2\le k'/2$, the maximum is at most
  \[
    \frac{k}{2}(A\vee0)+\frac{k'}{2}(B\vee0).
  \]
  Consequently,
  \begin{equation*}
    \log I\le p\log N+C_\beta t
    +\frac{k}{2}\log\left(\frac{\lambda^{-2}t}{N}\vee\ell^{2(1+\delta)}\right)
    +\frac{k'}{2}\log\left(\frac{\lambda^{2-\beta}t}{N}\vee\lambda^{4-\beta}\right).
  \end{equation*}
  Replacing $t$ by $k$ and $k'$ in the respective logarithms costs at most
  \[
    \frac{k}{2}\log(t/k)+\frac{k'}{2}\log(t/k')\le\frac{t}{2}\log2.
  \]
  This proves \eqref{eq:two species}.
\end{proof}

An isolated $2$-cycle $\{i,j\}$ is twice isolated if $\phi_2(i),\phi_2(j)\in I^\pair$.

We now estimate the quadratic errors when all components are isolated pairs. The graph count \eqref{eq:logD}, proved below, gives the entropy and $C_0(p-p_0)$ terms in $f$.

\begin{prop}\label{prop:isolated}
  Let $\beta>2$ and $0<\lambda\le e^{-1}$, let $N\ge2$ and $p$ be integers with $2\le p\le N$, and let $C,C_0>0$. For each digraph $\gamma$ on $\{1,\ldots,2p\}$ made of $p$ disjoint $2$-cycles, let $\Omega_\gamma$ be the set of configurations $Z_{2p}\in\Lambda^{2p}$ whose nearest-neighbor graph, as defined in Section~\ref{sec2}, is $\gamma$. For each such $\gamma$ and each set $S$ of $2$-cycles of $\gamma$, let $V_S$ be the set of vertices of these $2$-cycles, and let
  \begin{equation*}
    I(\gamma,S):=\int_{\Omega_\gamma}\exp\Bigg(-\beta\Bigl(\F_\lambda^{\nn}(Z_{2p})-C\sum_{i\notin V_S}\Bigl(\frac{\rr_1(z_i)}{\rr_2(z_i)}\Bigr)^2+\sum_{i\in V_S}\Bigl(\log\frac{\rr_2(z_i)\wedge\rr_2(z_{\phi_1(i)})}{\rr_1(z_i)}-C\Bigr)\Bigr)\Bigg)\dd Z_{2p}\,.
  \end{equation*}
  For every integer $p_0$ with $0\le p_0\le p$, let
  \begin{equation*}
    f(p_0):=p\log p-p_0\log p_0-(p-p_0)\log(p-p_0)+C_0(p-p_0)+\max_{\gamma,S}\log I(\gamma,S)\,,
  \end{equation*}
  where the maximum runs over all such digraphs $\gamma$ and all sets $S$ of $p-p_0$ of their $2$-cycles. In this proposition, all constants, including those in $O(\cdot)$, may depend on $\beta$, $C$ and $C_0$. Then the following hold.
  \begin{enumerate}
    \item If $p\ge N/2$, then, with $\varepsilon_\lambda$ as in subsection~\ref{subsec:notation},
  \begin{equation}\label{eq:n neutral}
    \max_{0\le p_0\le p} f(p_0)
    \leq
    p\bigl( \log N +(2-\beta)\log\lambda+\log\mc{Z}_\beta+C_\beta\varepsilon_\lambda\bigr)+C_\beta\log(N+1),
  \end{equation}
  and there exists $c_\beta>0$ such that, for every $0\le p_0\le p$,
  \begin{equation}\label{eq:n neutral conv}
    f(p_0)-p\bigl( \log N +(2-\beta)\log\lambda+\log\mc{Z}_\beta\bigr) \leq -c_\beta p \Bigl(\frac{p-p_0}{p}\Bigr)^2+O(p\varepsilon_\lambda)+C_\beta\log(N+1).
  \end{equation}
    \item If $p<N/2$, then
  \begin{equation}\label{eq:crude}
    \max_{0\le p_0\le p}f(p_0)
    \leq
    p\bigl(\log N+(2-\beta)\log\lambda\bigr)+C_\beta N.
  \end{equation}
  \end{enumerate}
\end{prop}
\medskip

\begin{proof}
  We integrate first over the relative positions within the pairs and then over one representative of each pair.

  \smallskip
  \noindent
  \textbf{Step 1: comparison of the radii.}
  Fix $\gamma$ and $S$ with $|S|=p-p_0$. We number the pairs of $\gamma$ from $1$ to $p$, and we identify $S$ with the set of indices of its pairs. We retain one endpoint $w_j$ of the $j$-th pair and write the other endpoint as $w_j+u_j$. Let
  \[
    \rho_j:=\left(\frac14\min_{l\ne j}|w_j-w_l|\right)\vee\lambda\,.
  \]
  On $\Omega_\gamma$, the two endpoints of the $j$-th pair are mutual nearest neighbors. Hence both endpoints $z$ satisfy $\rr_1(z)=(|u_j|/4)\vee\lambda$, and we have $|u_j|\le4\rho_j$ and $(|u_j|/4)\vee\lambda\le\rho_j$.
  For either endpoint $z$ of the pair, let $v$ be its second nearest neighbor and let $w_l$ be the retained endpoint of the pair containing $v$. The distance from $v$ to $w_l$ is at most $|z-v|$, and the distance from $z$ to $w_j$ is also at most $|z-v|$. Hence $\rho_j\le3\rr_2(z)$, which gives
  \[
    \frac{\rr_1(z)}{\rr_2(z)}\le3\,\frac{(|u_j|/4)\vee\lambda}{\rho_j}\,.
  \]
  The change of variables $Z_{2p}\mapsto(U_p,W_p)$ has Jacobian $1$. The factor corresponding to a pair not in $S$ is bounded above in the integrand by
  $e^{\beta\g_\lambda(u_j)+C_\beta(((|u_j|/4)\vee\lambda)/\rho_j)^2}$.
  The factor corresponding to a pair in $S$ is bounded above by
  $C_\beta e^{\beta\g_\lambda(u_j)}(((|u_j|/4)\vee\lambda)/\rho_j)^{2\beta}$.
  Since $\rho_j$ does not depend on $u_j$, we may enlarge the domain of $u_j$ to $|u_j|\le4\rho_j$.

  \smallskip
  \noindent
  \textbf{Step 2: integration over the pairs.}
  For $\rho\ge\lambda$, the inequality
  $e^{C_\beta s^2}\le1+C_\beta' s^2$ for $0\le s\le1$, followed by radial integration at $2\lambda$ and $4\rho$, gives
  \begin{equation}\label{align1}
    \int_{|u|\le4\rho}e^{\beta\g_\lambda(u)}
    e^{C_\beta(((|u|/4)\vee\lambda)/\rho)^2}\dd u
    \le\mc Z_\beta\lambda^{2-\beta}\left(1+C_\beta'\varphi_\beta(\rho)\right),
  \end{equation}
  where $C_\beta'>0$ depends on the fixed coefficient in the exponential and, for $x\ge\lambda$,
  \begin{equation}\label{defvarphi}
    \varphi_\beta(x):=\begin{cases}
      (\lambda/x)^{\beta-2},&2<\beta<4,\\
      (\lambda/x)^2(1+\log(x/\lambda)),&\beta=4,\\
      (\lambda/x)^2,&\beta>4.
    \end{cases}
  \end{equation}
  The quadratic correction is bounded by
  \[
    C_\beta\rho^{-2}
    \left(\lambda^{4-\beta}+\int_{2\lambda}^{4\rho}r^{3-\beta}\dd r\right).
  \]
  For a pair in $S$, the radial decomposition gives the bound
  \[
    C_\beta\rho^{-2\beta}\left(\lambda^{\beta+2}
    +\int_{2\lambda}^{4\rho}r^{\beta+1}\dd r\right)
    \le C_\beta\rho^{2-\beta}.
  \]
  Hence
  \begin{equation}\label{align2}
    \int_{|u|\le4\rho}e^{\beta\g_\lambda(u)}
    \left(\frac{(|u|/4)\vee\lambda}{\rho}\right)^{2\beta}\dd u
    \le C_\beta\rho^{2-\beta}.
  \end{equation}
  For $A>0$ and $S\subset\{1,\ldots,p\}$, define
  \begin{equation}\label{pred3}
    \mathcal R_{A,S}:=
    \mc Z_\beta^p\lambda^{(2-\beta)p}\int_{\Lambda^p}
    \prod_{j\notin S}\left(1+A\varphi_\beta(\rho_j)\right)
    \prod_{j\in S}A\Bigl(\frac{\lambda}{\rho_j}\Bigr)^{\beta-2}\dd W_p.
  \end{equation}
  We obtain $I(\gamma,S)\le\mathcal R_{A,S}$ for a fixed $A$ depending only on $\beta,C$.

  \smallskip
  \noindent
  \textbf{Step 3: integration of the representatives.}
  We prove that \eqref{eq:n neutral}, \eqref{eq:n neutral conv} and \eqref{eq:crude} also hold with $f(p_0)$ replaced by
  \begin{equation}\label{eq:representative-bound}
    p\log p-p_0\log p_0-q\log q+C_0q+\log\mathcal R_{A,S},
    \qquad q=|S|=p-p_0.
  \end{equation}
  The constants depend only on $\beta,A,C_0$, uniformly in $S$. Assume first $p\ge N/2$. We expand the first product in \eqref{pred3}, with $k$ selected factors, and use Propositions~\ref{prop:auxk} and~\ref{prop:two species}. The sum over $0\le k\le p_0$ costs at most $\log(p+1)$ beyond its largest term.
  Set $\sigma_0=q/p$ and $\sigma_1=k/p$, and let $t=\sigma_0+\sigma_1$, so that $t\le1$. We use
  \[
    \log\binom{p_0}{k}\le k\log(p/k)+k,\qquad
    -(1-\sigma_0)\log(1-\sigma_0)\le \sigma_0,
  \]
  and
  \[
    -\sigma_0\log \sigma_0-\sigma_1\log \sigma_1\le t\log(2/t).
  \]
  The following expressions bound each term, divided by $p$, after subtracting
  \[
    \log N+(2-\beta)\log\lambda+\log\mc Z_\beta
  \]
  and adding the entropy and $C_0q$ from $f$.

  If $2<\beta<4$, set $b=\beta-2$. Proposition~\ref{prop:auxk} with parameter $2b$ and $k+q$ selected points gives
  \[
    t\left(b\log\lambda-\log t+C_\beta\right)
    +\frac{t}{2}\begin{cases}
      0,&0<b<1,\\
      \log(\ell t\vee1),&b=1,\\
      \log(\lambda^{2-2b}t\vee1),&1<b<2.
    \end{cases}
  \]
  Here we used $p/N\le1$ in \eqref{def:gbetak}.
  If $\beta=4$, we bound the factors indexed by $S$ by adding
  $1+\log(\rho_j/\lambda)\ge1$. Proposition~\ref{prop:auxk} with parameter $4$, $\delta=1$ and $k+q$ selected points then gives
  \[
    t\left(2\log\lambda-\log t+C_\beta\right)
    +\frac{t}{2}\log(\lambda^{-2}t\vee\ell^4).
  \]
  Finally, if $\beta>4$, set $b=\beta-2>2$. Proposition~\ref{prop:two species} with parameter $2b$ and $\delta=0$ gives
  \begin{multline*}
    \sigma_0\left(b\log\lambda-\log \sigma_0+C_\beta\right)
    +\frac{\sigma_0}{2}\log(\lambda^{2-2b}\sigma_0\vee\lambda^{4-2b})\\
    +\sigma_1\left(2\log\lambda-\log \sigma_1+C_\beta\right)
    +\frac{\sigma_1}{2}\log(\lambda^{-2}\sigma_1\vee\ell^2).
  \end{multline*}
  Empty selected groups contribute no term.

  \smallskip
  \noindent
  \textbf{Step 4: optimization.}
  We apply \eqref{eq:sup-inequality} to each branch of the maxima above.
  For $2<\beta<4$, the branch with no extra logarithm gives
  $C_\beta\lambda^{\beta-2}$. The other branch gives
  $C_\beta\lambda^2\ell$ at $\beta=3$ and
  $C_\beta\lambda^2$ for $3<\beta<4$, both bounded by
  $C_\beta\lambda^{\beta-2}$.
  For $\beta=4$, the two branches give
  $C_\beta\lambda^2$ and $C_\beta\lambda^2\ell^2$.
  For $\beta>4$, the $\sigma_0$ terms give $C_\beta\lambda^2$, and the $\sigma_1$ terms give $C_\beta\lambda^2\ell$.
  This proves \eqref{eq:n neutral}.
  Adding $\sigma_0^2\le t$ to the expressions in Step 3 only changes $C_\beta$.
  The same optimization therefore proves \eqref{eq:n neutral conv}.

  For $p<N/2$, each factor in \eqref{pred3} is at most $1+A$, since $\varphi_\beta\le1$ on $[\lambda,\infty)$ and $\rho_j\ge\lambda$. Hence the representative integral is at most $\mc Z_\beta^p\lambda^{(2-\beta)p}N^p(1+A)^p$. The entropy is at most $p\log2$, which proves \eqref{eq:crude} for \eqref{eq:representative-bound} and hence for $f$.
\end{proof}

\subsection{Upper bound below the critical temperature}
We may now obtain the main upper bound.

Let us denote by $D_{2N,K,p,\widehat p_0}$ the set of nearest-neighbor graphs of configurations $Z_{2N}$, marked at all their twice-isolated $2$-cycles, with $K$ connected components, $p$ marked cycles and $\widehat p_0$ neutral marked cycles. A marked graph is the pair consisting of the nearest-neighbor graph and its set of twice-isolated cycles. The condition $\gamma_{2N}=\gamma$ includes these marks, and the measurable cells $\Omega_\gamma:=\Lambda^{2N}\cap\{\gamma_{2N}=\gamma\}$ partition $\Lambda^{2N}$ outside the excluded null set.

\begin{prop}[Upper bound]\label{prop:upper bound}
  Fix $\beta>2$. The bounds below hold uniformly in $N\ge2$ for $0<\lambda<\lambda_\beta$, where $\lambda_\beta\in(0,e^{-2}]$ depends only on $\beta$.
  Let $\omega_\lambda$ be as in \eqref{eq:defgammal}.
  We have
  \begin{equation}\label{ubz}
    \log \Z\leq 2N\log N+N\Big((2-\beta)\log \lambda+\log \mc{Z}_{\beta}-1
    +O_\beta(\omega_\lambda)\Big)+O_\beta(\log(N+1))\,,
  \end{equation}where $\mc{Z}_\beta$ is the constant defined in \eqref{def:Cbeta}. Moreover, for each $p_0\in\{0,\ldots,N\}$, let
  \begin{equation*}
    \mc A(p_0):=\{|I^\dip|=2p_0\}.
  \end{equation*}
  There exists $c_\beta>0$ such that
  \begin{multline}\label{eq:neutral bound}
    \log\int_{\Lambda^{2N}\cap\mc A(p_0)}e^{-\beta \F_\lambda(X_N,Y_N)}\dd X_N \dd Y_N \leq 2N\log N +N(2-\beta)\log \lambda\\
    +N(\log \mc{Z}_\beta-1)-c_\beta N \Bigl(\frac{p_0-N}{N}\Bigr)^2+O_\beta(N\omega_\lambda)+C_\beta\log(N+1).
  \end{multline}

  Let $D_R$ count the isolated dipoles of length greater than $R$. For $R\ge2\lambda$ and $n\in\{1,\ldots,N\}$, let $x_0=n/N$ and $\alpha=\mu_\beta([R/\lambda,+\infty))$. There exists $C_\beta>0$ such that, for $x_0\ge C_\beta\alpha$, we have
  \begin{multline}\label{eq:u bound}
    \log\int_{\Lambda^{2N}\cap\{D_R\ge n\}}e^{-\beta\F_\lambda}\dd Z_{2N}
    \le2N\log N+N\bigl((2-\beta)\log\lambda+\log\mc Z_\beta-1\bigr)\\
    -c_\beta Nx_0\log(x_0/\alpha)+O_\beta(N\omega_\lambda)+C_\beta\log(N+1).
  \end{multline}
\end{prop}

\medskip

\begin{proof} Again, we split the phase-space according to the nearest neighbor graph $\gamma_{2N}$ of the points.

  \noindent
  \textbf{Step 1: enumeration of graphs.} Let us enumerate $D_{2N,K,p,\widehat p_0}$, defined just above the proposition. Let $\gamma\in D_{2N,K,p,\widehat p_0}$. Consider the ordered connected components of $\gamma$.

  $\bullet$ One must first divide by $K!$ to account for the fact that connected components are unordered.

  $\bullet$ We then select the connected components that are a twice isolated $2$-cycle (so that there are $p$ of them), a neutral twice isolated $2$-cycle (so that there are $\widehat p_0$ of them). This gives
  \begin{equation*}
    \frac{K!}{\widehat p_0!(p-\widehat p_0)!(K-p)!}
  \end{equation*}
  possible choices.

  $\bullet$ Labelling the vertices of the $\widehat p_0$ twice isolated neutral $2$-cycles gives
  \begin{equation*}
    \frac{(N!)^2}{((N-\widehat p_0)!)^2}=\frac{(2N)!}{2^{2\widehat p_0}(2N-2\widehat p_0)!}e^{O(\log N)}
  \end{equation*}
  possible choices (after using Stirling's formula).

  $\bullet$ The number of ways to label the other $2$-cycles is bounded by the number of ways to form $K-\widehat p_0$ pairs with $2N-2\widehat p_0$ points, which is
  \begin{equation*}
    \frac{(2N-2\widehat p_0)!}{(2N-2K)!2^{K-\widehat p_0}}.
  \end{equation*}

  $\bullet$ Finally, assigning a nearest-neighbor to each vertex not in a $2$-cycle gives at most $(2N)^{2N-2K}$ possible choices.

  Combining the above we obtain
  \begin{equation*}
    |D_{2N,K,p,\widehat p_0}|\leq \frac{(2N)!(2N)^{2N-2K}}{(2N-2K)!2^{K+\widehat p_0} \widehat p_0!(p-\widehat p_0)!(K-p)! }e^{O(\log N)}.
  \end{equation*}
  Hence, using Stirling's formula,
  \begin{multline*}
    \log |D_{2N,K,p,\widehat p_0}|\le 2N\log (2N)-2N+(2N-2K)\log (2N)-(2N-2K)\log (2N-2K)+2N-2K\\-(K+\widehat p_0)\log 2-\widehat p_0\log \widehat p_0+\widehat p_0-(p-\widehat p_0)\log (p-\widehat p_0)+p-\widehat p_0-(K-p)\log(K-p)+K-p+O(\log N)\\
    =2N\log (2N)-(2N-2K)\log \frac{N-K}{N}-(K+\widehat p_0)\log 2\\ -\widehat p_0\log \widehat p_0-(p-\widehat p_0)\log (p-\widehat p_0)-(K-p)\log(K-p)-K+O(\log N).
  \end{multline*}

  Let us denote by $m$ the total number of isolated $2$-cycles. For each of the $m-p$ isolated pairs that are not twice isolated, choose a vertex whose second nearest neighbor belongs to $(I^\pair)^c$. By Lemma~\ref{lem:in-degree}, at most six such pairs map to each exterior point. Hence
  \begin{equation*}
    m-p\le6|(I^\pair)^c|=12(N-m).
  \end{equation*}
  Each of the $K-m$ nonisolated components contains a tree vertex, so
  \begin{equation*}
    K-m\le2(N-K).
  \end{equation*}
  Thus $|(I^\pair)^c|=2(N-K)+2(K-m)\le6(N-K)$, which gives $m-p\le M_0(N-K)$ with $M_0=36$, and
  \begin{equation}\label{eq:K-p}
    K-p\le(M_0+2)(N-K).
  \end{equation}
  Using this bound in the graph count, we obtain, for a constant $C_0>0$,
  \begin{multline}\label{eq:logD}
    \log |D_{2N,K,p,\widehat p_0}|\leq 2N\log N- (2N-2K)\log \frac{N-K}{N} -\widehat p_0\log \widehat p_0-(p-\widehat p_0)\log (p-\widehat p_0)\\ -(K-p)\log(K-p)-N+O(\log(N+1))+C_0(N-K)+C_0(p-\widehat p_0).
  \end{multline}

  \noindent
  \textbf{Step 2: Inserting the energy lower bound and splitting the integral.}

  For each $\widehat p_0\leq p\leq K\leq N$, let us define
  \begin{equation*}
    Q(K,p,\widehat p_0):=\max_{\gamma\in D_{2N,K,p,\widehat p_0}}\log \int_{\Omega_\gamma}e^{-\beta \F_\lambda(X_N,Y_N)}\dd X_N \dd Y_N.
  \end{equation*}

  To bound $Q(K,p,\widehat p_0)$ from above, we start by inserting the lower bound \eqref{minoF02 bis} into the definition of $Q(K,p,\widehat p_0)$, which gives
  \begin{equation}\label{gre1}
  \begin{aligned}
    Q(K,p,\widehat p_0)&\leq \max_{\gamma\in D_{2N,K,p,\widehat p_0}}\log \int_{\Lambda^{2N}\cap\{\gamma_{2N}=\gamma\}} \exp\Bigg(-\beta \Bigg(\F_\lambda^{\nn}(X_N, Y_N) \\
    &\qquad+\sum_{\substack{i\in I^{\pair}\backslash I^{\dip}\\ \phi_2(i) \in I^\pair, \phi_2(\phi_1(i)) \in I^\pair}}\left( \log
    \frac{\min(\rr_2(z_i), \rr_2(z_{\phi_1(i)}))}{\rr_1(z_i)}-C\right)\\
    &\qquad-C \sum_{\substack{i \in I^{\dip} ,\phi_2(i)\in I^\pair\\ \phi_2(\phi_1(i))\in I^\pair }}
    \left( \frac{\rr_1(z_i)}{\rr_2(z_i)}\right)^2\Bigg)\Bigg)\dd X_N\dd Y_N+C_\beta(N-K),
  \end{aligned}
  \end{equation}
  where we recall that $I^\pair$ and $ I^\dip$ are defined in \eqref{defIdip}.

  Fix $\gamma\in D_{2N,K,p,\widehat p_0}$. Let $I=I^{\pair}$, with $|I|=2m$, and let $I'$ be the set of vertices of the $p-\widehat p_0$ marked same-sign pairs. For $m\ge2$, define
  \begin{multline*}
    F^{(1)}((z_i)_{i\in I}):=-\frac12\sum_{i\in I}\g_\lambda(z_i-z_{\phi_1(i)})
    +\sum_{i\in I'}\left(\log\frac{\rr_2(z_i)\wedge\rr_2(z_{\phi_1(i)})}{\rr_1(z_i)}-C\right)\\
    -C\sum_{i\in I\setminus I'}\left(\frac{\rr_1(z_i)}{\rr_2(z_i)}\right)^2,
  \end{multline*}
  where the distances are computed within $I$ and $I'$ remains fixed. Thus the $m-p$ unmarked pairs receive the quadratic correction, irrespective of their signs. On the complement, define
  \begin{equation*}
    F^{(2)}((z_i)_{i\in I^c}):=-\frac12\sum_{i\in I^c}\g_\lambda(z_i-z_{\phi_1(i)}).
  \end{equation*}
  Both $I$ and $I^c$ are unions of whole components, so each vertex retains its nearest neighbor and $\rr_1$. At marked vertices, the second neighbor also lies in $I$, so $\rr_2$ is unchanged there. If $m\le1$, there are no marked pairs, since a second neighbor in $I$ would require another isolated pair. We then use only the nearest-neighbor term in $F^{(1)}$. The additional quadratic terms are nonpositive, so Corollary~\ref{coromino} gives
  \[
    \F_\lambda\ge F^{(1)}+F^{(2)}-C(N-K).
  \]
  Let $\gamma^1$ and $\gamma^2$ be the restrictions of the underlying nearest-neighbor graph to $I$ and $I^c$. Enlarging the integration domain now separates the variables:
  \begin{multline}\label{gred3}
    Q(K,p,\widehat p_0)\le\max_{\gamma\in D_{2N,K,p,\widehat p_0}}\Bigg(
    \log\int_{\Lambda^{2m}\cap\{\gamma_{2m}=\gamma^1\}}e^{-\beta F^{(1)}}\prod_{i\in I}\dd z_i\\
    +\log\int_{\Lambda^{2N-2m}\cap\{\gamma_{2N-2m}=\gamma^2\}}e^{-\beta F^{(2)}}\prod_{i\in I^c}\dd z_i+C_\beta(N-K)\Bigg).
  \end{multline}

  For $p\ge N/2$, Proposition~\ref{prop:isolated} applies with $m$ pairs and $S$ the set of the $p-\widehat p_0$ marked pairs, so that $\widehat p_0+m-p$ pairs are assigned the quadratic correction. For $q=p-\widehat p_0$, the function $s\mapsto s\log s-(s-q)\log(s-q)$ has derivative $\log(s/(s-q))\ge0$ for $s>q$ and is continuous at $q$. Hence
  \[
    m\log m-(\widehat p_0+m-p)\log(\widehat p_0+m-p)\ge p\log p-\widehat p_0\log \widehat p_0,
    \qquad
    \frac{(p-\widehat p_0)^2}{m}\ge\frac{(p-\widehat p_0)^2}{N}.
  \]
  Using \eqref{eq:representative-bound} with the constant $C_0$ in \eqref{eq:logD}, we obtain
  \begin{multline}\label{eq:F(1)}
    \log\int_{\Lambda^{2m}\cap\{\gamma_{2m}=\gamma^1\}}e^{-\beta F^{(1)}}\prod_{i\in I}\dd z_i
    +p\log p-\widehat p_0\log \widehat p_0-(p-\widehat p_0)\log(p-\widehat p_0)+C_0(p-\widehat p_0)\\
    \le m\bigl(\log N+(2-\beta)\log\lambda+\log\mc Z_\beta\bigr)
    +O_\beta(N\omega_\lambda)-c_\beta N\left(\frac{p-\widehat p_0}{N}\right)^2+C_\beta\log(N+1).
  \end{multline}
  Here $\omega_\lambda$ is as in \eqref{eq:defgammal}. The graph $\gamma^2$ has $K-m$ components and $2(N-m)$ vertices. For $m=N$, its integral is $1$ and contributes zero to \eqref{gred3}. For $m<N$, Lemma~\ref{lemma:nearest neigh} gives
  \begin{multline}\label{eq:F(2)}
    \log\int_{\Lambda^{2N-2m}\cap\{\gamma_{2N-2m}=\gamma^2\}}e^{-\beta F^{(2)}}\prod_{i\in I^c}\dd z_i
    \le(K-m)\bigl(\log N+(2-\beta)\log\lambda+\log\mc Z_\beta\bigr)\\
    +(2N-2K)\Bigg(\left(1-\frac\beta4\right)\log\frac{N}{N-K}\indic_{\beta<4}
    +\log\left(\ell+\log\frac{N}{N-K}\right)\indic_{\beta=4}\\
    +\left(2-\frac\beta2\right)\log\lambda\indic_{\beta>4}\Bigg)+C_\beta(N-K).
  \end{multline}

  Set $\delta=(N-K)/N$ and $\alpha_0=M_0+2$. By \eqref{eq:K-p},
  \[
    -\frac pN\log\frac pN-\frac{K-p}{N}\log\frac{K-p}{N}
    \le\alpha_0\delta\log(1/\delta)+C\delta.
  \]
  Combining \eqref{eq:logD}, \eqref{gred3}, \eqref{eq:F(1)} and \eqref{eq:F(2)}, we obtain, for $p\ge N/2$,
  \begin{multline}\label{eq:2cases p}
    \log|D_{2N,K,p,\widehat p_0}|+Q(K,p,\widehat p_0)
    \le2N\log N+N\bigl((2-\beta)\log\lambda+\log\mc Z_\beta-1\bigr)\\
    +N\chi_\beta(\delta)-c_\beta\frac{(p-\widehat p_0)^2}{N}
    +C_\beta N\omega_\lambda+C_\beta\log(N+1),
  \end{multline}
  where $\chi_\beta(0)=0$ and, for $0<\delta\le1$,
  \begin{equation}\label{defchi}
    \chi_\beta(\delta):=\begin{cases}
      \delta\bigl[-(\beta-2)\ell+(\alpha_0+4-\beta/2)\log(1/\delta)+C_\beta\bigr],&2<\beta<4,\\
      \delta\bigl[-2\ell+2\log\ell+(\alpha_0+4)\log(1/\delta)+C_\beta\bigr],&\beta=4,\\
      \delta\bigl[-2\ell+(\alpha_0+2)\log(1/\delta)+C_\beta\bigr],&\beta>4.
    \end{cases}
  \end{equation}
  For $\beta=4$, we used $\ell+\log(1/\delta)\le\ell/\delta$, which follows from $\ell\ge1$. For $p<N/2$, we use \eqref{eq:crude} if $2\le m<N/2$ and \eqref{eq:n neutral} if $m\ge N/2$. For $m\le1$, direct integration gives the same crude bound. This gives \eqref{eq:2cases p} with $C_\beta N$ in place of $C_\beta N\omega_\lambda$, after absorbing the quadratic term into this error.

  \noindent
  \textbf{Step 3: optimization.} Applying \eqref{eq:sup-inequality} to \eqref{defchi}, we obtain
  \[
    \sup_{0\le\delta\le1}\chi_\beta(\delta)\le C_\beta\begin{cases}
      \lambda^{\frac{2(\beta-2)}{12-\beta+2M_0}},&2<\beta<4,\\
      (\lambda\ell)^{\frac{2}{M_0+6}},&\beta=4,\\
      \lambda^{\frac{2}{M_0+4}},&\beta>4.
    \end{cases}
  \]
  By \eqref{eq:defgammal}, the right side is $C_\beta\omega_\lambda$. Applying the same estimate to $\chi_\beta(\delta)+\delta^2$, using $\delta^2\le\delta$, gives $\chi_\beta(\delta)\le-\delta^2+C_\beta\omega_\lambda$. Since
  \[
    1-\frac{\widehat p_0}{N}\le(1+\alpha_0)\delta+\frac{p-\widehat p_0}{N},
  \]
  we deduce from \eqref{eq:2cases p} that, for $p\ge N/2$,
  \begin{multline}\label{eq:quadra x3}
    \log|D_{2N,K,p,\widehat p_0}|+Q(K,p,\widehat p_0)
    \le2N\log N+N\bigl((2-\beta)\log\lambda+\log\mc Z_\beta-1\bigr)\\
    -c_\beta\frac{(N-\widehat p_0)^2}{N}+C_\beta N\omega_\lambda+C_\beta\log(N+1).
  \end{multline}
  For $p<N/2$, we have $\delta\ge[2(1+\alpha_0)]^{-1}$. In this range $\chi_\beta(\delta)\to-\infty$ uniformly as $\lambda\to0$, so \eqref{eq:quadra x3} also follows from the crude estimate above.

  Summing over $K,p,\widehat p_0$ costs at most $3\log(N+1)$, proving \eqref{ubz}. On $\mc A(p_0)$, the marked count satisfies
  \[
    \widehat p_0\le p_0\quad\Longrightarrow\quad
    -c_\beta\frac{(N-\widehat p_0)^2}{N}
    \le-c_\beta\frac{(N-p_0)^2}{N}.
  \]
  Summing \eqref{eq:quadra x3} over the classes compatible with $\mc A(p_0)$ proves \eqref{eq:neutral bound}. Dropping the nonnegative quadratic correction from the exponent in \eqref{gre1} gives
  \begin{multline}\label{eq:alt}
    \log\int_{\Lambda^{2N}} e^{-\beta G(Z_{2N})} \dd Z_{2N}\le2N\log N
    +N\bigl((2-\beta)\log\lambda+\log\mc{Z}_\beta-1\bigr)\\
    +O_\beta(N\omega_\lambda)+C_\beta\log(N+1),
  \end{multline}
  where
  \begin{equation*}
    G:=-\hal \sum_{ i=1}^{2N} \g_\lambda(z_i-z_{\phi_1(i)} )+ \sum_{\substack{i\in I^{\pair}\backslash I^{\dip}\\ \phi_2(i) \in I^\pair, \phi_2(\phi_1(i)) \in I^\pair}}\left( \log
    \frac{\rr_2(z_i) \wedge \rr_2(z_{\phi_1(i)})}{\rr_1(z_i)}-C\right).
  \end{equation*}

  \medskip

  \noindent
  \textbf{Step 4: control on the size of dipoles.}
  We use the separation \eqref{gred3}. In the proof of Proposition~\ref{prop:isolated}, the quadratic correction in each radial integral is bounded. Requiring $|u_j|>R$ for a pair $j\notin S$ therefore bounds its radial factor by
  \[
    C_\beta\int_{|u|\ge R}e^{\beta\g_\lambda(u)}\dd u
    =C_\beta\mc Z_\beta\lambda^{2-\beta}\mu_\beta([R/\lambda,+\infty)).
  \]
  The bound \eqref{align1} without a length restriction is at least $\mc Z_\beta\lambda^{2-\beta}$. Thus each prescribed long pair $j\notin S$ contributes a factor $C_\beta\mu_\beta([R/\lambda,+\infty))$ to $\mathcal R_{A,S}$ in \eqref{pred3}. The bounds \eqref{eq:representative-bound} therefore apply uniformly to every prescribed set of long pairs. This also covers the case $p<N/2$, because the proof of Proposition~\ref{prop:isolated} establishes \eqref{eq:n neutral} and \eqref{eq:crude} for \eqref{eq:representative-bound}. For $m\le1$, direct pair integration gives the same tail factor.

  For $m\ge2$, every isolated neutral pair lies in $I\setminus I'$ and receives the quadratic correction in $F^{(1)}$. On $\{D_R\ge n\}$, at least $n$ such pairs are long. Summing over their choices costs at most $\binom{N}{n}$. The graph sum is otherwise the same as in Steps 2 and 3, so
  \begin{multline*}
    \log\int_{\Lambda^{2N}\cap\{D_R\ge n\}}e^{-\beta\F_\lambda}\dd Z_{2N}
    \le2N\log N+N\bigl((2-\beta)\log\lambda+\log\mc Z_\beta-1\bigr)\\
    +\log\binom{N}{n}+n\log\bigl(C_\beta\mu_\beta([R/\lambda,+\infty))\bigr)
    +C_\beta N\omega_\lambda+C_\beta\log(N+1).
  \end{multline*}
  Let $x_0=n/N$ and $\alpha=\mu_\beta([R/\lambda,+\infty))$. If $x_0\ge C_\beta'\alpha$ and $C_\beta'$ is large enough, then
  \begin{equation*}
    \log\binom{N}{n}+n\log(C_\beta\alpha)
    \le-Nx_0\log\frac{x_0}{eC_\beta\alpha}
    \le-\frac12Nx_0\log\frac{x_0}{\alpha}.
  \end{equation*}
  This proves \eqref{eq:u bound}.
\end{proof}

\subsection{Upper bound at the critical temperature}

At $\beta=2$, the Gaussian-regularized mass $J_\lambda(1/\ell)$ of a neutral pair is of order $\ell$ by \eqref{eq:critical-B}, and the quadratic correction of Proposition~\ref{prop:mino2} changes the mass of a pair by $O(1)$ only, by \eqref{eq:critical-correction}. Hence, unlike in the proof of Proposition~\ref{prop:upper bound}, we do not need to single out the twice-isolated pairs, and we treat all isolated pairs alike.

\begin{prop}\label{prop:critical upper}
  There exist constants $\lambda_2\in(0,e^{-2}]$ and $C>0$ with the following property. Let $N\ge2$ be an integer, let $0<\lambda<\lambda_2$ and let $2\lambda\le R\le1$. Let $I^{\dip}$ be the set of particles in isolated dipoles, and let $D_R$ be the number of isolated dipoles of length larger than $R$. Then
  \begin{multline}\label{eq:critical upper}
    \log\int_{\Lambda^{2N}}\exp\Bigl(-2\F_\lambda+N-\frac12|I^{\dip}|+D_R\Bigr)\dd Z_{2N}
    \le2N\log N+N\bigl(\log(2\pi|\log\lambda|)-1\bigr)\\
    +CN\left(|\log\lambda|^{-\frac15}+\frac{\log(1/R)}{|\log\lambda|}\right)+C\log(N+1)\,.
  \end{multline}
\end{prop}

Since $N-\frac12|I^{\dip}|\ge0$ and $D_R\ge0$, taking $R=1$ bounds $\log\Z$ at $\beta=2$. The first of these terms yields \eqref{eq:b1} and the second yields \eqref{eq:critical lengths}, as shown in the proof of Theorem~\ref{theorem:th1}.

\begin{proof}
  We bound the integral on each cell of the nearest-neighbor graph, and then sum over the graphs.

  \smallskip
  \noindent
  \textbf{Step 1: energy on a fixed graph.}
  Fix an unmarked nearest-neighbor graph $\gamma$ with $K$ components, as in Section~\ref{sec2}, and write $\{\gamma_{2N}=\gamma\}$ for the set of configurations whose unmarked nearest-neighbor graph is $\gamma$. Let $m$ be the number of isolated pairs of $\gamma$, and let $p_0$ be the number of its isolated dipoles. The graph has $2(N-K)$ tree vertices, and each of the $K-m$ other components contains at least one of them, so $K-m\le2(N-K)$.

  In each isolated pair, let $w_j$ be the endpoint with the larger index, and write the other endpoint as $w_j+u_j$. We call retained points the points $w_j$ and all points of the other components. Let
  \[
    \rho_j:=\left(\frac14\min_{y\ne w_j}|w_j-y|\right)\vee\lambda\,,
  \]
  where the minimum runs over the retained points $y$. For both endpoints $z$ of the pair, we have $\rr_1(z)=(|u_j|/4)\vee\lambda$ and
  \[
    |u_j|\le4\rho_j\,,\qquad \rr_1(z)\le\rho_j\qquad\text{and}\qquad \rho_j\le3\rr_2(z)\,.
  \]
  The first two bounds follow from the fact that $w_j+u_j$ is the nearest neighbor of $w_j$. For the third bound, let $v$ be the second nearest neighbor of $z$, so that $|z-w_j|\le|z-v|$. Either $v$ is retained, or $v$ is the non-retained endpoint of another isolated pair, whose retained endpoint lies within distance $|z-v|$ of $v$. In both cases, a retained point other than $w_j$ lies within distance $3|z-v|$ of $w_j$. Hence $\rho_j\le(3|z-v|/4)\vee\lambda\le3\rr_2(z)$.

  We insert these bounds into Proposition~\ref{prop:mino2}, using \eqref{approxgeta}. The error of an isolated dipole is at most $C(\rr_1(w_j)/\rho_j)^2$. The correction of a same-sign isolated pair is at least $2\log(\rho_j/\rr_1(w_j))-C$. The other components cost at most $C(N-K)$, as in the proof of Corollary~\ref{coromino}. Since $(N-K)+(m-p_0)\le N-p_0$, we obtain
  \[
    \F_\lambda\ge\F_\lambda^{\nn}
    -C\sum_{\text{isolated dipoles}}\Bigl(\frac{\rr_1(w_j)}{\rho_j}\Bigr)^2
    +2\sum_{\text{same-sign isolated pairs}}\log\frac{\rho_j}{\rr_1(w_j)}
    -C(N-p_0)\,.
  \]

  \smallskip
  \noindent
  \textbf{Step 2: integration on a fixed graph.}
  We use the Gunson--Panta variables \eqref{defphi}. As in Step~2 of the proof of Lemma~\ref{lemma:nearest neigh}, the cycle and tree variables satisfy the packing constraint $\sum_i|u_i|^2\le36N/\pi$. We bound its indicator by the product of the indicators of the same constraint on the cycle variables and on the tree variables. The bound \eqref{eq:gaussian-decoupling} of Lemma~\ref{lem:gaussian decoupling}, with $s=36N/\pi$ and $a=1/\ell$, replaces the first indicator by the factor $\exp(-\sum_{\text{cycles}}|u_j|^2/\ell)$, at cost $e^{CN/\ell}$. The tilt $e^{N-\frac12|I^{\dip}|}=e^{N-p_0}$ is absorbed into $e^{C(N-p_0)}$.

  We first integrate the variables $u_j$ of the isolated pairs, with the retained points fixed. Then $\rho_j$ is fixed and $|u_j|\le4\rho_j$. By Step~1, an isolated dipole carries the weight
  \[
    e^{2\g_\lambda(u_j)-|u_j|^2/\ell}e^{C(\rr_1(w_j)/\rho_j)^2}\bigl(1+(e-1)\indic_{|u_j|>R}\bigr)\,,
  \]
  where the last factor comes from $e^{D_R}$. With $x=\rr_1(w_j)/\rho_j$, we write $e^{Cx^2}=1+(e^{Cx^2}-1)$. In the second term, we bound $e^{-|u_j|^2/\ell}$ by $1$ and the last factor by $e$. Since $\rr_1(w_j)=(|u_j|/4)\vee\lambda$, the bound \eqref{eq:critical-correction} then shows that the integral of this weight over $|u_j|\le4\rho_j$ is at most
  \[
    T_R:=J_\lambda(1/\ell)+(e-1)J_\lambda(1/\ell;R)+C\,.
  \]
  A same-sign isolated pair carries the weight $e^{2\g_\lambda(u_j)}(\rr_1(w_j)/\rho_j)^4$, up to a constant factor, and its integral is at most $C$ by \eqref{eq:critical-correction}.

  Each of the $K-m$ other cycles contributes $J_\lambda(1/\ell)$, and the $K$ reference points contribute $N^K$. The tree weight is $e^{\g_\lambda(u)}\le C|u|^{-1}$, by \eqref{approxgeta}. Lemma~\ref{lemdirichlet} with $\alpha=1/2$, as in \eqref{eq:bound dir}, bounds the integral over the $2(N-K)$ tree variables by $(CN/(N-K))^{N-K}$, which we read as $1$ when $K=N$. Collecting these factors, we obtain
  \begin{multline}\label{eq:critical-fixed-graph}
    \int_{\Lambda^{2N}\cap\{\gamma_{2N}=\gamma\}}e^{-2\F_\lambda+N-p_0+D_R}\dd Z_{2N}\\
    \le e^{C(N-p_0)+CN/\ell}N^KT_R^{p_0}J_\lambda(1/\ell)^{K-m}\Bigl(\frac{CN}{N-K}\Bigr)^{N-K}\,.
  \end{multline}

  \smallskip
  \noindent
  \textbf{Step 3: counting and optimization.}
  We now sum over the graphs with given $K$, $m$ and $p_0$. To build such a graph, we choose the $p_0$ isolated dipoles and the $2(N-K)$ tree vertices. We pair the other $2(K-p_0)$ vertices into $m-p_0$ same-sign isolated pairs and $K-m$ other cycles, and we assign a parent to each tree vertex. Hence the number of such graphs is at most
  \[
    \frac{(N!)^2}{((N-p_0)!)^2p_0!}\cdot\frac{(2N-2p_0)!}{2^{K-p_0}(2N-2K)!(m-p_0)!(K-m)!}\cdot(2N)^{2(N-K)}\,.
  \]
  By \eqref{eq:critical-B}, we have $J_\lambda(1/\ell)\ge\pi\ell$ for $\lambda$ small enough. Let $\delta=(N-K)/N$ and $\zeta=(m-p_0)/N$. We multiply the count by the right side of \eqref{eq:critical-fixed-graph}, and we use, for every integer $n\ge0$, the bounds $(2n)!\le4^n(n!)^2$ and $n!\ge(n/e)^n$, together with $N!/p_0!\le N^{N-p_0}$ and $K-m\le2(N-K)$. This bounds the product by
  \begin{multline*}
    N!N^NJ_\lambda(1/\ell)^N\Bigl(\frac{T_R}{J_\lambda(1/\ell)}\Bigr)^{p_0}e^{CN/\ell}\\
    \times\exp\Bigl(N\delta\bigl(5\log(1/\delta)-\log\ell+C\bigr)+N\zeta\bigl(\log(1/\zeta)-\log\ell+C\bigr)\Bigr)\,.
  \end{multline*}
  The terms $-\log\ell$ account for the $(N-K)+(m-p_0)$ factors $J_\lambda(1/\ell)\ge\pi\ell$ that are missing compared with $J_\lambda(1/\ell)^N$. The term $5\delta\log(1/\delta)$ collects $2\delta\log(1/\delta)$ from $(2N)^{2(N-K)}/(2N-2K)!$, $\delta\log(1/\delta)$ from the tree integral, and $2\delta\log(1/\delta)$ from $N^{K-m}/(K-m)!$.

  By \eqref{eq:critical-B} and \eqref{eq:critical-BR}, we have
  \[
    \log\frac{T_R}{J_\lambda(1/\ell)}\le\frac{C(1+J_\lambda(1/\ell;R))}{J_\lambda(1/\ell)}\le C\frac{1+\log\ell+\log(1/R)}{\ell}\,.
  \]
  We apply \eqref{eq:sup-inequality} with $(A,B)=(C-\log\ell,5)$ and with $(A,B)=(C-\log\ell,1)$. It bounds the exponent in the second line by $CN\ell^{-1/5}+CN/\ell$. Finally, we use $\log N!\le N\log N-N+C\log(N+1)$, the bound on $\log J_\lambda(1/\ell)$ in \eqref{eq:critical-B} and $(1+\log\ell)/\ell\le C\ell^{-1/5}$, and we sum over the at most $(N+1)^3$ triples $(K,m,p_0)$. This proves \eqref{eq:critical upper}.
\end{proof}

\section{Free energy lower bound and proof of Theorem \ref{theorem:th1}}\label{sec4}
We now derive a lower bound on the partition function matching \eqref{ubz} and, at $\beta=2$, \eqref{eq:critical upper}. For that, we use a method inspired by previous work on the one-component plasma and on Ginzburg-Landau, starting in \cite{gl13,SS2d}, which allows us, thanks to the electric formulation of the energy, to compute the interaction additively in terms of electric potentials defined in disjoint subregions of the space.

\begin{prop}\label{prop:lower bound}
  The following estimates hold as $\lambda\to0$, uniformly in $N\ge2$.
  For $\beta \in (2,+\infty)$, with $\varepsilon_\lambda$ as in subsection~\ref{subsec:notation}, we have
  \begin{equation}\label{lbz}
    \log \Z
    \ge 2N \log N + N(2-\beta) \log \lambda+N(\log \mc{Z}_{\beta}-1)+O_\beta(N\varepsilon_\lambda).
  \end{equation}
  For $\beta=2$, we have
  \begin{equation}\label{eq:lbz2}
    \log \Z\geq 2N\log N+N\log|\log \lambda|+N(\log \mc{Z}_\beta-1)-\frac{CN}{|\log\lambda|^{1/5}}.
  \end{equation}
\end{prop}

\begin{proof}
  \textbf{Step 1: bounding the energy from above.}
  We restrict $X_N$ to
  \[
    H:=\left\{X_N\in\Lambda^N:\begin{array}{l}
      |x_i-x_j|\ge8\lambda\text{ for }i\ne j,\\
      \dist(x_i,\partial\Lambda)\ge4\lambda\text{ for }1\le i\le N
    \end{array}\right\},
  \]
  and $y_i\in B(x_i,\hal\varrho(x_i))$, where
  \[
    \varrho(x_i):=\min\left\{\hal\min_{j\ne i}|x_j-x_i|,\dist(x_i,\partial\Lambda),1\right\}.
  \]
  For $\lambda$ small, we have $\varrho(x_i)\ge4\lambda$. The balls $B(x_i,\varrho(x_i))$ are disjoint and contained in $\Lambda$, and both smeared charges of the $i$-th pair are supported in $B(x_i,\frac34\varrho(x_i))$.
  For such configurations let us bound the energy from above. First we recall \eqref{eq:rewrF}.

  We then let, for each $i$, $\psi_i$ solve
  \[\left\{ \begin{array}{ll}
    -\Delta \psi_i= 2\pi (\delta_{x_i}^{(\lambda)} -\delta_{y_i}^{(\lambda)} )& \text{in} \ B(x_i, \varrho(x_i))\\
    \frac{\partial \psi_i}{\partial \nu}= 0 & \text{on} \ \partial B(x_i, \varrho(x_i))\end{array}\right.,
  \]
  where $\nu$ is the outer unit normal to the ball.

  The Neumann problem is compatible because the source has total mass zero.

  We then define a global ``electric field" $E$ by pasting together the electric fields defined over these disjoint balls:
  \[E:= \sum_{i=1}^N \indic_{B(x_i, \varrho(x_i))} \nab \psi_i.\]
  Thanks to the crucial choice of zero Neumann boundary conditions on the boundary of the disjoint balls, this vector field satisfies
  \begin{equation}\label{divE}
    -\div E= 2\pi \left(\sum_{i=1}^N \bigl(\delta_{x_i}^{(\lambda)} -\delta_{y_i}^{(\lambda)}\bigr)\right) = -\Delta h_\lambda \quad \text{in} \ \R^2,\end{equation}
  where $h_\lambda$ is the electric potential of the configuration as in \eqref{eq:heta}.
  We now bound the energy from above by the $L^2$ norm of $E$, using that $\nab h_\lambda$ is the $L^2$ projection of $E$ onto gradients. Expanding the square, we have
  \[
    \int_{\R^2} |E|^2= \int_{\R^2} |\nab h_\lambda|^2 + \int_{\R^2}|E- \nab h_\lambda|^2 + 2 \int_{\R^2} (E-\nab h_\lambda) \cdot \nab h_\lambda\,.
  \]
  Both fields lie in $L^2$. Since $E$ has compact support and neutrality gives $h_\lambda=O(|x|^{-1})$ and $\nab h_\lambda=O(|x|^{-2})$, integration by parts with a cutoff at infinity makes the last term vanish by \eqref{divE}. It thus follows that
  \begin{equation}\label{hui}\int_{\R^2} |\nab h_\lambda|^2 \leq \sum_{i=1}^N \int_{B(x_i, \varrho(x_i)) } |\nab \psi_i|^2, \end{equation} and thus we can reduce the computation to a sum over the disjoint balls.
  We next bound the right-hand side.
  First we let $v_i:= \psi_i - (\g*\delta_{x_i}^{(\lambda)} - \g*\delta_{y_i}^{(\lambda)})$. It solves
  \[\left\{ \begin{array}{ll}
    -\Delta v_i= 0 & \text{in} \ B(x_i, \varrho(x_i))\\
    \frac{\partial v_i}{\partial \nu}= \left( \frac{(x-x_i)}{|x-x_i|^2} - \frac{(x-y_i)}{|x-y_i|^2} \right) \cdot \nu & \text{on} \ \partial B(x_i, \varrho(x_i))\end{array}\right.
  \]
  and thus by elliptic regularity estimates we have
  \begin{equation}\label{estimatesv}
    \|\nab v_i\|_{L^\infty(B(x_i, \frac34 \varrho(x_i)))} \le C \frac{|x_i-y_i|}{\varrho(x_i)^2}
    .\end{equation}
  Now, using \eqref{estimatesv}, \eqref{geta0} and \eqref{eq:defgeta} and decomposing $\psi_i$ as above, we find
  \begin{align*}
    \int_{B(x_i, \varrho(x_i)) } |\nab \psi_i|^2 & = 2\pi \int_{B(x_i, \varrho(x_i)) } \psi_i \left(\delta_{x_i}^{(\lambda)} - \delta_{y_i}^{(\lambda)}\right)
    \\ &= 2\pi \left( 2 (\g(\lambda)+\kappa) - \int \g*\delta_{x_i}^{(\lambda)} \delta_{y_i}^{(\lambda)} - \int \g*\delta_{y_i}^{(\lambda)} \delta_{x_i}^{(\lambda)} \right) + 2\pi\int v_i \left( \delta_{x_i}^{(\lambda)} - \delta_{y_i}^{(\lambda)}\right)
    \\
    &= 4\pi (\g(\lambda)+\kappa)- 4\pi \iint \g(x-y) \delta_{x_i}^{(\lambda)} (x)\delta_{y_i}^{(\lambda)} (y) + O\Bigr( \frac{|x_i-y_i|^2}{\varrho(x_i)^2}\Bigr)\\
    & = 4\pi (\g(\lambda)+\kappa)- 4\pi \g_\lambda(x_i-y_i)+ O\Bigr( \frac{|x_i-y_i|^2}{\varrho(x_i)^2}\Bigr).
  \end{align*}
  Inserting this into \eqref{hui} and \eqref{eq:rewrF}, we deduce that
  \begin{equation}\label{majof} \F_\lambda(X_N, Y_N) \le - \sum_{i=1}^N \g_\lambda (x_i-y_i) + O\Bigr(\sum_{i=1}^N \frac{|x_i-y_i|^2}{\varrho(x_i)^2}\Bigr).\end{equation}

  In all cases, since we have built the configurations so that $|x_i-y_i|\le \hal \varrho(x_i)$, the error term is bounded by $O(N)$.

  For $\beta=2$, set $t=\ell^{-1/5}$. Restrict the $x_i$'s to boundary distance at least $t$ and mutual distances at least $4t$. Successive placement gives volume at least $N^Ne^{-CNt}$. Restrict each dipole length to $2\lambda\le|x_i-y_i|\le t^2$. For $\lambda$ small, the disjoint balls $B(x_i,t)$ contain both smeared charges. The preceding Neumann construction gives an energy error at most $CNt^2$. The $N!$ assignments of negative labels are disjoint; hence
  \[
    \Z\ge N!N^Ne^{-CNt-CNt^2}
    \left[2\pi\left(\ell-\tfrac25\log\ell-\log2\right)\right]^N.
  \]
  Using $N!\ge(N/e)^N$ proves \eqref{eq:lbz2}. We now assume $\beta>2$.
  \smallskip

  \noindent \textbf{Step 2: bounding the free energy.}
  For $r\ge4\lambda$, define
  \[
    P_\beta(r):=\int_0^{r/2}2\pi s\exp\left(\beta\g_\lambda(s)-A_\beta\frac{s^2}{r^2}\right)\dd s,
  \]
  with $A_\beta>0$ fixed large enough in view of \eqref{majof}. The $N!$ assignments of negative labels are disjoint; hence
  \begin{equation}\label{minz3}
    \Z\ge N!\int_H\prod_{i=1}^N P_\beta(\varrho(x_i))\dd X_N.
  \end{equation}
  By scaling and restricting the radial integral to $[0,2\lambda]$, we have
  \[
    c_\beta\le\frac{P_\beta(r)}{\mc Z_\beta\lambda^{2-\beta}}\le1.
  \]
  Moreover, using $1-e^{-t}\le t$ for $t\ge0$, we obtain
  \begin{align*}
    0\le\mc Z_\beta\lambda^{2-\beta}-P_\beta(r)
    &\le C_\beta\left(r^{2-\beta}+\frac1{r^2}\int_0^{r/2}\frac{s^3}{(s\vee\lambda)^\beta}\dd s\right)\\
    &\le C_\beta\mc Z_\beta\lambda^{2-\beta}\varphi_\beta(r),
  \end{align*}
  where $\varphi_\beta$ is defined in \eqref{defvarphi}. The preceding two bounds imply
  \begin{equation}\label{align57}
    \log\frac{P_\beta(r)}{\mc Z_\beta\lambda^{2-\beta}}\ge-C_\beta\varphi_\beta(r).
  \end{equation}

  A boundary margin $b$ removes area at most $4\sqrt N b$, and exclusion discs of radius $r$ remove at most $(N-1)\pi r^2$. With $b=4\lambda$ and $r=8\lambda$, successive placement leaves area at least $N-C\sqrt N\lambda-CN\lambda^2\ge N/2$. Using $1-x\ge e^{-2x}$ for $0\le x\le1/2$ gives
  \begin{equation}\label{eq:volume-lower}
    |H|\ge N^N\exp(-C\sqrt N\lambda-CN\lambda^2).
  \end{equation}
  Let $\mathbb Q$ be normalized Lebesgue measure on $H$. By Jensen's inequality and symmetry, \eqref{minz3} and \eqref{align57} give
  \begin{equation}\label{eq:Jensen}
    \log\Z\ge\log N!+\log|H|+N\log(\mc Z_\beta\lambda^{2-\beta})
    -C_\beta N\Esp_{\mathbb Q}[\varphi_\beta(\varrho(x_1))].
  \end{equation}

  \smallskip
  \noindent
  \textbf{Step 3: conclusion.}
  Given all other points, each $x_i$ is uniform on a set of area at least $N/2$. Its conditional density is therefore at most $2/N$. The union bound for the other particles and the area of the boundary strip give
  \[
    \mathbb Q(\varrho(x_1)\le r)\le C\left(r^2+\frac r{\sqrt N}\right),\qquad 4\lambda\le r<1.
  \]
  Recall $\varepsilon_\lambda$ from subsection~\ref{subsec:notation}, and let
  \[
    \eta_\lambda:=\begin{cases}
      \lambda^{\beta-2},&2<\beta<3,\\
      \lambda\ell,&\beta=3,\\
      \lambda,&\beta>3.
    \end{cases}
  \]
  Since $\varphi_\beta$ is decreasing, integration of the preceding probability bound yields
  \begin{align*}
    \Esp_{\mathbb Q}[\varphi_\beta(\varrho(x_1))]
    &\le\varphi_\beta(1)+C\int_{4\lambda}^1(-\varphi_\beta'(r))\left(r^2+\frac r{\sqrt N}\right)\dd r\\
    &\le C_\beta\left(\varepsilon_\lambda+\frac{\eta_\lambda}{\sqrt N}\right).
  \end{align*}
  Since $\lambda^2\le\varepsilon_\lambda$ and $\lambda\le\eta_\lambda$, \eqref{eq:volume-lower} and \eqref{eq:Jensen} imply
  \begin{equation*}
    \log\Z\ge\log N!+N\log N+N\log(\mc Z_\beta\lambda^{2-\beta})
    -C_\beta N\varepsilon_\lambda-C_\beta\sqrt N\eta_\lambda.
  \end{equation*}
  Finally, we have
  \[
    C_\beta\sqrt N\eta_\lambda\le N\varepsilon_\lambda+C_\beta'\frac{\eta_\lambda^2}{\varepsilon_\lambda},\qquad
    \frac{\eta_\lambda^2}{\varepsilon_\lambda}=o_\lambda(1).
  \]
  Using $\log N!\ge N\log N-N+1$ proves \eqref{lbz} for $\lambda$ small enough.
\end{proof}

\medskip

\begin{proof}[Proof of Theorem \ref{theorem:th1}]
  We take $\lambda_\beta<e^{-1}$ below the thresholds of Lemma~\ref{lemma:part nn} and Propositions~\ref{prop:upper bound}, \ref{prop:critical upper} and~\ref{prop:lower bound}. These estimates give $0<\Z<\infty$. Since $\lambda\ell\le1$, we have $\lambda^2\ell\le\lambda$, and comparing exponents with \eqref{eq:defgammal} shows that the errors in \eqref{lbz} and \eqref{eq:lbz2} are $O_\beta(N\omega_\lambda)$. For $\beta>2$, the upper bound is \eqref{ubz}. For $\beta=2$, Proposition~\ref{prop:critical upper} with $R=1$ bounds $\log\Z$ from above, since $D_1\ge0$ and $N-\frac12|I^{\dip}|\ge0$. Combining the upper and lower bounds proves \eqref{eq:devlogz}.

  We now prove \eqref{eq:momexp}. Let $\tilde\F_\lambda$ and $\K_{N,\beta}^\lambda$ be as in Lemma~\ref{lemma:part nn}. By the definition \eqref{defP} of the Gibbs measure, we have
  \begin{equation*}
    \Esp_{\P}[\exp(\beta( \F_\lambda-\tilde{\F}_\lambda))]=\frac{\K_{N,\beta}^\lambda}{\Z}.
  \end{equation*}
  Since $\ell\ge1$, comparing exponents with \eqref{eq:defgammal} gives $\bar\omega_\lambda\le\omega_\lambda$, where $\bar\omega_\lambda$ is defined in \eqref{def:o nn}. Subtracting \eqref{eq:devlogz} from the expansion of Lemma~\ref{lemma:part nn} therefore gives
  \begin{equation*}
    \log \frac{\K_{N,\beta}^\lambda}{\Z}=O_\beta(N\omega_\lambda)+O_\beta(\log(N+1)),
  \end{equation*}
  which proves \eqref{eq:momexp}.

  The events in item (3) are measurable, since they are finite unions of sets defined by distance comparisons. For each such event $A$,
  \[
    \P(A)=\frac1{\Z}\int_{\Lambda^{2N}\cap A}e^{-\beta\F_\lambda}\dd Z_{2N}.
  \]
  For $\beta>2$, subtracting the lower bound from \eqref{eq:neutral bound} proves \eqref{eq:b1}. For $\beta=2$, on $\mc A(p_0)$ we have $N-\frac12|I^{\dip}|=N-p_0\ge(N-p_0)^2/N$, so Proposition~\ref{prop:critical upper} with $R=1$ and \eqref{eq:lbz2} prove \eqref{eq:b1}.

  For $\beta>2$ and $s\ge2$, the identity $\g_1(s)=-\log s$ gives
  \[
    \mu_\beta([s,\infty))=C_\mu s^{2-\beta},
    \qquad C_\mu=\frac{2\pi}{(\beta-2)\mc Z_\beta}.
  \]
  Set $\alpha=(\lambda/R)^{\beta-2}$. Taking the threshold in \eqref{eq:b2} at least $\max\{C_\beta C_\mu,C_\mu^2,2\}$ ensures the hypothesis of \eqref{eq:u bound} and
  \[
    \log\frac{x_0}{C_\mu\alpha}\ge\frac12\log\frac{x_0}{\alpha}.
  \]
  Subtracting the lower bound proves \eqref{eq:b2}, since the errors in \eqref{eq:u bound} and \eqref{lbz} are both $O_\beta(N\omega_\lambda)$. For $\beta=2$, subtracting \eqref{eq:lbz2} from \eqref{eq:critical upper} and applying Markov's inequality gives \eqref{eq:critical lengths}, since $N-\frac12|I^{\dip}|\ge0$.
\end{proof}

\section{Energetic control on linear statistics}\label{sec5}
In this section, we use our ball-growth method introduced in Section \ref{sec2} to derive an energetic control on the fluctuations of linear statistics, similar to \cite[Prop. 2.5]{LS2} for the one-component plasma. We prove Theorem~\ref{prop:linear}, giving a fluctuation bound at scale $\sqrt N\,\omega_\lambda^{1/4}$ for Lipschitz test-functions. Let us emphasize that linear statistics are in fact expected to fluctuate much less, specifically at a rate $o_N(\sqrt{N})$ for small but fixed $\lambda$. Proving such a rigidity statement would require more involved techniques.

\begin{proof}[Proof of Theorem \ref{prop:linear}]
  \textbf{Step 1: the electric energy bounds the fluctuations.} As in the one-component case \cite{ss1d,LS2} or in \cite{LSZ}, the fluctuations are well bounded by the electric energy $
  \int |\nab h_{\vec{\alpha}}|^2$, where $h_{\vec{\alpha}}$ is as in \eqref{defhalpha}, for every $\vec\alpha\in(0,\infty)^{2N}$. We recall the elementary argument.

  Let $\xi$ be a compactly supported Lipschitz test-function from $\R^2 $ to $\R$.
  Let $\tilde\xi$ be $\xi-\xi(0)$ multiplied by a cutoff equal to $1$ on $\Lambda$, supported in $[-\sqrt N,2\sqrt N]^2$ and with gradient bounded by $C/\sqrt N$. Then $\Fluct_N(\tilde\xi)=\Fluct_N(\xi)$ and
  \[
    \|\nab\tilde\xi\|_{L^\infty}+N^{-1/2}\|\nab\tilde\xi\|_{L^2(\R^2)}\le C\|\nab\xi\|_{L^\infty}.
  \]
  Neutrality gives $\nab h_{\vec\alpha}\in L^2$. We apply Green's formula to mollifications of $\tilde\xi$ and pass to the limit in $W^{1,2}$. The Cauchy--Schwarz inequality gives
  \begin{equation*}
    \left|\int_{\R^2}\tilde\xi \, \dd\left( 2\pi \sum_{i=1}^{2N} d_i \delta_{z_i}^{(\alpha_i)}\right)\right|=
    \left|\int_{\R^2} \tilde\xi \Delta h_{\vec{\alpha}}\right|\le \|\nab\tilde\xi\|_{L^2(\R^2)}\left( \int_{\R^2}|\nab h_{\vec{\alpha}}|^2\right)^\hal.\end{equation*}
  On the other hand, by definition of the smeared charges, we may write
  \begin{equation*}
    \left|\int_{\R^2} \tilde\xi \, \dd\left( 2\pi \sum_{i=1}^{2N} d_i (\delta_{z_i}- \delta_{z_i}^{(\alpha_i)})\right)
    \right|\le C \|\nab\tilde\xi\|_{L^\infty} \sum_{i=1}^{2N} \alpha_i.\end{equation*}
  Combining the two relations, we deduce that
  \begin{equation*}
    \left|\Fluct_N(\xi)\right|\le C \|\nab \xi\|_{L^\infty} \left(\sqrt N\left( \int_{\R^2}|\nab h_{\vec{\alpha}}|^2\right)^\hal+ \sum_{i=1}^{2N} \alpha_i\right),
  \end{equation*}
  hence
  \begin{equation*}
    \left|\Fluct_N(\xi)\right|^2 \le C \|\nab \xi\|_{L^\infty}^2 \left(N \int_{\R^2}|\nab h_{\vec{\alpha}}|^2+ \left(\sum_{i=1}^{2N} \alpha_i\right)^2\right).\end{equation*}
  If $\xi(x)= \xi_0( x/\sqrt{N})$, then $\|\nab \xi\|_{L^\infty} \le \frac{\|\nab \xi_0\|_{L^\infty}}{\sqrt N}$. It follows that
  \begin{equation*}
    \left|\Fluct_N(\xi)\right|^2 \le C \|\nab \xi_0\|_{L^\infty}^2 \left(\int_{\R^2}|\nab h_{\vec{\alpha}}|^2+\frac{1}{N}\left(\sum_{i=1}^{2N} \alpha_i\right)^2\right).\end{equation*}
  \smallskip

  \noindent
  \textbf{Step 2: upper bound for the electric energy by the ball-growth method.}
  The proof consists in repeating the steps of the proof of Proposition \ref{prop:mino2}. Recall the radii defined in \eqref{deftau},
  \begin{equation*}
    \tau_i= \begin{cases}
      \rr_2(z_{i})\wedge \rr_2(z_{\phi_1(i) }) & \text{if $i\in \mathcal C_k$}\\
      \rr_2(z_i) & \text{otherwise}.\end{cases}
  \end{equation*}
  We let

  \begin{equation*}
    \alpha_i:= (\tau_i \wedge \theta) \vee \rr_1(z_i)\end{equation*}
  for a parameter $2\lambda\le\theta\le1$ to be chosen later.

  On a $2$-cycle, \eqref{rrj} and $\tau_i\ge\rr_1(z_i)$ give $\tau_i\ge\rr_2(z_i)/2$. Hence
  \[
    \alpha_i\ge\frac12\bigl((\rr_2(z_i)\wedge\theta)\vee\rr_1(z_i)\bigr).
  \]
  We also have $\lambda\le\rr_1(z_i)\le\alpha_i\le\tau_i$ for every $i$, while $\alpha_i=\alpha_{\phi_1(i)}$ on a $2$-cycle. Thus $|z_i-z_{\phi_1(i)}|/\alpha_i\le4$, as required for the kernel estimate.

  We then bound $\int |\nabla h_{\lambda}|^2$ from below as in the proof of Proposition \ref{prop:mino2} and find
  \begin{align*}
    & \int_{\R^2} |\nab h_{\lambda}|^2
    \ge \int_{\R^2} |\nab h_{\vec{\alpha}}|^2 \\
    &
    + 2\pi\sum_{k=1}^K \Bigg( \sum_{i\in \mathcal C_k} \Bigl[d_i d_{\phi_1(i)} \g_\lambda(z_i-z_{\phi_1(i)} ) -d_i \left( d_i+ d_{\phi_1(i)} \right)(\g(\alpha_i)+\kappa) \\
    & + O\Bigr( \Bigr( \frac{\rr_1(z_i)}{(\rr_2(z_i)\wedge \theta) \vee \rr_1(z_i) }\Bigr)^2\Bigr)\Bigr]
    -\sum_{i\in \mathcal{I}_k \backslash \mathcal C_k} \bigl(\g(\alpha_i)+\kappa\bigr) \Bigg)+4N\pi \g_\lambda(0).
  \end{align*}
  For $\beta=2$, we use this inequality directly. Let $S=1+\log\ell+\log(1/\theta)$. We use the retained points and the radii $\rho$ of Step~1 of the proof of Proposition~\ref{prop:critical upper}. For an isolated pair with relative position $u$, we write $a=(|u|/4)\vee\lambda$. Since $\rr_1(z)=a$ and $\rho\le3\rr_2(z)$ for both endpoints $z$, the error of an isolated dipole in the preceding inequality is at most $C((a/\rho)^2+\min\{a/\theta,1\}^2)$. Similarly, the correction of a same-sign isolated pair is at least $2\log(((\rho\wedge\theta)\vee a)/a)-C$.

  We integrate the variables $u$ of the isolated pairs as in Step~2 of that proof. Uniformly in $\rho\ge\lambda$, we have
  \begin{align*}
    \int_{|u|\le4\rho}e^{2\g_\lambda(u)-|u|^2/\ell}
    e^{C((a/\rho)^2+\min\{a/\theta,1\}^2)}\dd u&\le J_\lambda(1/\ell)+C'S\,,\\
    \int_{|u|\le4\rho}e^{2\g_\lambda(u)-|u|^2/\ell}
    \Bigl(\frac a{(\rho\wedge\theta)\vee a}\Bigr)^4\dd u&\le C'S\,.
  \end{align*}
  Indeed, since $\min\{a/\theta,1\}\le1$, the exponential factor in the first integral is at most
  \[
    1+e^C\bigl(e^{C(a/\rho)^2}-1\bigr)+\bigl(e^{C\min\{a/\theta,1\}^2}-1\bigr)\,.
  \]
  By \eqref{eq:critical-correction}, the second term adds $O(1)$ to $J_\lambda(1/\ell)$. On $|u|<4\theta$, we have $\min\{a/\theta,1\}=a/\theta$, and the third term adds $O(1)$, by \eqref{eq:critical-correction} with $\theta$ in place of $\rho$. On $|u|\ge4\theta$, the third term is at most $e^C$, and the integral of $e^{2\g_\lambda(u)-|u|^2/\ell}$ over this region is at most $J_\lambda(1/\ell;(4\theta)\wedge1)$. This is $O(S)$ by \eqref{eq:critical-BR}, since $2\lambda\le(4\theta)\wedge1\le1$. For the second integral, we use
  \[
    \Bigl(\frac a{(\rho\wedge\theta)\vee a}\Bigr)^4\le\Bigl(\frac a\rho\Bigr)^4+\min\Bigl\{\frac a\theta,1\Bigr\}^4
  \]
  and bound the two terms in the same way.

  Fix an unmarked nearest-neighbor graph $\gamma$ with $K$ components, $m$ isolated pairs and $p_0$ isolated dipoles. For $0<c\le1/(2\pi)$, the electric term is absorbed in $-2\F_\lambda$. The two preceding bounds, together with the Gaussian insertion and the cycle and tree integrals of Step~2 of the proof of Proposition~\ref{prop:critical upper}, give
  \begin{multline*}
    \int_{\Lambda^{2N}\cap\{\gamma_{2N}=\gamma\}}e^{-2\F_\lambda+c\int_{\R^2}|\nab h_{\vec\alpha}|^2}\dd Z_{2N}\\
    \le e^{C(N-p_0)+CN/\ell}N^K\bigl(J_\lambda(1/\ell)+CS\bigr)^{p_0}(CS)^{m-p_0}J_\lambda(1/\ell)^{K-m}\Bigl(\frac{CN}{N-K}\Bigr)^{N-K}\,,
  \end{multline*}
  where the last factor is $1$ when $K=N$. We conclude as in Step~3 of that proof, using $\log(J_\lambda(1/\ell)+CS)\le\log J_\lambda(1/\ell)+CS/\ell$. With $\delta=(N-K)/N$ and $\zeta=(m-p_0)/N$, this gives
  \begin{multline*}
    \log\int_{\Lambda^{2N}}e^{-2\F_\lambda+c\int_{\R^2}|\nab h_{\vec\alpha}|^2}\dd Z_{2N}
    \le2N\log N+N\bigl(\log(2\pi\ell)-1\bigr)
    +\frac{CNS}{\ell}+C\log(N+1)\\
    +N\sup_{0\le\delta,\zeta\le1}\left\{
    \delta\left[-\log\ell+5\log(1/\delta)+C\right]
    +\zeta\left[-\log\ell+\log S+\log(1/\zeta)+C\right]\right\}\,.
  \end{multline*}
  By \eqref{eq:sup-inequality}, the supremum is at most $C\ell^{-1/5}+CS/\ell$. We extend \eqref{defvarphi} to $\beta=2$ by setting $\varphi_2(x):=(1+|\log x|)/\ell$ for $x\ge\lambda$. Since $\ell^{-1/5}=\omega_\lambda$ and $S/\ell\le\varphi_2(\theta)+C\omega_\lambda$, dividing by the partition function and using \eqref{eq:lbz2} yields
  \begin{multline}\label{eq:critical electric}
    \log\Esp_{\P}\left[e^{c\int_{\R^2}|\nabla h_{\vec\alpha}|^2}\right]
    \le CN\left(\ell^{-1/5}+\frac S\ell\right)+C\log(N+1)\\
    \le CN\bigl(\omega_\lambda+\varphi_2(\theta)\bigr)+C\log(N+1).
  \end{multline}
  For the remainder of Steps 2--4, assume $\beta>2$.

  Since $\rr_1(z_i)\le\rr_2(z_i)$, we have
  \[
    \left(\frac{\rr_1(z_i)}{(\rr_2(z_i)\wedge\theta)\vee\rr_1(z_i)}\right)^2
    \le\left(\frac{\rr_1(z_i)}{\rr_2(z_i)}\right)^2
    +\left(\left(\frac{\rr_1(z_i)}\theta\right)^2\wedge1\right).
  \]
  Arguing as in the proof of Corollary~\ref{coromino}, we obtain
  \begin{multline*}
    \F_\lambda(X_N,Y_N)\ge\frac1{4\pi}\int_{\R^2}|\nab h_{\vec\alpha}|^2
    -\hal\sum_{i=1}^{2N}\g_\lambda(z_i-z_{\phi_1(i)})\\
    +\sum_{\substack{i\in I^\pair\setminus I^\dip\\\phi_2(i),\phi_2(\phi_1(i))\in I^\pair}}
    \left(\log\frac{\alpha_i}{\rr_1(z_i)}-C\right)\\
    -C\sum_{\substack{i\in I^\dip\\\phi_2(i),\phi_2(\phi_1(i))\in I^\pair}}
    \left[\left(\frac{\rr_1(z_i)}{\rr_2(z_i)}\right)^2
    +\left(\left(\frac{\rr_1(z_i)}\theta\right)^2\wedge1\right)\right]-C(N-K).
  \end{multline*}
  By Lemma~\ref{lem:in-degree}, at most $12|I^\pair\setminus I^\dip|$ indices in the neutral sum have $\phi_2(i)\notin I^\dip$. Their summands are at most $2$. Since $|I^\pair\setminus I^\dip|\le2N-|I^\dip|$ and $|I^\dip|/2\le K$, we may write
  \begin{equation*}
    \frac1{4\pi} \int_{\R^2} |\nabla h_{\vec{\alpha}}|^2\le A_1+C (A_2+A_3),
  \end{equation*}
  with
  \begin{align*}
    A_1&:= \F_\lambda(Z_{2N})+ \hal \sum_{ i=1}^{2N} \g_\lambda(z_i-z_{\phi_1(i)} )- \sum_{\substack{i\in I^{\pair}\backslash I^{\dip}\\ \phi_2(i) \in I^\pair, \phi_2(\phi_1(i)) \in I^\pair}}\left( \log
    \frac{\alpha_i}{\rr_1(z_i)}-C\right),
    \\
    A_2&:=\sum_{i\in I^{\dip},\phi_2(i)\in I^{\dip},\phi_2(\phi_1(i)) \in I^\pair } \Bigr(\frac{\rr_1(z_i)}{\rr_2(z_i) } \Bigr)^2 + \left(N-\frac12|I^\dip|\right),\\
    A_3&:=\sum_{i\in I^{\dip},\phi_2(i)\in I^{\dip} ,\phi_2(\phi_1(i)) \in I^\pair} \Bigl(\Bigr(\frac{\rr_1(z_i)}{{\theta} } \Bigr)^2 \wedge 1\Bigr).
  \end{align*}
  \textbf{Step 3: bounding an exponential moment of $A_1$ and $A_2$.}
  In the Gibbs integrand $e^{\beta(A_1-\F_\lambda)}$, the $\F_\lambda$ term of $A_1$ cancels exactly. We repeat the proof of \eqref{eq:alt}, keeping $\theta$ in the same-sign term. On a marked graph, let $|I^\pair|=2m$ and $q=p-\widehat p_0$. For $m\ge2$, we use the representatives of Proposition~\ref{prop:isolated}, with $\rho$ as in its proof and $a=(|u|/4)\vee\lambda$. The normalized same-sign radial factor is bounded by
  \begin{equation*}
    \frac{C_\beta}{\mc Z_\beta\lambda^{2-\beta}}
    \int_{|u|\le4\rho}e^{\beta\g_\lambda(u)}
    \left(\frac{a}{(\rho\wedge\theta)\vee a}\right)^{2\beta}\dd u
    \le C_\beta\left[\left(\frac\lambda\rho\right)^{\beta-2}
    +\left(\frac\lambda\theta\right)^{\beta-2}\right].
  \end{equation*}
  Indeed, the factor in the integrand is at most $(a/\rho)^{2\beta}+\min\{a/\theta,1\}^{2\beta}$. The first term is bounded by \eqref{align2}, and we integrate the second over $\R^2$.

  We expand the added constant over $s$ of the $q$ marked same-sign pairs. The entropy comparison preceding \eqref{eq:F(1)} bounds the entropy factor by $(m+1)\binom mq e^{C_0q}$. For $0\le s\le q$, we use
  \[
    \binom mq\binom qs=\binom m{q-s}\binom{m-q+s}s,
    \qquad \binom{m-q+s}s\le\binom ms.
  \]
  We use \eqref{eq:representative-bound} for the remaining integrals, assigning the $s$ pairs the factors $1+C_\beta\varphi_\beta(\rho)\ge1$. Summing the constant terms costs at most
  \[
    \left(1+C_\beta e^{C_0}(\lambda/\theta)^{\beta-2}\right)^m
    \le\exp\bigl(C_\beta N\varphi_\beta(\theta)\bigr).
  \]
  For $p<N/2$, the same crude bound applies, including $m\le1$. After reducing $\lambda_\beta$, the graph sum and Proposition~\ref{prop:lower bound} give
  \begin{equation}\label{eq:A1}
    \log\Esp_{\P}[e^{\beta A_1}]
    \le C_\beta N\bigl(\omega_\lambda+\varphi_\beta(\theta)\bigr)+C_\beta\log(N+1).
  \end{equation}
  To bound $A_2$, we repeat Steps 2 and 3 of Proposition~\ref{prop:upper bound} with a larger coefficient in the neutral quadratic term. On each marked graph,
  \[
    N-\frac12|I^\dip|\le N-\widehat p_0\le(M_0+3)(N-K)+(p-\widehat p_0).
  \]
  We absorb the first term into $C_\beta(N-K)$ and the second into the arbitrary $C_0(p-\widehat p_0)$ allowance in \eqref{eq:representative-bound}. Dividing by the partition function gives
  \begin{equation}\label{eq:A2}
    \log\Esp_{\P}[e^{A_2}]\le C_\beta N\omega_\lambda+C_\beta\log(N+1).
  \end{equation}
  \textbf{Step 4: bounding an exponential moment of $A_3$.} Fix $t>0$ and assume $2\lambda\le\theta\le1$. We use the marked graph decomposition of Proposition~\ref{prop:upper bound}:
  \begin{multline*}
    \log\int_{\Lambda^{2N}} e^{tA_3-\beta\F_\lambda}\dd Z_{2N}
    \le\max_{K,p,\widehat p_0}\Bigg(\log|D_{2N,K,p,\widehat p_0}|\\
    +\max_{\gamma\in D_{2N,K,p,\widehat p_0}}
    \log\int_{\Omega_{\gamma}}e^{tA_3-\beta\F_\lambda}\dd Z_{2N}\Bigg)+C\log(N+1).
  \end{multline*}
  Fix $\gamma\in D_{2N,K,p,\widehat p_0}$, let $I=I^{\pair}$ with $|I|=2m$, and let $\gamma^1$ and $\gamma^2$ be the restrictions of $\gamma$ to $I$ and $I^c$, as in \eqref{gred3}. On this domain,
  \[
    A_3\le\sum_{i\in I}\left(\left(\frac{\rr_1(z_i)}{\theta}\right)^2\wedge1\right),
  \]
  where the distances on the right are computed within $I$. Hence
  \begin{multline*}
    \log\int_{\Omega_{\gamma}}e^{tA_3-\beta\F_\lambda}\dd Z_{2N}\\
    \le\log\int_{\Lambda^{2m}\cap\{\gamma_{2m}=\gamma^1\}}
    \exp\left(t\sum_{i\in I}\left(\left(\frac{\rr_1(z_i)}{\theta}\right)^2\wedge1\right)-\beta F^{(1)}\right)\prod_{i\in I}\dd z_i\\
    +\log\int_{\Lambda^{2N-2m}\cap\{\gamma_{2N-2m}=\gamma^2\}}e^{-\beta F^{(2)}}\prod_{i\in I^c}\dd z_i+C_\beta(N-K).
  \end{multline*}
  For a pair in $I\setminus I'$, the extra radial contribution is controlled by
  \[
    \frac{1}{\mc Z_\beta\lambda^{2-\beta}}
    \int_{\R^2}e^{\beta\g_\lambda(u)}
    \left(\left(\frac{(|u|/4)\vee\lambda}{\theta}\right)^2\wedge1\right)\dd u
    \le C_\beta\varphi_\beta(\theta).
  \]
  This follows by splitting at $2\lambda$ and $4\theta$. For $0\le s\le1$,
  \[
    e^{2ts}\le1+(e^{2t}-1)s.
  \]
  The unperturbed quadratic factor is bounded on $|u|\le4\rho$. Hence the tilt adds at most $C_{\beta,t}\varphi_\beta(\theta)$ to the normalized radial factor, bounded by
  \[
    (1+C_{\beta,t}\varphi_\beta(\theta))(1+C_\beta\varphi_\beta(\rho)).
  \]
  For each marked same-sign pair, the tilt is at most $e^{2t}$, absorbed by increasing $C_0$. Thus \eqref{eq:representative-bound} and the entropy comparison above \eqref{eq:F(1)} give, for $p\ge N/2$,
  \begin{multline*}
    \log\int_{\Lambda^{2m}\cap\{\gamma_{2m}=\gamma^1\}}
    \exp\left(t\sum_{i\in I}\left(\left(\frac{\rr_1(z_i)}{\theta}\right)^2\wedge1\right)-\beta F^{(1)}\right)\prod_{i\in I}\dd z_i\\
    +p\log p-\widehat p_0\log \widehat p_0-(p-\widehat p_0)\log(p-\widehat p_0)+C_0(p-\widehat p_0)\\
    \le m\bigl(\log N+(2-\beta)\log\lambda+\log\mc Z_\beta\bigr)
    +O_{\beta,t}\bigl(N(\varphi_\beta(\theta)+\omega_\lambda)\bigr)
    -c_{\beta,t}N\left(\frac{p-\widehat p_0}{N}\right)^2+C_{\beta,t}\log(N+1).
  \end{multline*}
  For $p<N/2$, the tilt is at most $e^{2tm}$, so the crude bound suffices as in Step 3 of the upper-bound proof. Taking $0<\lambda<\lambda_{\beta,t}$ for a threshold depending only on $\beta,t$, combining with \eqref{eq:F(2)} and summing over the marked graphs yields
  \begin{multline*}
    \log\int_{\Lambda^{2N}} e^{tA_3-\beta\F_\lambda}\dd Z_{2N}
    \le2N\log N+N\bigl((2-\beta)\log\lambda+\log\mc Z_\beta-1\bigr)\\
    +O_{\beta,t}\bigl(N(\varphi_\beta(\theta)+\omega_\lambda)\bigr)+C_{\beta,t}\log(N+1).
  \end{multline*}
  Subtracting the lower bound of Proposition~\ref{prop:lower bound} gives
  \begin{equation}\label{eq:finalA3}
    \log\Esp_{\P}[e^{tA_3}]
    \le C_{\beta,t}N\bigl(\varphi_\beta(\theta)+\omega_\lambda\bigr)+C_{\beta,t}\log(N+1),
  \end{equation}
  with $\varphi_\beta$ as in \eqref{defvarphi}.

  \smallskip
  \noindent
  \textbf{Step 5: control on the radii.} Set
  \[
    R_0=\begin{cases}
      \omega_\lambda^{1/4}&\text{if $\beta=2$},\\
      \lambda\omega_\lambda^{-1/(2(\beta-2))}&\text{if $\beta>2$},
    \end{cases}
    \qquad
    \mc A_M=\{N-\frac12|I^{\dip}|+D_{R_0}\le MN\sqrt{\omega_\lambda}\}.
  \]
  For $\lambda$ small, we have $2\lambda\le R_0\le\omega_\lambda^{1/4}$. Indeed, for $\beta>2$, the upper bound amounts to $\lambda\le\omega_\lambda^{1/4+1/(2(\beta-2))}$, and the right side is at least $\lambda^{1/40}$ by \eqref{eq:defgammal}, since $\ell\ge1$. For $\beta>2$, we also have $\alpha=(\lambda/R_0)^{\beta-2}=\omega_\lambda^{1/2}$, and we use \eqref{eq:b1} and \eqref{eq:b2}. For $\beta=2$, we use \eqref{eq:b1} and \eqref{eq:critical lengths}, since $\log(1/R_0)/\ell=\log\ell/(20\ell)=o(\omega_\lambda)$. A union bound and summation over the counts give, for $M$ large,
  \begin{equation}\label{eq:PAM}
    \P(\mc A_M^c)\le C_\beta(N+1)^{C_\beta}e^{-c_\beta MN\omega_\lambda}.
  \end{equation}
  On $\mc A_M$, at most $2MN\sqrt{\omega_\lambda}$ particles lie outside isolated dipoles of length at most $R_0$. The balls $B(z_i,|z_i-z_{\phi_1(i)}|/2)$ are disjoint in $[-\sqrt N,2\sqrt N]^2$, so $\sum_i\rr_1(z_i)^2\le CN+2N\lambda^2\le C'N$. Thus
  \[
    \sum_{i=1}^{2N}\rr_1(z_i)
    \le2NR_0+\left(2MN\sqrt{\omega_\lambda}\sum_{i=1}^{2N}\rr_1(z_i)^2\right)^{1/2}
    \le C(1+\sqrt M)N\omega_\lambda^{1/4}.
  \]
  It follows that
  \begin{equation}\label{eq:sumrr1}
    \left(\sum_{i=1}^{2N}\rr_1(z_i)\right)^2\le C(1+M)N^2\sqrt{\omega_\lambda}.
  \end{equation}

  \smallskip
  \noindent
  \textbf{Step 6: conclusion.} For $\beta>2$, fix $t=1$ in \eqref{eq:finalA3} and choose $c_\beta>0$ so that $12\pi c_\beta\le\beta$ and $12\pi c_\beta C\le1$, with $C$ from Step 2. H\"older's inequality with three factors and
  \[
    \Esp_{\P}[e^{sX}]\le\bigl(\Esp_{\P}[e^{rX}]\bigr)^{s/r},\qquad 0\le s\le r,\quad r>0,
  \]
  give the following bound from \eqref{eq:A1}, \eqref{eq:A2} and \eqref{eq:finalA3}. For $\beta=2$, it follows from \eqref{eq:critical electric}.
  \begin{equation*}
    \log\Esp_{\P}\left[\exp\left(c_\beta \int_{\R^2} |\nabla h_{\vec{\alpha}}|^2 \right)\right]\leq C_\beta N(\varphi_\beta(\theta)+\omega_\lambda)+C_\beta\log(N+1).
  \end{equation*}
  Therefore, there exists an event $\mc{B}_M$ such that for all $M$ large enough,
  \begin{equation*}
    \int_{\R^2} |\nabla h_{\vec{\alpha}}|^2\leq C_\beta M N (\varphi_\beta(\theta)+\omega_\lambda)\quad\text{on }\mc B_M,
  \end{equation*}
  \begin{equation*}
    \P(\mc{B}_M^c)\leq C_\beta (N+1)^{C_\beta}e^{-c_\beta M N\omega_\lambda}.
  \end{equation*}
  Combining the last displays and \eqref{eq:sumrr1} and in view of the definition of $\alpha_i$, we deduce that for $M$ large enough, we have
  \begin{equation}\label{eq:f}
    |\Fluct_N(\xi)|^2\leq C_\beta M\Vert \nabla \xi_0 \Vert_{L^{\infty}}^2N({\theta}^2+\omega_\lambda^{\nicefrac{1}{2}}+\varphi_\beta(\theta))\quad \text{on the event $\mc{A}_M\cap \mc{B}_M$.}
  \end{equation}
  Optimizing over $\theta$, we may then choose $\theta$ as follows:
  \begin{equation*}
    \theta=\begin{cases}
      \ell^{-1/2} &\text{if $\beta=2$}\\
      \lambda^{\frac{(\beta-2)}{\beta}}& \text{if $\beta\in (2,4)$}\\
      \lambda^{1/2}\ell^{1/4}& \text{if $\beta=4$}\\
      \lambda^{1/2}& \text{if $\beta\in (4,\infty)$}.
    \end{cases}
  \end{equation*}
  We reduce $\lambda_\beta$ so that all previous estimates and $2\lambda\le\theta,R_0\le1$ hold. Noting that for $\theta$ defined as above,
  $\varphi_\beta(\theta)+\theta^2=O_\beta(\omega_\lambda)$, we combine \eqref{eq:f}, \eqref{eq:PAM} and the bound on $\P(\mc B_M^c)$. The elementary bound
  \[
    \min\{1,C_\beta(N+1)^{C_\beta}e^{-c_\beta MN\omega_\lambda}\}
    \le C_\beta'e^{-c_\beta'M\sqrt N\,\omega_\lambda}
  \]
  absorbs the polynomial prefactor. Replacing the event parameter $M$ by $M/C_\beta$ concludes the proof.
\end{proof}
\bibliographystyle{alpha}
\bibliography{bib.bib}
\end{document}